\documentclass[11pt]{article}
\usepackage{graphicx}
\usepackage[breaklinks=true, linktocpage=true]{hyperref}
\usepackage{amssymb}
\usepackage{tikz}
\usepackage{commath}
\usepackage{braket}
\usepackage{wasysym}
\usepackage{booktabs}
\usepackage{amsmath}
\usepackage{amsthm}
\usepackage{lipsum}
\usepackage{url}
\usepackage[boxed]{algorithm2e}
\usepackage{algpseudocode}
\usepackage[nameinlink]{cleveref}
\usepackage{bbm}
\usepackage{mathtools}
\usepackage{soul}

\usepackage{thmtools,thm-restate}
\usepackage{quantikz}

\usepackage[shortlabels]{enumitem}

\newif\ifcommentsfootnote
\commentsfootnotetrue

\newif\ifcommentsinline
\commentsinlinetrue

\hypersetup{
    colorlinks = true,
    linkcolor={purple},
    urlcolor={purple},
    citecolor={blue!80!black}
}

\definecolor{cb-salmon-pink}{RGB}{255, 182, 119}
\ifcommentsinline
    \newcommand{\nat}[1]{\textcolor{blue}{\textbf{Nat: }}\textcolor{gray}{#1}}
    \newcommand{\hnote}[1]{\textcolor{red}{[\textbf{HY: #1}]}}
\else
    \newcommand{\nat}[1]{}
    \newcommand{\hnote}[1]{}
\fi

\ifcommentsfootnote
    \newcommand{\natf}[1]{\footnote{{\bf \color{violet}Nat:} {#1}}}
\else
    \newcommand{\natf}[1]{}
\fi

\newcommand{\blue}[1]{{\color{blue}{#1}}}
\newcommand{\eps}{\varepsilon}
\renewcommand{\epsilon}{\eps}
\newcommand{\wt}[1]{\widetilde{#1}}

\newcommand{\sse}{\subseteq}

\newcommand{\zo}{\{0,1\}}

 \newcommand{\ZZ}{\mathbb{Z}}
\newcommand{\NN}{\mathbb{N}}
\newcommand{\CC}{\mathbb{C}}
\newcommand{\FF}{\mathbb{F}}

\theoremstyle{definition}
\newtheorem{definition}{Definition}[section]
\newtheorem{theorem}[definition]{Theorem}
\newtheorem{lemma}[definition]{Lemma}
\newtheorem{claim}[definition]{Claim}

\newtheorem{proposition}[definition]{Proposition}
\newtheorem{fact}[definition]{Fact}
\newtheorem{corollary}[definition]{Corollary}
\newtheorem{conjecture}{Conjecture}
\newtheorem{openproblem}{Open Problem}

\newtheorem{question}[definition]{Question}

\crefname{claim}{claim}{claims}

\newcommand{\calL}{\mathcal{L}}

\newcommand{\calP}{\mathcal{P}}

\newcommand{\calS}{\mathcal{S}}
\newcommand{\calT}{\mathcal{T}}
\newcommand{\calU}{\mathcal{U}}

\newcommand{\calD}{\mathcal{D}}
\newcommand{\calH}{\mathcal{H}}
\newcommand{\calC}{\mathcal{C}}

\renewcommand{\abs}[1]{{\left| #1 \right|}}

\newcommand{\pbra}[1]{{\left( #1 \right)}}
\newcommand{\cbra}[1]{{\left\{ #1 \right\}}}
\newcommand{\sbra}[1]{{\left[ #1 \right]}}

\newcommand{\poly}{\mathsf{poly}}

\newcommand{\QNC}{\mathsf{QNC}}
\newcommand{\QNCZ}{\QNC^0}

\newcommand{\NC}{\mathsf{NC}}

\newcommand{\NCo}{\NC^1}
\newcommand{\QAC}{\mathsf{QAC}}
\newcommand{\QACZ}{\QAC^0}
\newcommand{\QACZF}{\QAC^0_{f}}
\newcommand{\TCZ}[1][]{\mathsf{TC}^0_{#1}}
\newcommand{\ACZ}{\mathsf{AC}^0}

\NewDocumentCommand{\QNCZF}{o}{\IfNoValueTF{#1} {\QNCZ_f} {\QNCZ_{f,#1}} }

\newcommand{\CNOT}{\textsf{CNOT}}

\newcommand{\ot}{\otimes}
 
\DeclareMathOperator{\tr}{Tr}
\newcommand{\Tr}{\tr}
\renewcommand{\complement}{c}

\newcommand{\Hilbertn}[1][n]{\calH_{{#1}}}
\newcommand{\linops}[1]{\calL\pbra{#1}}
\newcommand{\linopsn}[1][n]{\linops{\Hilbertn[#1]}}
\newcommand{\states}[1][n]{\calD(\Hilbertn[{#1}])}
\newcommand{\unitaries}[1][n]{\calU_{#1}}
\newcommand{\paulis}[1][n]{\calP_{#1}}

\newcommand{\ctwostate}{s_2}
\newcommand{\cthreestate}{s_3}
\newcommand{\cfourstate}{s_4}

\newcommand{\ctwostateval}{46}
\newcommand{\cthreestateval}{54}
\newcommand{\cfourstateval}{62}

\newcommand{\ctwobool}{b_2}
\newcommand{\cthreebool}{b_3}
\newcommand{\cfourbool}{b_4}
\newcommand{\ctwoboolval}{16}
\newcommand{\cthreeboolval}{24}
\newcommand{\cfourboolval}{32} \usepackage[margin=.8in]{geometry}

\makeatletter
\newcommand{\oset}[2]{{\mathop{#2}\limits^{\vbox to -.4\ex@{\kern-\tw@\ex@
   \hbox{\scriptsize #1}\vss}}}}
\makeatother

\newcommand{\stabs}{\calS}
\newcommand{\sstate}[1][\stabs]{\rho(#1)}
\newcommand{\stabgens}{G}

\newcommand{\overlarrow}[1]{\oset{\text{\tiny$\leftarrow$}}{#1}}
\newcommand{\overrarrow}[1]{\oset{\text{\tiny$\rightarrow$}}{#1}}
\newcommand{\blc}{\overlarrow{\mathsf{L}}}
\newcommand{\flc}{\overrarrow{\mathsf{L}}}
\newcommand{\dlc}[1]{\blc(\flc(#1))}

\newcommand{\MI}[3][]{\textbf{I}_{{#1}}\pbra{{#2}:{#3}}}
\newcommand{\tnorm}[1]{\norm{#1}_{1}}
\newcommand{\depth}{\textsf{depth}}
\newcommand{\dU}{{\depth(U)}}
\newcommand{\cktblc}[2][U]{{#1}_{\blc({#2})}}
\newcommand{\EEnt}{von Neumann entropy\xspace}
\newcommand{\QMI}{quantum mutual information\xspace}

\newcommand{\EE}{S}

\newcommand{\Cliff}{\mathsf{Cliff}}
\newcommand{\Clifford}{\mathsf{Clifford}}
\newcommand{\cliffqncz}{\QNCZ \circ \Cliff}

\newcommand{\CAT}{\textsf{CAT}}
\newcommand{\CATn}[1][n]{\CAT_{#1}}
\newcommand{\bCAT}[1]{{#1}\CAT}
\newcommand{\bCATn}[2][n]{\bCAT{#2}_{#1}}

\newcommand{\altCQ}[1][]{\mathsf{A}_{#1} \mathsf{CQ}}
\newcommand{\altQC}[1][]{\mathsf{A}_{#1} \mathsf{QC}}
\newcommand{\MH}{\mathsf{MH}}

\newcommand{\cliffdepth}{\textsc{Clifford-rounds}\xspace}
\newcommand{\fanoutdepth}{\textsc{Fanout-depth}\xspace}
\newcommand{\IMrounds}{\textsc{Measurement-rounds}\xspace}
\newcommand{\IMparityrounds}{\textsc{Measurement}[\oplus]\textsc{rounds}\xspace}
\newcommand{\qnczrounds}{\QNCZ\textsc{-rounds}\xspace}
\newcommand{\tdepth}{\textsc{T-depth}\xspace}
\newcommand{\tcount}{\textsc{T-count}\xspace}

\newcommand{\MHlevel}{\textsc{MH-level}}

\newcommand{\Cstab}[1][\ell]{\calC_{#1}^{\textsf{stab}}}

\title{Quantum circuit lower bounds in the magic hierarchy}
\author{Natalie Parham}
\date{Columbia University}

\begin{document}
\maketitle
\begin{abstract}

   We introduce the \emph{magic hierarchy}, a quantum circuit model that alternates between arbitrary-sized Clifford circuits and constant-depth circuits with two-qubit gates ($\textsf{QNC}^0$). This model unifies existing circuit models, such as $\QACZF$ and models with adaptive intermediate measurements. 
    Despite its generality, we are able to prove nontrivial lower bounds.

    We prove new lower bounds in the first level of the hierarchy, showing that certain explicit quantum states cannot be approximately prepared by circuits consisting of a Clifford circuit followed by $\textsf{QNC}^0$. These states include ground states of some topologically ordered Hamiltonians and nonstabilizer quantum codes. Our techniques exploit the rigid structure of stabilizer codes and introduce an infectiousness property: if even a single state in a high distance code can be approximately prepared by one of these circuits, then the entire subspace must lie close to a perturbed stabilizer code. We also show that proving state preparation lower bounds beyond a certain level of the hierarchy would imply \emph{classical} circuit lower bounds beyond the reach of current techniques in complexity theory. 

    More broadly, our techniques go beyond lightcone-based methods and highlight how the magic hierarchy provides a natural framework for connecting circuit complexity, condensed matter, and Hamiltonian complexity.

\end{abstract}
\tableofcontents

\section{Introduction}

Understanding the limitations of quantum circuits is a central question at the intersection of quantum computing and many-body physics. A \emph{quantum circuit lower bound} means proving that a certain class of circuits cannot perform a specific task--- such as preparing a target state $\ket{\psi}$ or implementing a unitary $U$. 
Classical circuit lower bounds have led to breakthroughs across theoretical computer science, establishing deep connections to pseudorandomness \cite{natural-proofs,braverman2008polylogarithmic} and learning theory \cite{carmosino2016learning,oliveira2016conspiracies, lmn}. Furthermore, quantum circuit lower bounds have interesting implications for fundamental questions in quantum matter. 

In condensed matter physics, circuit depth provides a notion of equivalence between quantum phases: states in the same phase can be connected by shallow circuits, while lower bounds indicate phase transitions or topological order \cite{chen2010local,coser2019classification}. In Hamiltonian complexity, if no low energy states of a Hamiltonian can be prepared within a certain circuit class, it suggests long-range entanglement beyond the circuit's expressive power, pointing to exotic quantum phases that persist even at higher temperatures \cite{nirkhe-thesis}.

A key challenge is developing techniques beyond the \emph{lightcone argument}. This principle originates as Lieb-Robinson bounds in many-body physics \cite{lieb1972finite}, which formalizes the idea that quantum information propagates at a limited speed through local interactions, reaching only regions within its \emph{lightcone}. 

Constant-depth quantum circuits with local gates ($\QNCZ$ circuits) have lightcones of constant size, severely limiting their computational power. Any two regions with nonintersecting backwards lightcones can not be more entangled than their backwards lightcones on the input state. Thus globally entangled states such as the $n$-qubit CAT state $\frac{1}{\sqrt{2}} (\ket{0^n} + \ket{1^n})$ cannot be prepared by constant depth circuits starting from the unentangled $\ket{0^n}$ state.

In contrast, nonlocal circuits --- such as those with logarithmic-depth, or many-qubit gates--- have unbounded lightcones, making the argument ineffective for proving lower bounds. 

\begin{figure}
    \centering
    \resizebox{!}{3cm}{
    \begin{tikzpicture}[scale=0.6]
    \fill[fill=yellow, opacity=0.2] (5,1.425) -- (-5, 3) -- (-5, -0.6);
    \draw[dashed] (5,1.425) -- (-5, 3);
    \draw[dashed] (5,1.425) -- (-5, -.6);
    
    \node (as) at (0,0) {
    \begin{quantikz}[row sep = 1em, column sep=2em]
        & \gate[2]{} & \qw & \qw & \qw & \qw & \gate[2]{} & \qw \\
        & \qw & \gate[2]{} & \qw & \qw & \qw & \qw & \qw \\
        & \qw & \qw & \qw & \gate[2]{} & \qw & \qw & \qw \\
        & \qw & \gate[2]{} & \qw & \qw & \qw & \gate[2]{} & \qw \\
        & \qw & \qw & \qw & \qw & \qw & \qw & \qw 
    \end{quantikz}            
    };
    \end{tikzpicture}
    }
    \caption{A (backwards) lightcone of a qubit.}\label{fig:lightcone}
\end{figure}

The lightcone argument can also be used to prove lower bounds for \emph{classical} circuits with local gates, so long as the depth is below $\log(n)$. Any output bit of the circuit can only depend on the inputs in its backwards lightcone. Again, once the circuit depth is logarithmic, the lightcone argument is no longer useful. In fact, proving lower bounds against $O(\log n)$-depth classical circuits, $\NCo$, remains a major open problem in complexity theory.

However, in the classical world, we have been able to push beyond the lightcone argument with circuits that have unbounded lightcones but are still more restricted than $\NCo$ circuits. A prominent example is $\ACZ$ circuits--- boolean circuits of constant depth with AND and OR gates that are allowed arbitrarily many input bits, or ``unbounded fan-in''. Thus the outputs of the circuit can depend on all inputs, rendering the lightcone argument useless. Despite this, the structure of the model has enabled the development of powerful lower bound techniques, such as random restrictions \cite{furst1984parity,haastad1986computational}, Fourier analysis \cite{lmn}, and combinatorial methods \cite{smolensky1987algebraic}. These tools have yielded strong lower bounds and insights across complexity theory \cite{furst1984parity,rossman2015average}, pseudorandomness \cite{chattopadhyay2018pseudorandom,braverman2008polylogarithmic,chattopadhyay2019pseudorandom}, learning theory \cite{KushilevitzMansour:93,eskenazis2022low,kalai2008agnostically}, quantum-classical separations \cite{raz2022oracle}, and more.

The success in going beyond lightcone arguments in classical complexity theory can be attributed to three key aspects of the $\ACZ$ model: 
it is \emph{useful}, as lower bounds are actually provable; \emph{interesting}, due to its connections to broader areas of complexity theory and computation; and \emph{natural}, as it generalizes constant-depth boolean circuits and models parallel computation.

In the quantum setting, there has been some effort towards adapting techniques that were successful for proving classical $\ACZ$ lower bounds. In 1999 Moore proposed a quantum analogue, $\QACZ$, constant depth quantum circuits with arbitrary single qubit gates and the multi-qubit quantum AND gate $\ket{x} \to (-1)^{\textsf{AND}(x)}\ket{x}$ \cite{moore1999quantum}. While it resembles $\ACZ$,  $\QACZ$ has resisted the application of classical techniques--- such as random restrictions or Fourier methods--- that underlie many classical lower bounds. 
Recent work has adapted Fourier analytic methods to quantum circuits \cite{nadimpalli2024pauli},  mirroring \cite{lmn}, though these methods have not yet yielded lower bounds as strong as those known classically.
To date, lower bounds for $\QACZ$ have only been established for depth-two circuits \cite{pade2020depth,rosenthal-qac0,fenner2025tightboundsdepth2qaccircuits} or models with few ancilla qubits \cite{fang2003quantum,nadimpalli2024pauli,anshu2024computational}.
The difficulty of proving lower bounds for $\QACZ$ suggests the need for a circuit model that's expressive--- capable, for example, of generating global entanglement --- but is more tractable to analyze.

In this work, we introduce the magic hierarchy, a quantum circuit model consisting of circuits that alternate between arbitrary-size Clifford circuits and constant-depth circuits with arbitrary two-qubit gates ($\QNCZ$). We define the $k$th level of the magic hierarchy as circuits alternating $k$ times between Clifford and $\QNCZ$ circuits. 
The level in the hierarchy reflects increasing complexity of  \emph{long-range magic}---that is, nonstabilizerness that cannot be removed with a low-depth circuit \cite{wei2025long}--- hence the name

In this paper we aim to demonstrate that, just like the classical $\ACZ$ model, the magic hierarchy is interesting, useful, and natural: 

\vspace{5pt}
\noindent\textit{Useful:} We demonstrate the usefulness of the magic hierarchy by proving lower bounds at its first level, pushing beyond existing lower bounds. Our techniques exploit the rigid structure of both Clifford and $\QNCZ$ circuits. While each of these circuit models is individually well-understood, their respective analysis tools---such as stabilizer methods for Clifford circuits and lightcone arguments for $\QNCZ$--- do not extend naturally to the other.
This tension necessitates the development of novel approaches to establish lower bounds.

\vspace{5pt}
\noindent\textit{Interesting:} We highlight interesting connections between the magic hierarchy and condensed matter physics, Hamiltonian complexity, and important open problems in classical complexity. Additionally, we show that the magic hierarchy unifies other well studied models of quantum computation such as constant-depth circuits with multi-qubit Fanout gates and models of quantum circuits with adaptive intermediate measurements.

\vspace{5pt}
\noindent\textit{Natural:} Clifford circuits and $\QNCZ$ circuits are ubiquitous in quantum computing, both in theory and in practice. These circuits are used in quantum algorithms, error-correction, nonlocal games and more. They also align with constraints in noisy quantum hardware, where the accumulation of noise  limits the circuit depth. In most fault-tolerant architectures, Clifford operations are   inexpensive, while non-Clifford gates are more costly.

\vspace{5pt}
This suggests the magic hierarchy as a natural setting for proving quantum circuit lower bounds beyond constant-depth models.

The remainder of this introduction is structured as follows. We formally define this circuit model in \Cref{ssec:altQCdef}, establish connections between the level of the magic hierarchy and other complexity measures in \Cref{ssec:intro:connections}, and then state our main lower-bound results (\Cref{ssec:results}). We then discuss connections between quantum circuit lower bounds in the magic hierarchy and classical circuit lower bounds and barriers in \Cref{ssec:intro:barriers}, and finally we explore implications and connections to Hamiltonian complexity and condensed matter physics in \Cref{ssec:intro:condmat}.

\subsection{Formally defining the magic hierarchy}   \label{ssec:altQCdef}
We propose the study of quantum circuits that alternate between arbitrary-size circuits with Clifford gates, and constant-depth quantum circuits with arbitrary two-qubit gates, also known as $\QNCZ$. These gates can act between any pair of qubits and are not restricted by geometric locality. This choice simplifies the model while preserving the essential structural limitation of constant-depth circuits---namely, bounded lightcones. Furthermore, since Clifford circuits can permute qubits, enforcing geometric locality only increases the number of alternations by a constant factor.

\begin{figure}[h]
    \begin{center}   
    \begin{quantikz}
        \qw  & \gate[wires=5, nwires=4][3cm]{\Clifford} & \gate[wires=5, nwires=4]{\QNCZ} & \gate[wires=5, nwires=4][3cm]{\Clifford} & \gate[wires=5, nwires=4]{\QNCZ} & \gate[wires=5, nwires=4][3cm]{\Clifford} & \qw\\
        \qw & & & & & & \qw \\
        \qw & & & & & & \qw \\
        \vdots & & &  & & & \vdots \\
        \qw & & & & & & \qw 
    \end{quantikz}
    \end{center}
    \caption{An alternating Clifford and $\QNCZ$ circuit with $k=4$ alternations, starting with a Clifford circuit.}
\end{figure}

\begin{definition}[{$\altCQ[k]$, $\altQC[k]$} ]\label{def:altCQ}
    $\altCQ[k]$ refers to the class of circuits that can be decomposed into a product of $k+1$ circuits, starting with a Clifford circuit and alternating between Clifford and $\QNCZ$ circuits. More formally, a circuit family $\{C_n\}$ is in $\altCQ[k]$ if there exists a constant $d\in \NN$ such that for each $n\in \NN$, $C_n$ acts on $n$ qubits and can be written as the product $Q_k Cl_{k-1} Q_{k-2} \dots Cl_0$ where each $Q_j$ is a depth-$d$ circuit with two qubit gates, and each $Cl_j$ is a Clifford circuit. 
    $\altQC[k]$ is defined similarly but with the first circuit applied being a $\QNCZ$ circuit.

\end{definition}

Equivalently, we can use the following recursive definition: Let $ \altCQ[0] := \Cliff$, the set of all Clifford circuit families, 
and $\altQC[0]: = \QNCZ$, and for each integer $k\geq 1$, let
\begin{align}
    \altCQ[k] &:= \bigcup_{\{C_n\} \in \Cliff} \bigcup_{\ \{A_n\} \in \altQC[k-1]} \{A_n C_n: n \in \NN\}\\
    \altQC[k] &:= \bigcup_{\{Q_n\} \in \QNCZ} \bigcup_{\ \{A_n\} \in \altCQ[k-1]} \{A_n Q_n: n \in \NN\}.
\end{align}
Note that $\altCQ[k] \sse \altQC[k+1]$. We remark that at $k=1$, the order of application in these alternating circuits matters, as the circuit classes $\altCQ[1]$ and $\altCQ[1]$ are incomparable.

\begin{restatable}{theorem}{thmCQneqQC}\label{thm:intro:CQneqQC}
    $\altCQ[1] \not\sse \altQC[1]$ and $\altQC[1] \not\sse \altCQ[1]$.
\end{restatable}
\Cref{thm:intro:CQneqQC} is proved in \Cref{ssec:CQneqQC}. 

The number of alternations serves as a measure of the \emph{magic} or non-Cliffordness of the circuit. Inspired by this, we refer to all circuits with a finite number of alternations as the \emph{magic hierarchy} $\MH$. This name is more convenient for verbal discussion than ``alternating Clifford and $\QNCZ$ circuits with a constant number of alternations.''

\begin{definition}[The magic hierarchy, $\MH$]
    \emph{The magic hierarchy} $\MH := \bigcup_{k \in \NN} \altCQ[k] =  \bigcup_{k \in \NN} \altQC[k]$. Furthermore, for each $k \in \NN$, the \emph{$k$th level of the magic hierarchy} is $\MH_k := \altCQ[k] \cup \altQC[k]$. 
\end{definition}
We emphasize that the magic hierarchy contains circuits with a constant number of alternations that does not scale with $n$. 

For simplicity, we will often not refer to circuit families explicitly. Instead we will mention an $n$-qubit quantum circuit $C$, to implicitly refer to a circuit family $\{C_n\}$. Similarly, when we refer to an $n$-qubit quantum state $\ket{\psi}$, we are implicitly referring to a family of states $\{\ket{\psi_n}\}$ where for each $n\in \NN$, $\ket{\psi_n}$ is an $n$-qubit state.

\paragraph{Ancillary qubits} In \Cref{def:altCQ} we define our circuit class only in terms of the gates, without specifying input to the circuit such as ancillary qubits. When we consider the unitaries $U$ (exactly) \emph{implementable} by a circuit class, we allow the circuits $C$ to use a polynomial number of ancillary qubits $a = \poly(n)$ that must start and end in the $\ket{0}$ state: $C\ket{\phi}\ket{0^{a(n)}} = (U\ket{\phi}) \ot \ket{0^{a(n)}}$. Moreover, when we consider the states (exactly) \textit{preparable} by a circuit class, we mean the states  $\ket{\psi}$  that can be prepared starting from the $\ket{0^n}$ state along with $a = \poly(n)$ ancillary qubits that start and end in the $\ket{0}$ state: $C\ket{0^n} \ket{0^{a(n)}} = \ket{\psi} \ot \ket{0^{a(n)}}$. Our requirement that the ancillas end in the $\ket{0}$ state has been referred to as \emph{clean computation}. For interesting discussion on clean versus non-clean computation we refer the reader to \cite{rosenthal-thesis}.

\subsection{Equivalence with Fanout, Intermediate measurement, and \texorpdfstring{$T$-depth}{T-depth}}\label{ssec:intro:connections}

Now that we have defined the magic hierarchy, we show that it unifies several previously studied circuit models. In particular, the level of the magic hierarchy is a complexity measure that captures the use of nonlocal operations such as fanout gates, or adaptive intermediate measurements. 

\paragraph{Fanout gates}
The multi-qubit Fanout gate $\ket{b, x_1, \dots, x_m} \to \ket{b, x_1 \oplus b, \dots, x_m \oplus b}$ is surprisingly powerful. Constant depth quantum circuits with fanout gates ($\QNCZF$) can compute the quantum Fourier transform, all the non-classical parts of Shor's algorithm  \cite{hoyer2005quantum,cleve2000fast}, and even simulate classical threshold circuits ($\TCZ$) \cite{hoyer2005quantum}. Høyer and Spalek showed that $\QNCZF$ is also equivalent to $\QACZF$,  $\QACZ$ circuits with Fanout gates. 
In 1999, Moore posed the question of whether $\QACZ = \QACZF$, and despite considerable effort \cite{pade2020depth,fang2003quantum,rosenthal-qac0,nadimpalli2024pauli,anshu2024computational} this remains one of the central open problems in quantum circuit complexity. As we show below, the magic hierarchy provides us with a more fine-grained approach for studying $\QACZF$.

We define the \emph{fanout depth}, $\fanoutdepth(U)$, of a unitary $U$ as the minimum number of layers of parallel Fanout gates needed to implement $U$ with a $\QNCZF$ circuit.  We show that the level of the magic hierarchy, $\MHlevel(U)$, required to implement $U$ is equivalent to $\fanoutdepth(U)$ up to constant factors. 

\begin{restatable}{proposition}{fanouteqMHlevel}\label{prop:fanouteqMHlevel}
    $\frac{1}{2} \cdot \MHlevel \leq \fanoutdepth \leq 2 \cdot \MHlevel + 4$.
\end{restatable}
Therefore, increasing levels of the magic hierarchy interpolate between $\QNCZ$ and $\QNCZF$. Since $\QNCZF = \QACZF$ \cite{hoyer2005quantum,takahashi2016collapse}, we conclude the following. 
\begin{restatable}{theorem}{MHeqQACZF}\label{thm:intro:MHeqQACZF}
    $\MH = \QACZF$
\end{restatable}

\paragraph{Adaptive intermediate measurements:} 
Adaptive intermediate measurements are another surprisingly powerful nonlocal operation that is used to supplement the power of constant depth circuits. 
We consider constant-depth quantum circuits that are allowed mid-circuit measurements where classical processing is done on the measurement outcomes and fed into the next stage of the circuit (also known as \emph{feedforward}). More formally, the circuit has the following structure:
\begin{enumerate}
    \item A constant depth quantum circuit $Q$ with 2-qubit gates ($\QNCZ$) is applied to the input state
    \item A subset $S\sse [n]$ of the qubits is measured in the computational basis, getting outcome $x\in \{0,1\}^{|S|}$. Let $\rho_x$ be the post-measurement state on the qubits in $S^\complement$. 
    \item A classical circuit $C:\zo^{|S|}\to \{0,1\}^*$ is applied to the outcome measurements, and the state $\proj{C(x)}\ot \rho_x$ is fed into the next quantum layer of the circuit. 
    \item This continues for a constant number of rounds.
\end{enumerate}

Recent work has shown adaptive measurements significantly enhance the power of constant-depth quantum circuits, both for state preparation and fault tolerance. 
Such circuits can prepare states with long range entanglement and topological order \cite{smith2024constant,piroli2021quantum,
tantivasadakarn2023hierarchy,piroli2024approximating,tantivasadakarn2024long,verresen2021efficiently,lu2022measurement,bravyi2022adaptive,tantivasadakarn2023shortest,li2023symmetry} --- states that are not preparable by standard $\QNCZ$ without adaptivity \cite{bravyi2006lieb}. 
Buhrman et al. refer to these circuits as local alternating classical-quantum circuits (LACCQ) and showed that they can prepare the $W$ state, Dicke state, and many-body scar states \cite{buhrman2024state}. Furthermore, Gidney and Bergamaschi \cite{gidney2025constant} demonstrate that adaptive local Clifford circuits can realize quantum codes with parameters that surpass the Bravyi-Poulin-Terhal bound \cite{BPT} on the rate-distance tradeoff for local stabilizer codes.

In general, when the computational power of the intermediate classical circuits is unbounded then the number of rounds of measurements needed to implement a unitary $U$, $\IMrounds$ provides a lower bound to the level of the magic hierarchy needed to implement $U$.

\begin{restatable}{proposition}{IMleqMHlevel}\label{prop:IMleqMHlevel}
    $\IMrounds \leq 2 \cdot \MHlevel + 4$.
\end{restatable}
Furthermore, if the intermediate classical circuit may contain only parity gates\footnote{The parity gate computes $\textsf{parity}(x) = \sum_i x_i \mod 2$ on its input $x$, and can output many copies of $\textsf{parity}(x)$.}, the required number of rounds $\IMparityrounds$ to implement $U$ is equivalent to the level of $\MH$ required.

\begin{restatable}{proposition}{IMpareqMHlevel}\label{prop:IMpareqMHlevel}
    $\IMparityrounds = 2 \cdot \MHlevel + 4$
\end{restatable}
Our discussion thus far has only considered the complexity of implementing \emph{unitaries}, but circuits with intermediate measurements implement channels that are not in general unitary. Granting $\MH$ circuits partial trace and considering the complexity of implementing a quantum channel, \Cref{prop:IMleqMHlevel,prop:IMpareqMHlevel} also hold for general quantum channels. 

\paragraph{Magic} The level of the magic hierarchy serves as a gate-set-agnostic measure of nonstabilizerness--- also known as ``magic''--- in quantum computation. Recent concurrent works have referred to states that cannot be prepared by a Clifford circuit followed by a $\QNCZ$ circuit as exhibiting \emph{long-range magic} \cite{wei2025long}, \emph{nonlocal magic} \cite{andreadakis2025exact} or \emph{long-range nonstabilizerness} \cite{korbany2025long}. 

The most widely used measure of magic is $T$-count, the number of $T$ gates needed in a Clifford$+T$ circuit. Circuits with $T$-count $t$ can be simulated classically in time $\poly(n)2^{O(t)}$ \cite{bravyi2016improved}. $T$-count is also a meaningful cost metric in fault-tolerant quantum computing, where $T$ Clifford gates are typically inexpensive, but non-Clifford operations—such as $T$ gates—are costly to implement.

The $T$-\textit{depth} --- the number of layers of $T$ gates \cite{amy2013meet}--- reflects the cost of the $T$ gates when it is feasible to apply them in parallel. The level of the magic hierarchy provides a lower bound on $T$-depth, which in turn lower bounds $T$-count.

\begin{restatable}{proposition}{proptdepthmhlevel}\label{prop:tdepth-mhlevel}
    $\frac{1}{2} \cdot \MHlevel \leq \tdepth \leq \tcount$.
\end{restatable}

The $T$-depth may be relevant in scalable fault-tolerant architectures with sufficient space to implement $T$-gates in parallel. However, in near-term hardware where such parallelization is limited, $T$-count remains the more practical metric.

It isn't clear whether $\tdepth$ is a lower bound on the level of the magic hierarchy to within a constant factor. The Solovay-Kitaev theorem tells us that any $\QNCZ$ circuit can be $\eps$-approximately implemented by a Clifford$+T$ circuit but with circuit size overhead $\poly \log (1/\eps)$, and it is unclear whether this can be transformed into a circuit with constant $T$-depth, independent of $\eps$. We leave this as an open question.

\begin{openproblem}
    Is $\MHlevel = \Theta(\tdepth)$?
\end{openproblem}

\noindent We sketch the proofs of \Cref{prop:fanouteqMHlevel,prop:IMleqMHlevel,prop:IMpareqMHlevel}  in \Cref{ssec:connections}.

\subsection{Lower bounds in the first level of the magic hierarchy}\label{ssec:results}
As a first step to characterizing the computational power of the magic hierarchy, we develop techniques to prove limitations of circuits in the first level. In particular, we consider circuits consisting of first an arbitrary-size Clifford circuit followed by a constant-depth circuit ($\altCQ[1]$). We use two different techniques for proving that explicit state families can not be approximately prepared by $\altCQ[1]$. The first argument applies to states that satisfy certain properties about their mutual information. The second result is relevant to codewords of quantum error correcting codes --- or from the physics perspective, states with topological order.

\subsubsection{Mutual information-based lower bounds}
For a quantum state $\rho$ we denote its mutual information between disjoint subsets of qubits $A, B$ as $\MI[\rho]{A}{B}$. We show that if a quantum state has many small regions with non-integer mutual information, and the mutual information stays below the next integer
 as the region sizes increase, then these states can not be prepared by $\altCQ[1]$ circuits.

\begin{theorem}[Informal, see \Cref{thm:robust-lb-with-params}]\label{thm:robust-cliffqnc-lb}
    Suppose $\psi$ is an $n$ qubit state satisfying:
    \begin{enumerate}
        \item The mutual information between any pair of qubits is at least $\alpha$
        \item  The mutual information between any two disjoint regions of size $s$ is at most $\beta$.
        \end{enumerate}
    for some constants $\alpha, \beta$ with $k < \alpha \leq \beta < k+ 1$, and integers $k\geq 0$, $s\geq 1$.  
If a Clifford circuit $C$ and a $\QNC$ circuit $U$ approximately prepares $\ket{\psi}$ to within trace distance $\eps$, then $U$ has depth at least $ \Omega\pbra{\log\min\cbra{ s , 1/\eps}}$. In particular, when $s = \omega(1)$ this state can not be prepared by $\altCQ[1]$ to within $o(1)$ error.
\end{theorem}

\paragraph*{Example 1: the $\gamma$ biased CAT state} For each $\gamma \in [0,1]$, the $\gamma$-biased CAT state is defined as
\[
    \ket{\gamma CAT} := \sqrt{\gamma} \ket{0^n} + \sqrt{1 - \gamma} \ket{1^n}.
\]
For any disjoint regions $A, B$ such that $A \cup B \subset [n]$, it can be easily verified that $\MI{A}{B} = H(\gamma)$ where $H(p) := -p \log p - (1-p) \log (1-p)$ is the binary entropy function. We can set $H(\gamma)$ to take any value in $[0,1]$ by appropriately setting $\gamma \in [0, 1/2]$. We can thus apply  \Cref{thm:robust-cliffqnc-lb} with $s=n-1$ to get a depth lower bound of $\Omega\pbra{\log\min\cbra{ n, 1/\eps}}$. 

\paragraph{Example 2: The W-state} The $W$ state is the uniform superposition over all computational basis states with hamming weight $1$.
\[\ket{W_n} = \frac{1}{\sqrt{n}}\sum_{\substack{x\in \zo^n: \  |x| = 1}} \ket{x}\]
The $W$ state has small but nonzero mutual information between all regions. \Cref{thm:robust-cliffqnc-lb} allows one to show that there exists some polynomial $p$ such that $\altCQ[1]$ cannot prepare states that are $1/p(n)$-close to the $\ket{W_n}$ state.

\paragraph*{Lower bounds for gluing biased CAT states}
The mutual information-based techniques used to prove \Cref{thm:robust-cliffqnc-lb} can also be used for lower bounds against constant-depth circuits ($\QNCZ$) for some other state transformation problems. As an example, we show that it is not possible to glue two CAT states into a larger one with ($\QNCZ$). 
\begin{theorem}[see \Cref{thm:cat-state-gluing-lb}] \label{thm:informal:cat-state-gluing-lb}
    Suppose $U\in \unitaries[2n]$ is a $2n$ qubit quantum circuit for $n\geq 2$ such that 
    \begin{align}
        \ket{\bCATn[2n]{\alpha}} = U \ket{\bCATn[n]{\beta}} \ot \ket{\bCATn[n]{\beta}}  \label{eq:catstate-gluing}
    \end{align}
    For some $\alpha, \beta \in [0,1]$. 
    If $H(\alpha) \neq 2 H(\beta)$, then $U$ requires depth $\Omega(\log n)$. Furthermore, if $\alpha$ and $\beta$ satisfy $\abs{H(\alpha) - 2H(\beta)} = \Omega(1)$ then even $\eps$-approximately preparing the state requires depth $\Omega(\log \min \{n , 1/\eps\})$.
\end{theorem}
 In the case we also have that $H(\alpha) > 2H(\beta)$ then we can prove a lower bound of $\Omega(\log n)$, which is indpendent of approximation error! This is found in \Cref{thm:cat-state-gluing-eps-indep}.
This lower bound leaves open the setting when $H(\alpha) = 2 H(\beta)$.
\begin{openproblem}
    Can a constant-depth circuit transform two $\beta$-biased CAT states into one $\alpha$-biased CAT state, as in \Cref{eq:catstate-gluing}, for $H(\alpha) = 2 H(\beta)$?
\end{openproblem}

\subsubsection{Codespace lower bounds}

In the next result we identify quantum error-correcting codes such that $\altCQ[1]$ circuits can not approximately prepare \emph{any} of the states in their codespace. More formally, we consider families of codes $\{\calC_n\}$ where for each $n\in \NN$, $\calC_n\sse \CC^{2^n}$ is a subspace of $n$-qubit states. We say that $\{\calC_n\}$ has distance $d(n)$ if, for each $n$ and any subset $A$ of less than $d(n)$ qubits, there exists a recovery map that can correct erasure on $A$: that is, $\textsf{Rec}(\tr_A \rho) = \rho$ for each $\rho \in \calC_n$. For notational simplicity, we often refer to a code $\calC$ on $n$ qubits, implicitly referring to a family of codes.

We say that a subspace $\calC$ is a \emph{perturbed stabilizer code}  if there is some $\QNCZ$ unitary that maps $\calC$ to a stabilizer code $\calC'$. \footnote{More formally, we are referring to families of codes and unitaries $\{\calC_n\}$, $\{\calC'_n\}$ and $\{U_n\}$.}
It is not hard to see that for each perturbed stabilizer code there is a $\altCQ[1]$ encoding circuit.  This is because any stabilizer code has a Clifford encoding circuit.

We show that essentially, this is all $\altCQ[1]$ can do. 
If a local\footnote{A quantum error correcting code is $\ell$-\emph{local} if it is defined as the ground space of an $\ell$-local Hamiltonian.} quantum error correcting code with super-constant distance is \emph{not} approximately a perturbed stabilizer code, then $\cliffqncz$ cannot prepare \emph{any} of the states in the codespace, even approximately. We show that this applies to:
\newcommand{\Hcat}{\Psi^{\textsf{CAT}}}
\begin{enumerate}
    \item \textbf{The Feynman-Kitaev history state} for the circuit that prepares $\ket{\textsf{CAT}_n} = \frac{1}{\sqrt{2}}(\ket{0^n} + \ket{1^n})$.
    \begin{align*}
    \ket{\Hcat}:= \frac{1}{\sqrt{n}} \sum_{t=0}^n \ket{\textsf{unary(t)}}_{\textsf{time}} \ot\ket{\textsf{CAT}_t} \ket{0^{n-t}}_{\textsf{state}}
    \end{align*}
    \begin{restatable}{theorem}{thmcathistorystaterobustlb}
        There is some $\eps \geq \Omega(1/n^3)$ such that preparing a state that is $\eps$-close in trace distance to the state $\ket{\Hcat}$ starting from a stabilizer state requires a circuit that has depth at least $\Omega(\log n)$
    \end{restatable}
    \item Codes with codespaces with \textbf{dimension that is not a power of two} 
    \begin{restatable}{theorem}{thmdimrobustlb}\label{thm:dim-power2-lb}
    Suppose $H = \sum_{i=1}^m h_i$ is an $\ell$-local Hamiltonian with $m \geq n$, each $\norm{h_i}_\infty \leq 1$, and gap $\Delta\leq 1$. Let $\calC$ be the groundspace of $H$ and suppose it has distance $d$. If the dimension of $\calC$ is not a power of two then there exists an $\eps \geq \Omega(\Delta/m)$ such that preparing a state $\eps$ close to $\calC$ starting with a stabilizer state requires a circuit with depth $\geq \frac{1}{2}\log(d/\ell)$ 
\end{restatable}
    Examples of such codes include some non-regular XP stabilizer codes \cite{webster2022xp}, the quantum double model for some nonabelian anyons \cite{cui2020kitaev}, and the Fibonacci anyon code on certain manifolds as mentioned in \cite{wei2025long,schotte2022quantum}. 
\end{enumerate}

To prove these, our key structural result is an \emph{infectiousness property}: if a codespace contains even one state preparable by $\cliffqncz$, then the entire codespace is a perturbed stabilizer code.
\begin{theorem}[Infectiousness of {$\cliffqncz$} states in codes]
    \label{thm:intro:CQinfectious}
    Let $\calC$ be the groundspace of an $\ell$-local Hamiltonian with distance $\omega(\ell)$. If there exists a state $\ket{\psi}\in \calC$ that can be prepared by a $\cliffqncz$ circuit, then $\calC = U \calC_S$ for some stabilizer code $\calC_S$ with $O(\ell)$-local stabilizer generators, and some $\QNCZ$ circuit~$U$. 
\end{theorem}
\noindent Surprisingly, this theorem allows us to reason about the structure of the entire code $\calC$, even if just one of its codewords can be prepared by an $\cliffqncz$ circuit. For this reason it may be independently interesting. 

We also prove a robust infectiousness property (\Cref{thm:robust-code-stab-containment,cor:infectiousness-ham-technical}): If the Hamiltonian is gapped\footnote{We note that typically a Hamiltonian is considered ``gapless'' if the spectral gap $\Delta$ is decreasing with $n$. But in our result, as long as the spectral gap is $\Delta = 1/\poly(n)$, we can still get a lower bound against approximating the state to within error below some $1/\poly(n)$.}, then approximately preparing a codestate means that the entire codespace is approximately a perturbed stabilizer code.

Surprisingly, this theorem allows us to reason about the structure of the entire code $\calC$, even if just one of its codewords can be approximately prepared by an $\altCQ[1]$ circuit. For this reason it may be independently interesting.

\paragraph{Ancillas} Our code-based lower bounds hold against any $\altCQ[1]$ circuit that uses arbitrarily many ancilla qubits that start and end in the all zeros state. This is because we can consider the codespace in the larger Hilbert space including the extra ancilla qubits set to $\ket{0}$. The Hamiltonian will just have an extra term $(-\proj{0}_i)$ for each ancilla qubit $i$, and the distance of the code does not increase.

We remark that in independent concurrent work, the authors of \cite{wei2025long} also show the infectiousness property of $\altCQ[1]$ circuits (\Cref{thm:intro:CQinfectious}), with the robust version implicit in the proof of their Theorem 9.  Furthermore, the authors also independently prove \Cref{thm:dim-power2-lb}.

\subsection{Connections to classical circuit lower bounds and the natural proofs barrier} \label{ssec:intro:barriers}

While the magic hierarchy suggests a structured approach to proving increasingly stronger lower bounds one level at a time, in this section we discuss the implications of these lower bounds for \emph{classical circuit complexity}. These connections are summarized in \Cref{fig:barriers} for state preparation lower bounds.

\begin{figure}
    \centering
    \resizebox{!}{8cm}{\begin{tikzpicture}
        \def\xaxis{1}      \def\xlabels{\xaxis -0.2}      \def\xshading{\xaxis}      \def\width{7}  
    \def\stwoheight{3}
    \def\sthreeheight{4}
    \def\sfourheight{5}
    \def\oneheight{0.5}
    \def\height{8}

        \draw[thick,->] (\xaxis,0) -- (\xaxis,\height);

    \node[anchor=east] at (\xaxis-1.3, 0.5*\height) {Level of $\MH$};

        \foreach \y/\label in {0/0, \oneheight/1, \stwoheight/$s_2$, \sthreeheight/$s_3$, \sfourheight/$s_4$} {
        \node[left] at (\xlabels,\y) {\label};          \draw[dashed, line width = 0.7pt] (\xaxis,\y) -- (\xshading+\width,\y);      }
    \node[anchor=west] at (\xshading + \width, 0) {Clifford, $\QNCZ$};
    \node[anchor=west] at (\xshading+\width, \oneheight) {Our lower bounds};
    \node[anchor=west] at (\xshading+\width, \stwoheight) {implies depth-2 $\TCZ$ lower bounds};
    \node[anchor=west] at (\xshading+\width, \sthreeheight) {implies depth-3 $\TCZ$ lower bounds};
    \node[anchor=west] at (\xshading+\width, \sfourheight) {implies depth-4 $\TCZ$ lower bounds};

        \fill[blue!20, opacity=0.5] (\xshading,\oneheight) rectangle (\xshading+\width,\sthreeheight);
    \node[above, align=center] at (\xshading+\width/2,1.75) {Requires new techniques \\no known barriers};

    \fill[gray!30, opacity=0.5] (\xshading,\sthreeheight) rectangle (\xshading+\width,\sfourheight);
    \node[above] at (\xshading+\width/2,\sthreeheight) { Groundbreaking $\TCZ$ lower bounds};

    \fill[red!20, opacity=0.5] (\xshading,\sfourheight) rectangle (\xshading+\width,\height);
    \node[above] at (\xshading+\width/2,\sfourheight ) {Natural proofs barrier for $\TCZ$};

\end{tikzpicture}
    }
    \caption{Connection between magic hierarchy lower bounds for state preparation and classical circuit  lower bounds. At various levels of the magic hierarchy, lower bounds imply depth lower bounds for $\TCZ$. We indicate the barriers facing $\TCZ$ lower bounds in the shaded regions. We show that $\ctwostate \leq \ctwostateval, \cthreestate \leq \cthreestateval$, and $\cfourstate \leq \cfourstateval$, but it is open whether these upper bounds are tight.}
    \label{fig:barriers}
\end{figure}

\paragraph{Sufficiently strong lower bounds against $\MH$ imply $\TCZ$ lower bounds}
It is a longstanding open problem to prove that an explicit Boolean function cannot be computed by $\TCZ$ circuits, constant-depth classical circuits with threshold gates. $\TCZ$ represents a key frontier in circuit complexity: while lower bounds are known for depth-2 threshold circuits, no nontrivial lower bounds are known even for depth-3 $\TCZ$ circuits.

Since $\QACZF$ circuits can simulate $\TCZ$ circuits \cite{hoyer2005quantum}, any lower bound against $\MH = \QACZF$ for computing a Boolean function would immediately imply a lower bound against $\TCZ$ — a major breakthrough in classical complexity theory, with potential implications for breaking longstanding cryptographic assumptions. In particular, $\TCZ$ encounters the \emph{natural proofs barrier} \cite{natural-proofs}. There are known constructions of pseudorandom functions implementable by depth-4 $\TCZ$ circuits \cite{prf-tc0-naor,krause2001pseudorandom} under the decisional Diffie-Hellman (DDH) assumption. Thus, any ``natural proof'' that establishes an explicit lower bound against depth-4 $\TCZ$ could be used to break the security of DDH.

Because natural proofs encompass many existing techniques, this barrier is regarded as a central and serious obstacle to proving explicit circuit lower bounds — and it underscores why progress against $\TCZ$ (or equivalently against $\MH$) would constitute a major advance.

In the context of state preparation, the story is less obvious. It isn't immediately clear that a lower bound for state preparation would imply any classical circuit lower bound. Recently, Rosenthal's work  \cite{rosenthal2024efficient} shows that for each quantum state $\ket{\psi}$, there exists a Boolean function $f_\psi$ such that a $\QACZF$ circuit with query access to $f_\psi$ can prepare $\ket{\psi}$ to within exponentially small error\footnote{We refer the interested reader to Rosenthal's thesis \cite{rosenthal-thesis}, for a clear exposition of this indirect reduction.}. This implies that quantum circuit lower bounds for preparing some explicit state do, in fact, imply classical circuit lower bounds for an explicit Boolean function.

\paragraph*{Implied depth lower bounds of $\TCZ$}
An explicit state preparation lower bound against the $k$th level of the magic hierarchy would imply the following $\TCZ$ depth lower bounds.
\begin{restatable}{theorem}{statepreptczlb}\label{thm:stateprepbarrier}
    Suppose an explicit family of quantum states $\{\ket{\psi_n}\}_n$ can not be prepared by $\altQC[8d +30]$ to within trace distance $\eps = 1/\poly(n)$ for some $d \in \NN$. Then an explicit family of boolean functions $\{f_n^\psi\}_n$ (defined in \cite{rosenthal2024efficient}), can not be computed by depth-$d$ $\TCZ$ circuits.
\end{restatable}
For implementing a boolean function $f$ with a quantum circuit $U_f : \ket{x}\ket{b}\ket{0^a} \to \ket{x}\ket{f(x) + b} \ket{0^a}$, we get the analogous statement.

\begin{restatable}{theorem}{boolfunctczlb}\label{thm:boolfuncbarrier}
    Suppose the family of boolean functions $\{f_n\}_n$ can not be implemented by $\altQC[8d]$ for some $d \in \NN$, then  $\{f_n\}_n$ can not be implemented by depth-$d$ $\TCZ$ circuits.
\end{restatable}

\Cref{thm:boolfuncbarrier,thm:stateprepbarrier} are proved in \Cref{sec:classical-conn}.
Using the above theorems, we see that there exists universal constants $1\leq \ctwostate \leq \cthreestate, \leq \cfourstate$ such that within the first $\ctwostate \leq \ctwostateval$ levels state preparation lower bounds do not imply any lower bounds for $\TCZ$, while above the $\ctwostate$th level, these lower bounds would depth-2 $\TCZ$ lower bounds. Since depth-2 $\TCZ$ lower bounds are already known this lends itself as reasonable testing grounds for proving lower bounds.

Above level $\cthreestate \leq \cthreestateval$, state preparation lower bounds imply depth-3 $\TCZ$ lower bounds, consitituting major progress in complexity theory. Proving depth-3 $\TCZ$ lower bounds has been an open problem complexity theorists have worked on for over a decade, though proving such a lower buond would not imply anything unexpected such as breaking cryptography. 

On the other hand, above level $\cfourstate \leq \cfourstateval$, we encounter the natural proofs barrier for depth-4 $\TCZ$ circuits. As discussed above, any explicit lower bound against depth-4 $\TCZ$ that constitutes a ``natural proof'' \cite{natural-proofs} would break the security of the decisional Diffie-Hellman (DDH) assumption. Thus, assuming DDH security, it is not possible to prove lower bounds above the $\cfourstate$th level of the magic hierarchy with a proof technique that translates to a ``natural proof'' against $\TCZ$. Though, it is unclear how restrictive this is on quantum circuit lower bound techniques.

Furthermore, we show the same is true for lower bounds for implementing a boolean function in the magic hierarchy, with our estimations $\ctwobool \leq \ctwoboolval, \cthreebool \leq \cthreeboolval$ and $\cfourbool \leq \cfourboolval$. We note that these limits would occur at smaller depth if \Cref{thm:stateprepbarrier,thm:boolfuncbarrier} were tightened, but we are unaware of a tighter argument. 

The most optimistic takeaway is that stronger circuit lower bounds for the magic hierarchy can lead us to a non-natural proof against $\TCZ$ --- which would be a groundbreaking result. A more cautious view is that progress beyond a certain level may ultimately be infeasible. However, there is still a lot of room for progress in the lower levels of the hierarchy. For example, $\MH$ lower bounds for level-2 or level-3 lower bounds may be feasible but still would require new techniques and could have interesting implications in condensed matter. Furthermore, these connections aren't known to hold in the setting where the quantum circuit is not allowed ancillas, or for lower bounds for implementing unitaries.

\subsection{Classical simulability, Hamiltonian complexity}
Considering a growing complexity at each level of the hierarchy, it is natural to ask: how does this affect the classical complexity of simulation? It is well-known that one can efficiently estimate $O(\log)$-sized observables of states prepared by $\QNCZ$ circuits by just simulating the backwards lightcone. We note that this remains true for states prepared by $\altCQ[1]$ circuits. 
\begin{theorem}[Folklore]
    For each family of quantum states $\{\ket{\psi_n}\}$ prepared with $\altCQ[1]$, and any observable $O$ supported on $O(\log n)$ qubits, the expectation value $\bra{\psi_n} O \ket{\psi_n}$ can be computed in time $\poly(n)$ classically. 
\end{theorem}

Since local observables of these $\altCQ[1]$ states can be computed efficiently classically, so can their energies with respect to a local Hamiltonian. Therefore, assuming $\textsf{QMA} \neq \textsf{NP}$, there exists a family of local Hamiltonians with ground states that cannot be prepared by $\altCQ[1]$ circuits. Furthermore, the quantum PCP conjecture \cite{aharonov2013guest} further implies that there exists families of local Hamiltonians such that even all their low-energy states--- states with energy below some constant--- can not be prepared by $\altCQ[1]$. 
\begin{conjecture}[No low-energy {$\altCQ[1]$} states] \label{conj:nltsCQ}
    There exists a fixed constant $\eps > 0$, an explicit family of $O(1)$-local Hamiltonians $\{H^{(n)}\}$ with ground-energy $0$, where $H^{(n)} = \sum_{i =1}^m h^{(n)}_i$ acts on $n$ qubits and $m = \Theta(n)$ such that for any family of states $\{\psi_n\}$ satisfying 
    $\Tr(H^{(n)} \psi_n) < \eps n$, the state $\psi_n$ can not be prepared by an $\altCQ[1]$ circuit.
\end{conjecture}
The authors of \cite{wei2025long} coin this conjecture the \emph{No low-energy trivial magic} conjecture. It was also mentioned in a talk by Nirkhe \cite{nirkhe-talk-beyondnlts} along with the following fact.
\begin{fact}\label{fact:impliesnlts}
    The quantum PCP conjecture and $\textsf{QMA} \neq \textsf{NP}$ implies \Cref{conj:nltsCQ}.
\end{fact}
\Cref{conj:nltsCQ} is a generalization of the longstanding No Low Energy Trivial States (NLTS) conjecture \cite{freedman2013quantum}, now a proven theorem \cite{anshu2023nlts} which proves the above but against $\QNCZ$ circuits rather than $\altCQ[1]$. The authors of \cite{coble2023local} subsequently proved a strengthening of NLTS by proving that a family of Hamiltonians has low energy states that can not be prepared by either $\QNCZ$ or Clifford circuits --- in other words they cannot be prepared in the 0th level of the magic hierarchy. We note that the ground states of the Hamiltonians in \cite{coble2023local} can in fact be prepared by $\altCQ[1]$ circuits. Thus they are a natural next candidate for a strengthening of the NLTS theorem. \Cref{conj:nltsCQ} can be physically interpreted as the statement that it is possible for states with long range entanglement and magic to exist in an environment at ``room temperature.''

Our results in \Cref{ssec:results} contain lower bounds for ground states of local Hamiltonians, but our techniques fall short in proving lower bounds for all states with energy at most  $\eps n$ for some constant $\eps$. Rather, our results can be interpreted as lower bounds against preparing states with energy $\eps n$ for some $\eps = 1/\poly(n)$.

Can we further generalize an NLTS-like statement to higher levels of the magic hierarchy?  We pose this as an interesting open question for further exploration.

\begin{openproblem}
    What is the classical computational complexity of estimating local (size $O(1)$) observables of states prepared by $k$ alternations between Clifford and $\QNCZ$ circuits? 
\end{openproblem}
Zhang and Zhang \cite{Zhang2024ClassicalSO} recently showed that computing local output observables to within \emph{multiplicative} error 1 is GapP-hard at level $k=2$. Multiplicative error $1$ meaning that the estimate $\braket{\wt{P}}$ satisfies $|\braket{\wt{P}} - \braket{P}| \leq |\braket{P}|$. More specifically, they show that it is GapP-complete to compute single-qubit Pauli observables of states prepared by circuits of the form Clifford followed by single qubit gates $T^{1/2} = \proj{0} + e^{i \pi/8} \proj{1}$, and then again Clifford (Theorem III.4 in \cite{Zhang2024ClassicalSO}). On the other hand, they show that if you instead want to compute Pauli observables of these circuits up to $1/\poly(n)$-\emph{additive} error, then this can be done in polynomial-time classically. More generally, they construct an algorithm for estimating Pauli observables up to $1/\poly(n)$-additive error for states prepared by circuits of the form Clifford$-D-$Clifford, where $D$ is a diagonal circuit with efficiently-computable entries.

An exciting next step is the question: are the single-output observables of Clifford$-\QNCZ-$Clifford circuits ($\altCQ[2]$) also easy to estimate classically to within $1/\poly(n)$ additive error?

\subsection{Summary and future directions}
In this work, we introduced the \emph{magic hierarchy}, a quantum circuit model that captures and unifies several previously studied models, interpolating between circuits where lower bounds are well understood ($\QNCZ$ and Clifford circuits) and the more powerful model $\QACZF$, for which proving lower bounds remains a major open problem with significant implications for classical complexity theory. Our main technical results establish lower bounds at the first level of the hierarchy. We developed novel techniques to show that certain explicit quantum states—such as ground states of topologically ordered Hamiltonians and nonstabilizer quantum codes—cannot be approximately prepared by a Clifford circuit followed by a $\textsf{QNC}^0$ circuit. More broadly, the magic hierarchy provides a natural framework for developing stronger circuit lower bounds and for understanding exotic quantum phases through the lens of circuit complexity.

We now highlight several directions for future exploration.

\paragraph{Strengthening lower bounds in the magic hierarchy} Much remains unknown about the structure of the magic hierarchy. We highlight several open problems related to proving stronger circuit lower bounds.

\begin{openproblem}
    Prove a lower bound against $\altQC[1]$ or a higher level of the magic hierarchy.
\end{openproblem}

Since $\MH = \QACZF$ the magic hierarchy provides a pathway of increasing complexity starting from $\QNCZ$ all the way to $\QACZF$. 
Is this hierarchy strict? Does each additional alternation increase computational power?
\begin{openproblem}[Magic state hierarchy]\label{openprob:statehierarchy}
    Does there exist an explicit sequence of families of quantum states $\{\ket{\psi_n^1}\}, \{\ket{\psi_n^2}\}, \dots$ such that for each $k \in \NN$, the family of states $\{\ket{\psi_n^k}\}$ can be  prepared with $k$ alternations between Clifford and $\QNCZ$ circuits, but can not be approximately prepared with only $k-1$ alternations?
\end{openproblem}
We note that, even proving this conjecture for a specific sequence would still not establish a lower bound on $\MH$ or $\QACZF$ since all such states still lie within the hierarchy.

Another interesting question concerns the relationship between the two classes  $\altCQ[k]$ and $\altQC[k]$. In \Cref{thm:intro:CQneqQC} we show that $\altCQ[1]$ and $\altQC[1]$ are incomparable in the context of implementing a \emph{unitary}, but their relationship for state preparation remains open.
\begin{openproblem}
    Are $\altCQ[1]$ and $\altQC[1]$ comparable for state preparation?
\end{openproblem}

\paragraph{Condensed matter}\label{ssec:intro:condmat}
In condensed matter systems, circuit complexity plays a fundamental role in characterizing quantum phases, as two states belong to the same phase if and only if a constant-depth circuit with geometrically-local gates transforms one into the other \cite{chen2010local}. Consequently, distinguishing quantum phases can establish circuit lower bounds in settings with local gates on a finite-dimensional lattice.

A recent conjecture in \cite{tantivasadakarn2023hierarchy} suggests that intermediate measurements create a hierarchy of topologically ordered phases for preparing quantum double states. \cite{sheffer2025preparing} further show how these states can be constructed using constant-depth circuits with fanout gates ($\QNCZF$). This suggests that the computational power is strictly increasing with the level of the magic hierarchy. Furthermore, \cite{tantivasadakarn2023hierarchy} conjectures that the Fibonacci anyon code and the quantum double model for nonsolvable groups lie beyond this hierarchy, implying that it is not in $\QACZF$. While these ideas remain conjectural, they suggest promising avenues for leveraging the magic hierarchy to study topological phases and circuit lower bounds.

\paragraph{Learning} Seperately, Clifford circuits and constant-depth ($\QNCZ$) circuits are efficiently learnable \cite{aaronson2004improved,huang2024learning}. What is the sample and computational complexity of learning circuits in the magic hierarchy?

\paragraph{Upper bounds} In this work we only proved lower bounds for preparing certain states within the magic hierarchy. Are there matching upper bounds? 

\paragraph{Quantum advantage} Constant depth quantum circuits already exhibit quantum advantage for sampling tasks under standard complexity assumptions \cite{terhal2002adaptive,aaronson2005quantum}. However, they offer no advantage for decision problems, since observables on small subsets of qubits can be efficiently computed classically.
\begin{openproblem}
    Is there a decision (or search) problem that can be solved by the magic hierarchy and can not be solved in polynomial time classically? If so, at which level does this occur?
\end{openproblem}
Factoring, for example, can be computed using circuits in the magic hierarchy along with classical processing. If one could show that this classical processing---namely computing the greatest common divisor---can be done in the magic hierarchy, then so can factoring. 

In the case where there exists a specific level $k$ at which we can achieve quantum advantage, this is excellent news for achieving quantum advantage on near-term quantum hardware. There would only need to be a constant number of ``difficult-to-implement'' layers of the circuit: whether as non-Clifford gates on a fault-tolerant device, or fanout gates in a constant depth ciruit.

On the other hand, if there is no level that achieves quantum advantage, then this implies a separation between the magic hierarchy and $\QACZF$ from $\textsf{BQP}$, assuming that polynomial-time quantum computation is more powerful than classical $\textsf{BQP} \neq \textsf{BPP}$.

\subsection*{Statement on concurrent work}
While this paper was in preparation, two independent works --- \cite{korbany2025long,wei2025long} --- appeared, which also prove lower bounds for $\altCQ[1]$ circuits. These works are from more of a condensed matter perspective and refer to states that are not preparable by these circuits as being in phases of matter with so-called  \emph{long-range-nonstabilizerness} and \emph{long-range magic} respectively.

\cite{korbany2025long} prove that Clifford circuits followed by \emph{geometrically-local} constant-depth circuits cannot prepare certain 1D matrix product states. Like our \Cref{thm:robust-cliffqnc-lb}, their proof is also based on mutual information. Like us, they exploit the fact that stabilizer states have integer-valued mutual information, show that low depth circuits cannot change the mutual information, and apply the Fannes-Audenaert inequality to make their bounds robust. Unlike our approach, their argument relies on geometric locality and specific structural properties of 1D translationally invariant RG fixed points.

\cite{wei2025long} independently proved the infectiousness property (\Cref{thm:intro:CQinfectious}). Their initial version established only the exact form, but in their more recent updated version, the robust version is implicit in their proof of \blue{Theorem 9}, and can be derived with little additional effort. They also independently prove \Cref{thm:dim-power2-lb}, showing that codes whose dimension is not a power of two cannot be prepared by $\altCQ[1]$ circuits.  While their circuit model is formulated with geometrically-local gates, we note that all but one of their lower bound results extend naturally to the non-local gate setting relevant for $\altCQ[1]$.

\subsection*{Acknowledgements}
I thank Henry Yuen for many helpful discussions and feedback on the magic hierarchy. I thank Sergey Bravyi and Chinmay Nirkhe for helpful conversations in the early days of this project. I am also grateful to Ben Foxman, Assaf Harel, Dale Jacobs, Gregory Rosenthal, Jack Spilecki, and Tianqi Yang. I thank Fuchuan Wei and Zi-Wen Liu for useful correspondance regarding \cite{wei2025long}. This work is supported by the Google PhD fellowship and the Department of Energy Grant GG019849. I thank the Simons Institute for the Theory of Computing, where some of this work was conducted.

\section{Preliminaries}
\subsection{Notations}
For each $n\in \NN = \{1, 2, 3, \dots\}$, denoting the number of qubits we consider the Hilbert space $\Hilbertn = \CC^{2^n}$. Let $\linopsn$ be the set of linear operators mapping $\Hilbertn$ to $\Hilbertn$. We denote the set of $n$-qubit density matrices As $\states$, elements of $\linopsn$ that are positive semidefinite and have trace 1. We use $\tr(\cdot)$ to denote trace, and $\tr_S(\cdot)$ to denote partial trace over some subset $S\sse [n]$ of the qubits. For $X\in \linopsn$, we denote trace norm as $\tnorm{A} = \tr\pbra{|A|}$, where $|A| = \sqrt{A^\dagger A}$. We use $\log(\cdot)$ to denote logarithm base $2$. 
\subsection{Circuits}
A quantum circuit $U$ on $n$-qubits is the matrix product of sequence of unitaries $U = U_d U_{d-1} \dots U_1$ such that each $U_i$ is a tensor product of one and two-qubit unitaries, or \emph{gates}. The \emph{depth} of $U$, denoted $\dU$, is the number of unitaries $U_i$ (or \emph{layers}) in this sequence. In this example, the depth is $d$. 

We can associate a directed graph $G_U$ with the circuit $U$, with vertices for each input qubit, vertices for each output qubit, and a vertex for each gate. For each wire in the circuit, we add directed edges to the graph corresponding to the path of the wire from input to output through gates.
\begin{definition}[Lightcone]
    For a circuit $U$ on $n$-qubits, and a qubit $i\in [n]$, the \emph{(backwards) lightcone} of $i$, denoted $\blc_U(i)$ is the subset of all qubits $j\in [n]$ whose input vertex is connected to the output vertex $i$ in the directed graph $G_U$ associated with $U$. When $U$ is clear from context, we denote $\blc(i) = \blc_U(i)$. For each subset $S\subseteq[n]$, we define its lightcone as $\blc(S) := \bigcup_{i\in S}\blc(i)$. Furthermore we define  $\cktblc[U]{S}$ as the sub-circuit on $\blc_U(S)$ induced by the directed paths from input vertices in $\blc_U(S)$ to output vertices in $S$. We define the \emph{forward} lightcone $\flc(S)$ of input qubits $S$ similarly, but considering paths in the reverse direction. Or equivalently, $\flc_U(S) = \blc_{U^\dagger}(S)$. We define the \emph{blowup} of $U$ as the minimum integer $B$ such that for each $S\sse [n]$, $|\blc_U(S)|, |\flc_U(S)| \leq B \cdot |S|$.
\end{definition}
 
\begin{fact} \label{fact:lightcone-on-state}
    Consider the circuit $U$ on $n$ qubits, and some subset $S\subseteq[n]$. For any $n$-qubit quantum state $\rho$
    \begin{align}
        \tr_{S^c} U \rho U^\dagger = \tr_{S^c} U (\rho_{\blc(S)} \otimes I_{\blc(S)^c})U^\dagger = \tr_{\blc(S) \setminus S} \cktblc[U]{S}\rho_{\blc(S)}\cktblc[U]{S}^\dagger \label{eq:state-lc-def}.
    \end{align}
                \end{fact}
\begin{proof}
    $\cktblc[U]{S}$ is the sub-circuit induced by the backwards lightcone of $S$, and let $W$ be the remaining circuit such that $U =  \pbra{W \otimes I_S}\pbra{\cktblc[U]{S} \ot I_{\blc(S)^c}}$. Now since $W$ does not act on the qubits of $S$, we have that
    \begin{align*}
        \tr_{S^c} U \rho U^\dagger &= \tr_{S^c}  \pbra{W \otimes I_S}\pbra{\cktblc[U]{S} \ot I_{\blc(S)^c}} \  \rho \ \pbra{\cktblc[U]{S}^\dagger \ot I_{\blc(S)^c}} \pbra{W^\dagger \otimes I_S}\\
        &= \tr_{S^c}  \pbra{\cktblc[U]{S} \ot I_{\blc(S)^c}} \  \rho \ \pbra{\cktblc[U]{S}^\dagger \ot I_{\blc(S)^c}} \\
        &= \tr_{\blc(S) \setminus S} \sbra{  \tr_{\blc(S)^c} \pbra{\cktblc[U]{S} \ot I_{\blc(S)^c}} \  \rho \ \pbra{\cktblc[U]{S}^\dagger \ot I_{\blc(S)^c}} }\\
        &= \tr_{\blc(S) \setminus S} \sbra{ \cktblc[U]{S} \  \tr_{\blc(S)^c} \sbra{ \rho} \ \cktblc[U]{S}} = \tr_{\blc(S) \setminus S} \sbra{ \cktblc[U]{S} \  \rho_{\blc(S)} \ \cktblc[U]{S}}
    \end{align*}
\end{proof}
For a unitary implemented by a quantum circuit with one and two-qubit gates, the size of the lightcones are restricted by the circuit depth.
\begin{fact}
    For any $n$-qubit unitary $U$, and any $S\subseteq[n]$, $|\blc_{U}(S)| \leq 2^\dU \cdot|S|$. 
\end{fact}

\begin{fact}\label{fact:lc-preserves-product}
    Let $\rho, \rho'$ be $n$-qubit quantum states and $U\in \unitaries[n]$ such that $\rho' = U \rho U^\dagger$. For any $S, T \sse [n]$, if both $\blc(S) \cap \blc(T) = \emptyset$ and $\rho_{\blc(ST)} = \rho_{\blc(S)} \ot \rho_{\blc(T)}$, then $\rho'_{ST} = \rho'_S \ot \rho'_T$.
\end{fact}
\begin{proof}
    By \Cref{fact:lightcone-on-state}, we have that 
    \begin{align*}
        \rho_{ST} = \pbra{U \rho' U^\dagger}_{ST} &= \pbra{U_{\blc(S\cup T)} \rho_{\blc(S\cup T) }U_{\blc(S\cup T)}^\dagger}_{ST}
    \end{align*}
    Now since $\blc(S) \cap \blc(T) = \emptyset$, we have that $U_{\blc(ST)} = U_{\blc(S)} \ot U_{\blc(T)}$. Furthermore, since $\rho_{\blc(ST)} = \rho_{\blc(S)} \ot \rho_{\blc(T)}$ we can rewrite the above as
    \begin{align*}
        \pbra{U_{\blc(S)} \rho_{\blc(S)} U^\dagger_{\blc(S)}}_S \ot \pbra{U_{\blc(T)} \rho_{\blc(T)} U^\dagger_{\blc(T)}}_T = \rho_S \ot \rho_T.
    \end{align*}
\end{proof}
\subsection{Entropy and correlation measures}
\begin{definition}[von Neumann entropy / entanglement entropy]
    \begin{align*}
        \EE(\rho) := -\tr \rho \log \rho.
    \end{align*}
\end{definition}
\begin{definition}[Quantum mutual information]
    \begin{align*}
        \MI[\rho]{A}{B} := \EE(\rho_A) + \EE(\rho_B) - \EE(\rho_{AB})
    \end{align*}
    For an $n$-qubit state $\rho$, the \QMI between two regions takes a real value in $[0,n]$.
\end{definition}
The \QMI between $A$ and $B$ is a measure of the shared information, or correlations between regions $A$ and $B$. We note that this includes classical correlations. The state $\frac{1}{2} (\proj{00} + \proj{11})$ is a classical mixture of product states $\ket{0}\otimes \ket{0}$ and $\ket{1}\otimes\ket{1}$, yet the \QMI between the two qubits is maximal $\MI{1}{2} = 1$. 
Since \EEnt is invariant under unitary transformation, we have that the \QMI $\MI{A}{B}$ is also invariant under local unitary transformations on $A$ and $B$ in tensor product.
\begin{align}
    \MI[\rho]{A}{B} = \MI[U\rho U^\dagger]{A}{B}, && \text{for } U = U_A \otimes U_B \otimes I_{(A\cup B)^c} \label{eq:mi-local-unitary-invariance}
\end{align}
Furthermore, \QMI is nonincreasing under partial trace \cite{tqi-watrous}. So for subsets of qubits $A\subseteq A', B \subseteq B'$  of quantum state $\rho$, we have  $\MI[\rho]{A}{B} \leq \MI[\rho]{A'}{B'}$. 

\subsection{Stabilizer states}
\begin{definition}[Pauli group]
    The $n$-qubit Pauli group, denoted $\paulis[n]$, contains elements that are a tensor product of $n$ single-qubit Pauli matrices. 
    \[\paulis[n] := \cbra{i^b \sigma_1 \ot \dots \ot \sigma_n : b \in \{0, 1, 2, 3\},\  \sigma_1, \dots, \sigma_n \in \{I, X, Y, Z\} }\]
    Here, $I$ is the identity and $X, Y, Z$ are the single-qubit Pauli matrices.
    \begin{align*}
        I = \begin{pmatrix}
            1 & 0 \\ 0 & 1 
        \end{pmatrix}, \quad
        X = \begin{pmatrix}
            0 & 1 \\ 1 & 0
        \end{pmatrix},
        \quad
        Y = \begin{pmatrix}
            0 & -i \\ i & 0
        \end{pmatrix},
        \quad
        Z = \begin{pmatrix}
            1 & 0 \\ 0 & -1
        \end{pmatrix}
    \end{align*}
\end{definition}

\begin{definition}[stabilizer group]
    \emph{A stabilizer group} is an abelian subgroup of the Pauli group $\paulis[n]$. That is $\stabs \sse \paulis[n]$ is a stabilizer group if all elements of $\stabs$ commute with each other. A \emph{generating set} of $\stabs$ is a subset $\stabgens \sse \stabs$ such that each element of $\stabs$ can be written as a product of elements in $\stabgens$. The dimension of $\stabs$, denoted as $|\stabs| \in \{0, 1, \dots, n\}$ is the size of a minimal generating set. The total number of elements in $\stabs$ is $2^{|\stabs|}$.
\end{definition}

\begin{definition}[stabilizer state]
    For a stabilizer group $\calS$, the corresponding \emph{stabilizer state} is $\sstate := \frac{1}{2^n} \sum_{\sigma \in \stabs} \sigma$.
\end{definition}
Note that this definition of stabilizer states allows for mixed quantum states. The dimension of the stabilizer group is indicative of the purity of the state. In the case where $|\stabs| = n$, the state $\sstate$ is pure, that is $\rho = \proj{\psi}$ for some pure state $\ket{\psi}$. More generally, we have that 
\begin{align*}
    \EE(\sstate) = n -  |\stabs|
\end{align*}
After tracing out a subsystem $A^c \subset [n]$ of a stabilizer state, the remaining state is
\begin{align}
    \sstate_A = \tr_{A^c} \sstate = \frac{1}{2^n} \sum_{\sigma \in \stabs} \tr_{A^c} \sigma = \frac{1}{2^{|A|}} \sum_{\sigma \in \stabs_A} \sigma   \label{eq:stab-rdm}
\end{align}
Where $\stabs_A := \cbra{\sigma_A \in \paulis[|A|] :  \sigma_A\ot I^{\ot |A^c|} \in \stabs}$. Note that $\stabs_A$ is also an abelian subgroup of $\paulis[|A|]$ and so it is also a stabilizer group, and $\sstate_A$ a stabilizer state.

\begin{fact}\label{fact:stab-overlap-powerof2}
    For any two stabilizer states $\rho, \rho$, their overlap $\tr(\rho\sigma)$ is a power of $\frac{1}{2}$. 
\end{fact}
\begin{proof}
    Let $\stabs$ and $\stabs'$ be the stabilizer groups of $\rho$ and $\rho$ respectively. 
    \begin{align*}
        \tr(\rho \sigma) = \frac{1}{2^{2n}}\sum_{\substack{\sigma \in \stabs, \sigma' \in \stabs'}} \tr(\sigma \sigma') = \frac{1}{2^n} \abs{\stabs \cap \stabs'}
    \end{align*}
    Since $\stabs \cap \stabs'$ is also an abelian subgroup of $\paulis$, $\abs{\stabs \cap \stabs'} = 2^k$ for some $k \in \{0, 1, \dots, n\}$. 
\end{proof}

\section{Mutual information-based lower bound} \label{sec:robust-lb}
In this section we prove \Cref{thm:robust-cliffqnc-lb}, by proving a stronger version as \Cref{thm:robust-lb-with-params}. The proof of this theorem is relatively clean. We provide an intuitive proof overview by-example. 

\subsection{Technical overview by example: biased CAT state} For an intuitive exposition, let's consider the biased CAT state $\ket{\bCATn[n]{\gamma}} = \sqrt{\gamma}\ket{0^n} + \sqrt{1-\gamma}\ket{1^n}$. Suppose for the sake of contradiction that there is a constant-depth circuit $U$ that maps a stabilizer state $\ket{\phi_0}$ to~$\ket{\bCAT{\gamma}}$.

Our lower bound can be broken up into two main steps. We show the following:
\begin{enumerate}
    \item \textbf{Stabilizer states have integer mutual information.} 
    On the other hand, we see that for the biased CAT state, the mutual information between any two disjoint regions $A \cup B \subset [n]$, is $H(\gamma) = -\gamma\log \gamma - (1-\gamma) \log(1-\gamma)$. Thus when $\gamma \in (0,1/2)$, it does not have integer mutual information. This is one way to see that the biased CAT state for $\gamma \in (0, 1/2)$ is not a stabilizer state.
    \item \textbf{A low depth circuit preserves the range of mutual information of a state}. If $U$ is a constant-depth circuit, then there exists a pair of qubits $i,j\in [n]$ with disjoint backwards-forwards lightcones as shown in the figure below.

    \begin{center}
        \begin{tikzpicture}[x=1cm, y=1cm]
        
\def\nQ{28}
        \def\nD{3}

\def\wireSep{0.14cm}         \def\gateSep{.3cm}         \def\gateWidth{0.15cm}       \def\gateHeight{.2cm}      

\foreach \q in {1,...,\nQ} {
    \draw (\gateWidth, -\q*\wireSep) -- (\nD*\gateSep +  \gateWidth, -\q*\wireSep);
    }

\foreach \d in {1,...,\nD} {
    \foreach \q in {1,...,\numexpr\nQ-1\relax} {
        \ifodd\q
        \ifodd\d
            \node [shape=rectangle,
            minimum width=\gateWidth,
            minimum height=\gateHeight,
            anchor=center,
            inner sep=0pt,
            outer sep=0pt,
            draw,
            fill=white] at  ({\d*\gateSep}, {- \q*\wireSep - 0.5*\wireSep})  {};
        \fi
        \else
        \unless\ifodd\d
            \node [shape=rectangle,
            minimum width=\gateWidth,
            minimum height=\gateHeight,
            anchor=center,
            inner sep=0pt,
            outer sep=0pt,
            draw,
            fill=white] at  ({\d*\gateSep}, {- \q*\wireSep - 0.5*\wireSep})  {};
        \fi
        \fi
    }
    }

    \node at (-5* \gateSep, -0.5*\nQ * \wireSep) {$\ket{\phi_0}$};
    \node at (5* \gateSep + \nD*\gateSep, -0.5*\nQ * \wireSep) {$\ket{\phi_1}$};
\node[right=0.2cm] (outi) at (\nD*\gateSep+1,-5*\wireSep) {$i$};
        
\draw[fill=blue!10, draw=blue, opacity=0.5]
        (0.2, {-5.5*\wireSep + \nD*\wireSep}) coordinate (T) -- ({\nD*\gateSep + \gateWidth}, {-5*\wireSep}) -- (0.2, {-5.3*\wireSep - \nD*\wireSep}) coordinate (B); \draw[decorate,decoration={brace, amplitude=0.2cm, mirror, raise=0.5pt}] ($(T)-(0.2cm,0)$) -- ($(B)-(0.2cm,0)$) node[midway, left=0.2cm]{$\blc(i)$} ;

\draw[fill=orange!10, draw=none, opacity=0.5]
        (\nD*\gateSep + \gateWidth, -0.1) coordinate (TT)   -- (T) --(B) -- (\nD*\gateSep + \gateWidth, {-4.5*\wireSep - 2*\nD*\wireSep}) coordinate (BB);
        \draw[draw=orange, dashed] (TT) -- (T);
        \draw[draw=orange, dashed] (BB) -- (B);

        \draw[decorate,decoration={brace, amplitude=5pt, raise=0.5pt}] ($(TT)+(0.3cm,0)$) -- ($(BB)+(0.3cm,0)$) node[midway, right=.3cm]{$\dlc{i}$} ;

        \node[right=0.2cm] (outj) at (\nD*\gateSep+1,-22*\wireSep) {$j$};

\draw[fill=blue!10, draw=blue, opacity=0.5]
        (0.2, {-22*\wireSep + \nD*\wireSep}) coordinate (Tb) -- ({\nD*\gateSep + \gateWidth}, {-22*\wireSep}) -- (0.2, {-23*\wireSep - \nD*\wireSep}) coordinate (Bb); \draw[decorate,decoration={brace, amplitude=0.2cm, mirror, raise=0.5pt}] ($(Tb)-(0.2cm,0)$) -- ($(Bb)-(0.2cm,0)$) node[midway, left=0.2cm]{$\blc(j)$} ;

\draw[fill=orange!10, draw=none, opacity=0.5]
        (\nD*\gateSep + \gateWidth,{-22.5*\wireSep + 2*\nD*\wireSep} ) coordinate (TbT)    -- (Tb) --(Bb) -- (\nD*\gateSep + \gateWidth, {-22.5*\wireSep - 2*\nD*\wireSep}) coordinate (BbB);
        \draw[orange, dashed] (TbT) -- (Tb);
        \draw[orange, dashed] (BbB) -- (Bb);

        \draw[decorate,decoration={brace, amplitude=0.2cm, raise=0.5pt}] ($(TbT)+(0.3cm,0)$) -- ($(BbB)+(0.3cm,0)$) node[midway, right=0.2cm]{$\dlc{j}$} ;
        
        \end{tikzpicture}
    \end{center}
    Since mutual information cannot increase by applying local channels to the two regions, the mutual information between $i$ and $j$ in the state $\phi_1$ is at most the mutual information between its backwards lightcones on the initial state $\phi_0$, so $\MI[\phi_1]{i}{j} \leq \ \MI[\phi_0]{\blc(i)}{\blc(j)}$.
    Furthermore, we can use the same reasoning to see that mutual information between $\blc(i)$ and $\blc(j)$ on $\phi_0$ is at most the mutual information between the double lightcone regions $\dlc{i}, \dlc{j}$ in $\phi_1$.
    \begin{align*}
        \MI[\phi_1]{i}{j} \leq \ \MI[\phi_0]{\blc(i)}{\blc(j)} \leq\  \MI[\phi_1]{\flc(\blc(i))}{\flc(\blc(j))}.
    \end{align*} 
\end{enumerate}
Therefore, if $\ket{\phi_1} = \ket{\bCAT{\gamma}}$, since the $\gamma$-biased CAT state has mutual information $H(\gamma)$ between any two disjoint regions, we have:
    \begin{align*}
        \MI[\phi_0]{\blc(i)}{\blc(j)} =  H(\gamma).
    \end{align*}
But from step 1 above we must have that $\MI[\phi_0]{\blc(i)}{\blc(j)}$ is an integer, and so we arrive at a contradiction if $\gamma \in (0, 1/2)$.
To make our lower bound robust, we show, using the Fannes–Audenaert inequality that for states that are near stabilizer, the mutual information between small regions is also near-integer.

\subsection{Proof of mutual information-based lower bound}
We begin by walking through a sequence of \nameCrefs{lem:mi-lightcones} and \nameCrefs{claim:mi-near-stabilizer} that we use to build our proof.

The proof of \Cref{thm:robust-cliffqnc-lb} follows from two key insights. The first is a sense in which \emph{low depth circuits preserve the range of mutual information} between constant-sized regions. This is encompassed in  \Cref{lem:mi-lightcones} below.

\begin{lemma}\label{lem:mi-lightcones}
    If $\psi = U \psi' U^\dagger$ and $\flc_U(\blc_U(i)) \cap \flc_U(\blc_U(j)) = \emptyset$ then
    \begin{align}
        \MI[\psi]{i}{j} \leq \ \MI[\psi']{\blc(i)}{\blc(j)} \leq\  \MI[\psi]{\flc(\blc(i))}{\flc(\blc(j))}.
    \end{align} 
\end{lemma}
\begin{proof}
    Since $\blc(i) \cap \blc(j)$, the sub-circuit $\cktblc[U]{i,j}$ induced by the backwards lightcone of $\{i,j\}$ has a tensor product structure
    \begin{align}
        \cktblc[U]{i,j} = \cktblc[U]{i} \ot \cktblc[U]{j}
    \end{align} 
    By \Cref{fact:lightcone-on-state} we have that.
    \begin{align}
        \psi_{i,j} = \tr_{\{i,j\}^c} U \psi' U^\dagger = \tr_{\{i, j\}^c} V \psi'_{\flc(i,j)} V^\dagger && \text{for } V = \cktblc[U]{i} \ot \cktblc[U]{j} \otimes I_{(\blc(i)\cup \blc(j))^c}
    \end{align}
    So the mutual information can be written as follows.
    \begin{align}
        \MI[\psi]{i}{j} = \MI[U \psi' U^\dagger]{i}{j} 
        =\MI[V \psi' V^\dagger]{i}{j}   
    \end{align}
    Now, using the fact that quantum mutual information is non-increasing under partial trace, and that quantum mutual information is invariant under local unitaries (\Cref{eq:mi-local-unitary-invariance}) this expression is at most
    \begin{align}
        \leq \MI[V \psi' V^\dagger]{\flc(i)}{\flc(j)} 
        = \MI[\psi']{\flc(i)}{\flc(j)}.
    \end{align}
    We have shown the first inequality in our \nameCref{lem:mi-lightcones} $\MI[\psi]{i}{j} \leq \MI[\psi']{\blc(i)}{\blc(j)}$. The second inequality follows from the same argument. 
\end{proof}

The second key observation is that the class of stabilizer states are \emph{discrete}, and so is their mutual information. 
 
\begin{lemma}\label{lem:stab-mi-integer}
    For each $n$-qubit stabilizer state $\varphi$, and each pair of disjoint subsets of qubits $A, B \subseteq [n]$, if $A\cap B = \emptyset$, then $\MI[\varphi]{A}{B} \in \{0, 1, \dots, n\}$.
\end{lemma}
\begin{proof} 
    Let $\calS$ be the stabilizer group for $\varphi$. As mentioned in \cite{stabs-entanglement}, we can write the entanglement of the state $\varphi$ as $\EE(\varphi) = n - |\calS|$. Here $|\calS|$ is the rank of the stabilizer group $\calS$, the size of the minimal generating set. We first prove this fact for completeness. We can write $\varphi$ as a linear combination of its stabilizer group $\calS$, $\varphi = \frac{1}{2^n} \sum_{s\in \calS} s$.  This is proportional to the projector $\Pi$ onto the $+1$ eigenspace of $\calS$. $\varphi= 2^{|\calS|-n}\Pi$, where $|\calS|$ is the rank of the stabilizer group $\calS$, the size of its minimal generating set. Therefore the entropy of $\varphi$ is $\EE(\varphi) = -\tr(\rho\log \rho) = - \sum_{s \in \calS} 2^{|\calS| -n} \log 2^{|\calS| - n} = n - |\calS|$. 

    We now consider the reduced state on qubits in $A$. We can write 
    \begin{align}\varphi_A = \frac{1}{2^n} \sum_{s \in \calS} \tr_{A^c} s =  \frac{1}{2^{|A|}} \sum_{s\in \calS_A} s.\end{align}
    Where $\calS_A = \cbra{\tr_{A^c} s : s \in \calS : \tr(s)\neq 0}$. Thus $\varphi_A$ is also a stabilizer state with stabilizer group $\calS_A$. So $\EE(\varphi_A) = |A| - |\calS_A|$. Overall, the mutual information is 
    \begin{align}\MI[\varphi]{A}{B} := \EE(\varphi_A) +\EE(\varphi_B) - \EE(\varphi_{AB}) = |A| - |\calS_A| + |B| - |\calS_B| -(|B| -  |\calS_{AB}|) = |\calS_{AB}| - |\calS_A| - |\calS_B|\end{align}
     which is an integer, as desired. Finally, the fact that mutual information is non-negative completes the proof.

\end{proof}

Already we can sketch a proof for lower bounding the circuit depth required to exactly prepare $\psi$. Combining these two \nameCrefs{lem:mi-lightcones}, we have that if $\psi = U\phi^S U^\dagger$ for some stabilizer state $\phi^S$, then $\MI[\psi]{i}{j} \leq \ \ell \leq\  \MI[\psi]{\flc(\blc(i))}{\flc(\blc(j))}$ for some integer $\ell$. Therefore, if  $\psi$ is chosen such that this interval does not contain an integer, we reach a contradiction. Our properties in \Cref{thm:robust-cliffqnc-lb}, do precisely this when the depth (and therefore also lightcone size) is small.

Before proving \Cref{thm:robust-cliffqnc-lb} we first prove the following sequence of claims to make the above argument robust. We now show that for states that are nearby stabilizer, the mutual information between small regions is also near integer.
\begin{claim}\label{claim:mi-near-stabilizer}
    If $\tnorm{\rho - \varphi} \leq \eps$ and $\varphi$ is a stabilizer state, then for each $A, B \in [n]$ disjoint subsets of qubits ($A \cap B = \emptyset$) there exists an integer $k$ such that 
    \begin{align}\abs{\MI[\rho]{A}{B} -k } = \abs{\MI[\rho]{A}{B} - \MI[\varphi]{A}{B} }  \leq 2 \eps \cdot (|A| + |B|) + 3H(\eps). \end{align}
    Where $H(\eps) = -\eps \log \eps - (1-\eps)\log(1-\eps)$ is the binary entropy function.
\end{claim}
\begin{proof}
    For any two $n$-qubit states $\rho, \sigma$,  Fannes–Audenaert inequality says that if the states are close $\tnorm{\rho - \sigma} \leq \eps$ then so is their entanglement entropy $\abs{ \EE(\rho) - \EE(\sigma) } \leq \eps n + H(\eps)$. Now using the definition of mutual information, and the triangle inequality we have that
    \begin{align}
        \abs{\MI[\rho]{A}{B} - \MI[\varphi]{A}{B}}  &= \abs{ \EE(\rho_A) + \EE(\rho_B) - \EE(\rho_{AB}) + \EE(\varphi_A) + \EE(\varphi_B) - \EE(\varphi_{AB})}\\
        &\leq \abs{\EE(\rho_A) - \EE(\varphi_A)} + \abs{\EE(\rho_B) - \EE(\varphi_B)} + \abs{\EE(\rho_{AB})-\EE(\varphi_{AB})}\\
        &\leq \pbra{\eps |A| + H(\eps)} + \pbra{\eps |B| + H(\eps)} + \pbra{\eps (|A|+ |B|) + H(\eps)}\\
        &= 2\eps (|A| + |B|) + 3H(\eps)
        .
    \end{align}
    In the second inequality we used the Fannes–Audenaert inequality and that trace norm is non-increasing under partial trace. Finally since $\varphi$ is a stabilizer state, by \Cref{lem:stab-mi-integer} we have that $\MI[\varphi]{A}{B} = k$ for some integer $k$, completing the proof.
\end{proof}

\begin{claim}\label{claim:lb-mi-dist}
    Suppose $\psi = U \psi' U^\dagger$. If there exists $i,j \in [n]$ with nonintersecting backwards-forwards lightcones $\flc(\blc(i)) \cap \flc(\blc(j)) = \emptyset$ with respect to $U$, then there exists an integer $\ell$ such that
    \begin{align}
        \min_{k\in \ZZ}\abs{ \MI[\psi']{\blc(i)}{\blc(j)}  -k}  \geq \min \cbra{ \MI[\psi]{i}{j} - (\ell-1) \ ,\  \ell - \MI[\psi]{\flc\blc(i)}{\flc\blc(j)} }.
    \end{align}
\end{claim}
\begin{proof}
    Let $\ell$ be the integer such that $\MI[\psi']{\blc(i)}{\blc(j)} \in [\ell-1, \ell)$. Suppose that $\ell$  the closest integer to $\MI[\psi']{\blc(i)}{\blc(j)}$,  $\abs{\MI[\psi']{\blc(i)}{\blc(j)} - \ell } \leq \frac{1}{2}$ then using \Cref{lem:mi-lightcones} we get the following.
    \begin{align}
        \ell- \MI[\psi]{\flc(\blc(i))}{\flc(\blc(j))} \leq \ell - \MI[\psi']{\blc(i)}{\blc(j)} = \min_{k\in \ZZ}\abs{ \MI[\psi']{\blc(i)}{\blc(j)}  -k} 
    \end{align}
    On the other hand, if $\ell - 1$ is the nearest integer to $\MI[\psi']{\blc(i)}{\blc(j)}$, then we have
    \begin{align}
        \MI[\psi]{i}{j} -(\ell -1) \leq \ \MI[\psi']{\blc(i)}{\blc(j)}  -(\ell -1)  =\ \min_{k\in \ZZ}\abs{ \MI[\psi']{\blc(i)}{\blc(j)}  -k}.
    \end{align}
    Since at least one of these equations must hold, we have proved the \nameCref{claim:lb-mi-dist}.
\end{proof}

We are now ready to prove the main result of this section.

\begin{lemma}\label{lem:intermediate:robust-lb}
    Suppose $n$-qubit state $\phi$ is prepared by a unitary $U$ acting on a stabilizer state $\phi^S$, $\phi = U \phi^S U^\dagger$. And suppose $\psi$ is an $n$ qubit state and $i\neq j \in [n]$ satisfying:
    \begin{enumerate}
        \item $\MI[\psi]{i}{j} \geq \alpha$, and
        \item $\MI[\psi]{\dlc{i}}{\dlc{j}} \leq \beta$
    \end{enumerate}
    for some $0 \leq \alpha \leq \beta \leq 1$. Let $\eps = \tnorm{\psi - \phi}$. If  $\eps \leq \pbra{\frac{\min\cbra{\alpha, 1- \beta}}{1 + 3 e}}^{\ln(4)}$, then $\dU \geq \Omega(\log(1/\eps))$.
\end{lemma}
\begin{proof}
    Let $\psi'$ the state such that $\psi = U \psi' U^\dagger$. 
    \Cref{claim:lb-mi-dist} tells us that 
    \begin{align}
        \min_{k\in \ZZ}\abs{ \MI[\psi']{\blc(i)}{\blc(j)}  -k} \geq \min \cbra{\alpha, 1- \beta}.
    \end{align}
    But \Cref{claim:mi-near-stabilizer}, and the fact that $|A| +|B| \leq 2^{\dU+1}$ tells us that 
    \begin{align}
        \min_k \abs{\MI[\psi']{A}{B} -k }  \leq 2 \eps \cdot 2^{\dU+1} + 3H(\eps).
    \end{align}
    So
    \begin{align}
        2^{\dU+2} \eps + 3H(\eps) &\geq \min\cbra{\alpha, 1-\beta} \label{eq:forkfornew}\\
         \dU&\geq \log \pbra{\frac{1}{\eps} \pbra{\min\cbra{\alpha, 1-\beta} - 3 H(\eps)}} -1
    \end{align}
    Using the fact that $H(\eps) \leq e \cdot \eps^{1/\ln 4}$, we have
    \begin{align}
        \dU&\geq \log \pbra{\frac{1}{\eps} \pbra{\min\cbra{\alpha, 1-\beta} - 3 e\cdot  \eps^{1/\ln 4}}} -1 .
    \end{align}
    Suppose, as in the \namecref{lem:intermediate:robust-lb} statement, that $\eps \leq \pbra{\frac{\min\cbra{\alpha, 1- \beta}}{1 + 3 e}}^{\ln(4)}$, and so  $(1+3e) \eps^{1/\ln 4} < \min\{\alpha, 1- \beta\}$. Further rearranging we see that $\min\{a, 1- \beta\}  - 3e \cdot \eps^{1/\ln 4} > \eps^{1/\ln 4}$. So plugging this into the equation above,
    \begin{align}
        \dU \geq \log \pbra{\eps^{1/\ln 4 - 1}}- 1 = (1 - 1/\ln 4) \log \eps - 1 = \Omega(\log(1/\eps)).
    \end{align}
\end{proof}
\begin{theorem}[Formal version of \Cref{thm:robust-cliffqnc-lb}, with auxilary state]\label{thm:robust-lb-with-params}
    Suppose $\psi$ is an $n$ qubit state satisfying:
    \begin{enumerate}
        \item \label{item:lem:MI-qubit-lowerbound} For each $i\neq j \in [n]$, $\MI[\psi]{i}{j} \geq \alpha$, and
        \item \label{item:lem:MI-subset-upperbound} For each $A, B\subseteq [n]: A\cap B = \emptyset$ and  $|A|, |B| < s$, $\MI[\psi]{A}{B} \leq \beta$.
    \end{enumerate}
    For some $k < \alpha \leq \beta < k+1$, and integers $k\geq 0, s \geq 1$. Suppose $\phi$  is an $n$-qubit state prepared by an $(n+a+a')$-qubit unitary $U$ as $\phi\ot \varphi\ot \proj{0^{a'}} = U \phi^S U^\dagger$ for an $(n+a+a')$-qubit stabilizer state $\phi^S$ and an $a$-qubit auxiliary state $\varphi$. Let $\eps = \tnorm{\psi - \phi}$. At least one of the following is true:
    \begin{enumerate}[(a)]
        \item \label{item:lem:Heps-bound} $\eps > \pbra{\frac{\min\cbra{\alpha, 1- \beta}}{1 + 3 e}}^{\ln(4)}$, or
        \item \label{item:lem:logdepth} $\dU \geq \frac{1}{4}\log(s)$, or 
        \item \label{item:lem:log-anc-bound} $\dU \geq \log(n/a)$, or 
        \item \label{item:lem:logeps-bound} $\dU \geq \Omega(\log(1/\eps))$. 
    \end{enumerate} 
\end{theorem}
\noindent In the special case where $\alpha, \beta$ are constants, we get a depth lower bound of $\Omega\pbra{\log\min\cbra{n/a, s,1/\eps}}$, implying \Cref{thm:robust-cliffqnc-lb}.

\begin{proof}[Proof of \Cref{thm:robust-lb-with-params}]
    We will prove this by contradiction assuming that \Cref{item:lem:Heps-bound,item:lem:logdepth,item:lem:log-anc-bound,item:lem:logeps-bound} are all false. The assumption that $\dU < \frac{1}{4} \log s$ (\Cref{item:lem:logdepth} being false) tells us that for each $i \in [n]$, there exists some $j\neq i \in [n]$ such that $\blc(\flc(i)) \cap \blc(\flc(j)) = \emptyset$. This is because $\blc(\flc(i)) \cap \blc(\flc(j)) = \emptyset \iff j \notin \blc(\flc(\blc(\flc(i))))$, and $|\blc(\flc(\blc(\flc(i))))| \leq 2^{4 \cdot \dU} < s\leq n$. So there exists some $j \in [n] \setminus \blc(\flc(\blc(\flc(i))))$. 
    
    Let $C := [n]$ denote the set of computational qubits, and $A := \{n+1, \dots, a\}$ denote the set of auxiliary qubits in state $\varphi$, and $A' := \{n+a+1, \dots, n+a + a'\}$ the final set of auxiliary qubits in state $\proj{0^{a'}}$. Next we use the fact that mutual information is additive with respect to tensor product: $\MI[\rho_A\ot \rho_B]{A_1 \cup B_1}{A_2\cup B_2} = \MI[\rho_A]{A_1}{A_2} + \MI[\rho_B]{B_1}{B_2}$, where $A_1, A_2\sse A$ and $B_1, B_2 \sse B$. Therefore,
    \begin{equation}
        \begin{aligned}
            \MI[\psi\ot \varphi]{\blc(\flc(i))}{\blc(\flc(j))} 
            &= \MI[\psi]{\dlc{i}\cap C}{\dlc{j}\cap C}  + \MI[\varphi]{\dlc{i}\cap A}{\dlc{j}\cap A} \\ 
            & \hspace{10em} + \MI[\proj{0^a}]{\dlc{i}\cap A'}{\dlc{j}\cap A'}
        \end{aligned}
    \end{equation}
    The assumption that $\dU < \frac{1}{4} \log s$ implies that $|\blc(\flc(i))|, |\blc(\flc(j))| < s$, so \Cref{item:lem:MI-subset-upperbound} tells us that the first term is $\MI[\psi]{\dlc{i}\cap C}{\dlc{j} \cap C} \leq \beta$. Since $\proj{0^a}$ is a product state, the third term is zero  $\MI[\proj{0^a}]{\dlc{i}\cap A'}{\dlc{j}\cap A'}~=~0$. Furthermore, since we assume that $\dU < \log (n/a)$ (\Cref{item:lem:log-anc-bound} being false), then $|\dlc{A}| \leq a \cdot 2^\dU <n$ therefore $C \not\sse \dlc{A}$, so we can choose an $i\in [n]$ such that $\dlc{i} \cap A = \emptyset$ and again the second term $\MI[\varphi]{\dlc{i}\cap A}{\dlc{j}\cap A}$ is zero. Therefore, we have that  $\MI[\psi\ot \varphi]{\blc(\flc(i))}{\blc(\flc(j))} \leq \beta$. Furthermore, \Cref{item:lem:MI-qubit-lowerbound} tells us that $\MI[\psi\ot \varphi]{i}{j} \geq \alpha$. Now since we assumed that  $\eps \leq \pbra{\frac{\min\cbra{\alpha, 1- \beta}}{1 + 3 e}}^{\ln(4)}$ (\Cref{item:lem:Heps-bound} being false),  \Cref{lem:intermediate:robust-lb} tells us that $\dU \geq \Omega(\log(1/\eps))$, contradicting our assumption that \Cref{item:lem:logeps-bound} is false.

\end{proof}

\section{Lower bounds for quantum error correcting codes}

In this section we identify some subspaces such that $\altCQ[1]$ circuits can not prepare \emph{any} of the states in the subspace, even approximately. In particular, we consider \emph{locally-defined subspaces} which we can interpret as codespaces of codes with local-checks, or groundspaces of local Hamiltonians with topological order. 

We show that any local codespace that is not approximately a perturbed stabilizer code --- a stabilizer code rotated by some constant depth ($\QNCZ$) circuit --- cannot be prepared by $\altCQ[1]$ circuits. This allows us to prove lower bounds for the Feynman-Kitaev history state for preparing the CAT state, and for states in a codespace with dimension that is not a power of 2. 

Our key ingredient is proving an infectiousness property which we state informally below.

\paragraph{(Robust) infectiousness of $\altCQ[1]$ in codes:} \textit{For each local code $\calC$ with good distance, if ${\altCQ[1]}$ can (approximately) prepare a single state in $\calC$ then $\calC$ is (approximately) a perturbed stabilizer code.}

\bigskip
\noindent 
We provide two separate proofs to show the \emph{exact} infectiousness property and the \emph{robust} infectiousness property. The proof to show exact infectiousness is more simple, but can only be used to prove lower bounds for exact state preparation. We need the robust infectiousness property in order to prove robust lower bounds, meaning lower bounds against even \emph{approximately} preparing a state. The proof of the exact version is in \Cref{ssec:exact-infectiousness} and the robust version is in \Cref{ssec:robust-infectiousness}.

\subsection{Setup and notations}

Below we define and introduce convenient notation for local-subspaces.

\begin{definition}[Local subspace/codespace, {$\calC_\ell(\phi)$}]
For an $n$-qubit pure quantum state $\ket{\psi}\in \CC^{2^n}$, we define $\calC_\ell(\psi)\sse \CC^{2^n}$ to be the subspace spanned by all pure states that are $\ell$-locally equivalent to $\psi$.
    \begin{align}
    \calC_\ell(\psi) := \textsf{span}\{\ket{\phi}:\phi_A = \psi_A \ \text{for each} \ A \sse [n] \ \text{with} \ |A|\leq \ell\}
\end{align}
\end{definition}

Note that all states in this subspace are not necessarily $\ell$-locally equivalent to $\phi$\footnote{Take for example the CAT state $\ket{CAT} = \frac{1}{\sqrt{2}}( \ket{0^n} + \ket{1^n})$. Note that for $\ell < n$, it is locally equivalent to the state $\ket{CAT^-} = \frac{1}{\sqrt{2}}(\ket{0^n} - \ket{1^n})$, though their linear combination $\frac{1}{\sqrt{2}}(\ket{CAT} + \ket{CAT^-}) = \ket{0^n}$ is not locally equivalent to $\ket{CAT}$}. 
We will interpret such subspaces as quantum error correcting codes, and sometimes refer to them as codespaces. 
\begin{definition}[Code distance]
    The \emph{distance} of a subspace $\calC$, denoted $d(\calC)$, is defined as the maximum integer $d$ such that for each region $A\sse [n]$ of size $|A|< d$, there exists a CPTP map $\textsf{Rec}$ that can recover any state in $\calC$
\begin{align}
    \textsf{Rec}(\Tr_A \phi) = \phi, && \text{for each} \ \ket{\phi} \in \calC
\end{align}
For notational convenience, the distance of $\calC_\ell(\phi)$ is denoted $d_\ell(\phi) = d(\calC_\ell(\phi))$. 
\end{definition}
Note that when $\phi$ is $\ell$-locally unique, $\calC_\ell(\phi) = \{\ket{\phi}\}$, so it has maximal distance $n+1$ as the channel that always outputs $\ket{\phi}$ is a valid recovery map.
We also introduce new notation for local stabilizer codes.
\begin{definition}[Local stabilizer code, {$\Cstab[\ell](\phi)$}]
    For each stabilizer state $\ket{\phi}$, and integer $\ell \geq 1$, we define $\Cstab[\ell](\phi)$ to be the subspace of states that are stabilized by the $\ell$-local stabilizers of $\phi$.
\end{definition}
Note that for any stabilizer state $\ket{\phi}$ the states that are locally equivalent to $\ket{\phi}$ will also be stabilized by its local stabilizers so $\calC_\ell(\phi) \sse \Cstab[\ell](\phi)$, but in general $\Cstab[\ell](\phi)$ may not be a subset of~$\calC_\ell(\phi)$.

We will make use of the following well-known fact which shows that quantum codes that have distance $d$ have codestates that are $\ell$-locally equivalent for each $0 < \ell < d$. A simple proof can be found in Fact 2 of \cite{anshu2020circuit}.
\begin{fact}(Local Indistinguishability)\label{fact:codes-recovery-indistinguishability}
    If quantum code $\calC$ on $n$ qubits is correctable on a region $M\sse [n]$ then the reduced density matrix $\psi_M$ is the same for each $\ket{\psi} \in \calC$.
\end{fact}
We use this fact to show that these codes $\calC_\ell(\psi)$ defined above are more general than codes defined by Hamiltonian ground spaces with sufficiently large distance.

\begin{fact}\label{fact:arbcode-is-lecode}
    For any quantum code $\calC$ defined as the ground space of an $\ell$-local Hamiltonian, if the distance of $\calC$ is greater than $\ell$ then $\calC = \calC_\ell(\phi)$ for each $\ket{\phi} \in \calC$.
\end{fact}
\begin{proof}
    Let $\ket{\phi} \in \calC$. First note that any state $\ket{\phi'}\in \calC_\ell(\phi)$ that is $\ell$-locally equivalent to $\phi$ will have the same energy with respect to the $\ell$-local Hamiltonian because this is just a sum of expectation values of size $\ell$ observables. Therefore, $\calC_\ell(\phi) \sse \calC$. The other direction from \Cref{fact:codes-recovery-indistinguishability}, which tells us that the codewords of $\calC$ are locally-equivalent on any subset of size $< d(\calC)$. Therefore, since $d> \ell$, they are $\ell$-locally equivalent. This implies that $\calC \sse \calC_\ell(\phi)$. 
\end{proof}
While the set of all pure states that are $\ell$-locally equivalent to $\phi$ is not, in general, a subspace \Cref{fact:codes-recovery-indistinguishability} shows that when the distance of $\calC_\ell(\phi)$ is greater than $\ell$, it is in fact a subspace. 

Local indistinguishability also allows us to prove the following useful fact.
\begin{fact}\label{fact:codes-equal-if-dgel}
    For each $n$-qubit quantum state $\ket{\psi}$ and integer $\ell \geq 1$, if $d_\ell(\psi) > \ell$, then for each integer $\ell'$ with $\ell \leq \ell' < d_\ell(\psi)$, we have $\calC_\ell(\psi) = \calC_{\ell'}(\psi)$.
\end{fact}
\begin{proof}
    Since $\ell'\geq \ell$ it is trivial to see that $\calC_\ell(\psi) \supseteq \calC_{\ell'}$ since each state that's $\ell'$-locally equivalent to $\psi$ is also $\ell$-locally equivalent to $\psi$. Now, to show that $\calC_\ell \subseteq \calC_{\ell'}$ we use \Cref{fact:codes-recovery-indistinguishability}. Suppose for the sake of contradiction that there exists a $\ket{\phi} \in \calC_\ell(\psi) \setminus \calC_{\ell'}(\psi)$. Since $\ket{\phi}\notin \calC_{\ell'}(\psi)$, there exists a region $M \sse [n]$ of size $|M| \leq \ell'< d_\ell(\psi)$ such that $\psi_M \neq \phi_M$. But since both $\ket{\phi}, \ket{\psi} \in \calC_\ell(\psi)$, \Cref{fact:codes-recovery-indistinguishability} tells us that they are $d_\ell(\psi)$-locally equivalent, so since $|M| < d_\ell(\psi)$ this implies $\phi_M = \psi_M$, a contradiction.
\end{proof}

A stabilizer code can be defined as the common $+1$ eigenspace of local commuting projectors $\Pi_g = \frac{1}{2} (I +g)$ for each generator $g$ in a generating set of the stabilizer group.
Quantum codes that are defined by commuting local projectors have codestates with restricted entanglement structure, as was shown by Bravyi, Poulin and Terhal in \cite{BPT}. 
\begin{lemma}[Disentangling Lemma \cite{BPT}]
    Consider pairwise commuting projectors $\Pi_1, \dots \Pi_m$. Let the codespace $\calC$ be the $1$-eigenspace of $\Pi= \prod_i \Pi_i$. For each subset of qubits $M \sse [n]$, let $\partial(M)$ be the set of all $a\in [n]$ such that there is some $\Pi_i$ acting nontrivially on both $a$ and some $b \in M$, and let $\partial^+(M) = \partial(M) \setminus M$. If $M$ and $\partial^+(M)$ are both correctable regions, then there exists a unitary operator $U_{\partial(M)}$ acting only on $\delta(M)$ such that
    \begin{align}
        U_{\partial(M)} \Pi U_{\partial(M)}^\dagger = \proj{\phi_M} \otimes \Pi_{M^c}.
    \end{align}
\end{lemma}

In \Cref{ssec:exact-code-lb} we use a variant of the Disentangling Lemma and \Cref{cor:infectiousness-ham-technical} to prove circuit lower bounds for preparing code states of codes with many qubits that are pairwise correlated.

\subsection{Proof of exact infectiousness}\label{ssec:exact-infectiousness}

In the following theorem we show that for a local code that's defined as the groundspace of a local Hamiltonian with sufficiently large distance, if one of the codewords can be prepared exactly in low depth starting from a stabilizer state, then the code is actually a local stabilizer code that is perturbed by a low-depth circuit. 

\begin{theorem}\label{thm:le-code-contained-in-commproj-code}
    Suppose $\ket{\psi} = U \ket{\phi}$ for stabilizer state $\ket{\phi}$ and circuit $U$ with blowup $B$. For each integer $\ell \geq 1$, if $B < \pbra{d_\ell(\psi)/\ell}^{1/2}$
    then $\calC_\ell(\psi) = U \Cstab[B \ell](\phi)$.
            \end{theorem}

Surprisingly, \Cref{thm:le-code-contained-in-commproj-code} allows us to reason about the structure of the entire codespace $\calC_\ell(\psi)$, even if just one of the states $\ket{\phi}\in \calC_\ell(\psi)$ is prepared in low-depth from stabilizer.

We emphasize that the condition that the code is defined by locally equivalent states is more general than requiring a code to be the ground space of a local Hamiltonian. Using \Cref{fact:arbcode-is-lecode} below we get the following Corollary of \Cref{thm:le-code-contained-in-commproj-code}.
\begin{corollary}\label{cor:code-subseteq-commuting}
    Let $\calC$ be the groundspace of an $\ell$-local Hamiltonian with distance $d$. If there exists a state $\ket{\psi}\in \calC$ such that $\ket{\psi} = U \ket{\phi}$ for a stabilizer state $\ket{\phi}$ and a quantum circuit $U$ with blowup $B< \pbra{d/\ell}^{1/2}$, then $\calC = U \calC_S$ for some stabilizer code $\calC_S$ with $O(\ell)$-local stabilizer generators.
\end{corollary}

We will use the following sequence of claims to prove \Cref{thm:le-code-contained-in-commproj-code}.  For unitary $U$ and subspace $\calC$ we use $U \calC$ to denote $\{U \ket{\varphi} : \ket{\varphi} \in \calC\}$.
\begin{claim} \label{claim:code-blowup-containment}
    Suppose $\ket{\phi_1} = U \ket{\phi_0}$ for circuit $U$ with blowup $B$. Then
    \[ \calC_{B^2 \ell}(\phi_1)  \sse U \calC_{B\ell}(\phi_0) \sse  \calC_{\ell} (\phi_1) \]
    Furthermore, if $d_{\ell}(\phi_1) > B^2 \ell$ then 
    \begin{align*}
        \calC_\ell(\phi_1) = U \calC_{B\ell}(\phi_0).
    \end{align*}
\end{claim}
\begin{proof}
    We begin by proving the first equation. It is sufficient for us to prove the second containment since the first is an equivalent statement with a change of variables and replacing $U$ with $U^\dagger$. Suppose for the sake of contradiction that there is some $\ket{\varphi} \in \calC_{B\ell}(\phi_0)$ such that $U \ket{\varphi} \notin \calC_{\ell}(\phi_1)$. Therefore, there is some region $A \sse [n]$ of size $|A| \leq \ell$ such that $(U \varphi U^\dagger)_A \neq (\phi_1)_A = (U \phi_0 U^\dagger)_A$. Using \Cref{fact:lightcone-on-state} we can write this inequality as 
    \begin{align}
        \tr_{A^c} U \pbra{(\phi_0)_{\blc(A)} \ot I_{\blc(A)^c}} U^\dagger \neq \tr_{A^c} U \pbra{\varphi_{\blc(A)} \ot I_{\blc(A)^c}} U^\dagger.
    \end{align}
    Therefore $(\phi_0)_{\blc(A)} \neq \varphi_{\blc(A)}$. Since $|\blc(A)| \leq B |A| \leq B^2\ell$, it follows that $\varphi$ is not $B^2 \ell$-locally equivalent to $\phi_1$. 
    But this contradicts our assumption that $\varphi \in \calC_{B^2\ell}(\phi_0)$. Thus we have shown that $\calC_{B^2 \ell}(\phi_1) \sse U \calC_{B\ell}(\phi_0) \sse \calC_\ell(\phi_1)$.

    If $d_\ell(\phi_1)> B^2 \ell$, then \Cref{fact:codes-equal-if-dgel} implies that $\calC_{\ell}(\phi_1) = \calC_{B^2 \ell}(\phi_1)$, therefore $\calC_\ell(\phi_1) = U \calC_{B\ell}(\phi_0)$.
\end{proof}

\begin{claim}\label{claim:distance-preserved}
    \begin{align}
       \frac{1}{B} d(U \calC) \leq d(\calC) \leq B \cdot d(U \calC) \label{eq:code-evolution-distance-relations}
    \end{align}
    Suppose $\ket{\phi_1} = U \ket{\phi_0}$ for circuit $U$ with blowup $B$. Then 
    \begin{align}
        \frac{1}{B} d_{\ell}(\phi_1) \leq d_{B\ell}(\phi_0) \leq B \cdot d_{B^2 \ell}(\phi_1). \label{eq:distance-bounds-lecodes}
    \end{align}
\end{claim}
\begin{proof}
    Consider some region $A\sse [n]$ of qubits, and its forwards lightcone $\flc_U(A)$ with respect to circuit $U$. If the region $\flc_U(A)$ is correctable for codespace $U \calC$, then we can also correct any error on region $A$ for codespace $\calC$ by evolving the state by $U$, correcting to the codespace $U \calC$, and then evolving back to $\calC$ with $U^\dagger$. This is because $\tr_{\flc(A)} \pbra{U (\psi_{A^c} \ot I_{A}/2^{|A|})U^\dagger} = \tr_{\flc(A)} U \rho U^\dagger$. This follows from \Cref{fact:lightcone-on-state} and setting $S = \flc(A)^c$ and observing that $\blc(S) = \blc(\flc(A)^c) = A^c$. Therefore, if $U\calC$ can correct errors of size $Bt$, then $\calC$ can correct errors of size $t$. Since there is some region $A \sse [n]$ with $|A| = d(\calC)$ that is not correctable on $\calC$, this implies that $\flc(A)$ is not a correctable region on $U \calC$. Therefore, $d(U \calC) < |\flc(A)| \leq B \cdot d(\calC)$, proving \Cref{eq:code-evolution-distance-relations}.
    
    We now prove \Cref{eq:distance-bounds-lecodes}. By \Cref{eq:code-evolution-distance-relations}, we have that $d_{B \ell}(\phi_0) = d(\calC_{B \ell}(\phi_0)) \leq B \cdot d( U \calC_{B\ell}(\phi_0))$. By \Cref{claim:code-blowup-containment}, we have that $ U \calC_{B \ell}(\phi_0) \supseteq \calC_{B^2\ell}(\phi_1)$, therefore, $d(U \calC_{B \ell}(\phi_0)) \leq d(\calC_{B^2\ell}(\phi_1)) = d_{B^2\ell}(\phi_1)$. Combining these expressions it follows that $d_{B \ell}(\phi_0) \leq B \cdot d_{B^2\ell}(\phi_1)$, the second inequality. Furthermore, the first inequality is also implied by the second by a change of variables and considering $U^\dagger$ as $U$.
\end{proof}

\begin{claim}[Infectiousness of stabilizer states]\label{claim:stabcoderound}
    Consider stabilizer state $\rho = \sstate$, with stabilizer group $\stabs$. Let $\Cstab[\ell] (\rho)$ denote the subspace of pure states spanned by the set of pure states that are stabilized by $\stabs_\ell = \{\sigma \in \stabs : |\sigma| \leq \ell\}$, and $\calC_\ell(\rho)$ defined as above. If $d_\ell(\rho) > \ell$, then $\calC_\ell(\rho) = \Cstab[\ell](\rho)$.  
\end{claim}
\begin{proof}
    It is straightforward to show that $\calC_\ell(\rho) \sse \Cstab[\ell](\rho)$, since each state that is $\ell$-locally-equivalent to $\rho$ must also be stabilized by $\stabs_\ell$. We next prove the less-trivial direction $\Cstab[\ell](\rho) \sse \calC_\ell(\rho)$. Suppose for the sake of contradiction that there exists some state in $\Cstab[\ell](\rho)$ that is not in $\calC_\ell(\rho)$. Since  $\Cstab[\ell]$ is spanned by a basis of stabilizer states, we can assume without loss of generality that this state $\phi^S \in \Cstab[\ell] \setminus\calC_\ell(\rho) $ is a stabilizer state. Let $\calT$ be the stabilizer group of $\phi^S$. Since $\phi^S$ is stabilized by $\stabs_\ell$ we can break up $\calT = \stabs_k \cup \calT'$. Because $\phi^S \notin \calC_\ell(\rho)$, it is not locally equivalent to $\rho$ and so there exists a subset of qubits $A$ with $|A|\leq \ell$ such that $\rho_A \neq \phi^S_A$. Since $\phi^S_A = \frac{1}{2^{|A|}} \sum_{\sigma \in \calT_A} \sigma_A$, there must exist a $\sigma \in \calT \setminus \stabs$ that acts nontrivially on at most $\ell$ qubits. Now consider the state $\rho' = \sigma \rho \sigma$. Since $\sigma \in \calT$ and $\calT$ is an abelian subgroup containing $\stabs_\ell$, $\sigma$ commutes with $\sigma_\ell$, all $\ell$-local stabilizers of $\rho$. Therefore $\rho$ and $\rho'$ are $\ell$-locally equivalent. However, since $\sigma \notin \stabs$, we have that $\rho \neq \rho'$. But since $\sigma$ is an $\ell$-local operator, this contradicts the assumption that $d_\ell(\rho) > \ell$. 
\end{proof}

\begin{proof}[Proof of \Cref{thm:le-code-contained-in-commproj-code}]
    By \Cref{claim:distance-preserved}, $d_{B \ell}(\phi) \geq \frac{1}{B} \cdot d_\ell(\psi)$. Since we assume that $B < \pbra{d_\ell(\psi)/\ell}^{1/2}$, rearranging we have $d_\ell(\psi) > B^2 \ell$. Therefore $d_{B \ell}(\phi)> B \ell$. Now since $\ket{\phi}$ is a stabilizer state \Cref{claim:stabcoderound} tells us that $\calC_{B \ell}(\phi)$ is a stabilizer code with $B \ell$-local stabilizer generators.
                                By \Cref{claim:code-blowup-containment} we have that 
    \begin{align}
        \calC_{B^2 \ell}(\psi) \sse U \calC_{B\ell}(\phi) \sse \calC_\ell(\psi)
    \end{align}
    Since we assume that $B < (d_\ell(\psi)/\ell)^{1/2}$, we have $d_\ell(\psi) > B^2 \ell$, so \Cref{fact:codes-equal-if-dgel} implies that $\calC_{B^2\ell}(\psi)=~\calC_\ell(\psi)$. Therefore,
    \begin{align}
        \calC := U \calC_{B \ell}(\phi) = \calC_\ell(\psi).
    \end{align}
\end{proof}

\subsection{Proof of robust infectiousness} \label{ssec:robust-code-lb}\label{ssec:robust-infectiousness}
\begin{definition}
    We say a subspace $\calC$ is $(\eps,\delta, \ell)$-robust for $0 \leq \eps \leq \delta \leq 1$ and integer $\ell\in [n]$ if for each $n$-qubit state $\ket{\phi}$ if $\ket{\phi}$ is $\ell$-locally $\eps$-close to $\calC$, then $\ket{\phi}$ is $\delta$ close to $\calC$. 
    \begin{itemize}
        \item []\textbf{$\ell$-locally $\eps$-close:} for each $A\sse [n]$ with $|A|\leq \ell$ there exists a $\phi'\in \calC$ such that $\frac{1}{2}\tnorm{\phi_A - \phi'_A}\leq~\eps$.
        \item [] \textbf{$\delta$-close:} There exists a $\phi'\in \calC$ such that $\frac{1}{2}\tnorm{\phi - \phi'} \leq \delta$.
    \end{itemize}
\end{definition}
An example of such robust subspaces are the groundspaces of gapped local Hamiltonians.
\begin{proposition}\label{prop:gap-to-robust}
    For each Hamiltonian $H = \sum_{i=1}^m h_i$ with terms $h_i$ that are $\ell$-local and $\norm{h_i}_\infty \leq~1$, such that the difference between the minimum eigenvalue and the second smallest eigenvalue is $g$, for each $\eps\in[0,1]$, the groundspace of $H$ is $(\eps, \sqrt{\eps m/g}, \ell)$-robust.
\end{proposition}
\begin{proof}
    Let $\calC$ be the ground space of $H$. 
        
    Suppose $\ket{\psi}$ is $\ell$-locally $\eps$-close to the groundspace of $H$. For each $i\in [m]$ let $S_i \sse [n]$ be the subset of qubits that $h_i$ is supported on, and let $\phi^{(i)} \in \calC$ be the state such that $\frac{1}{2} \tnorm{\phi^{(i)}_{S_i} - \psi_{S_i}} \leq \eps$.
    We assume without loss of generality that the minimum eigenvalue is 0. Then the energy of $\psi$~is
    \begin{align*}
        \Tr(H \psi) = \sum_{i=1}^m \tr(h_i \psi) 
        &= \sum_{i=1}^m \tr(h_i \phi^{(i)}) + \tr(h_i(\psi - \phi^{(i)}))\\
        &= \sum_{i=1}^m \tr(h_i(\psi_{S_i} - \phi^{(i)}_{S_i}) )\\
        &\leq   \sum_{i=1}^m  \norm{h_i}_\infty \cdot \tnorm{\psi_{S_i} - \phi^{(i)}_{S_i}}
        \leq \eps m
    \end{align*}
    Where in the first inequality we used Hölders and the triangle inequality.    

    We can write $H = \sum_i \lambda_i \Pi_i$ where $0 =\lambda_0 < \lambda_1 < \dots$ are the eigenvalues of $H$ and for each $i$, $\Pi_i$ is the projector onto the subspace of vectors with eigenvalue $\lambda_i$. Let $D$ be the trace distance of $\ket{\psi}$ to the nearest state in $\calC$, $D:= \sqrt{1 - \tr(\Pi_0 \psi)}$. 
    \begin{align*}
        \tr(H \psi) = \sum_i \lambda_i \tr(\Pi_i \psi) \geq \lambda_1 \tr((I - \Pi_0) \psi) = \lambda_1(1 - \tr(\Pi_0 \psi))= \lambda_1 D^2
    \end{align*}
    Therefore $D \leq \sqrt{\tr(H\psi) /\lambda_1} \leq \sqrt{\eps m /\lambda_1}$.
\end{proof}

\begin{theorem}\label{thm:robust-code-stab-containment}
    Suppose $\calC_\ell(\psi)$ is $(\eps, \delta, \ell)$-robust and $\delta < \frac{1}{8}$, with distance $d$.  Suppose one can prepare the state $\ket{\psi'} = U \ket{\phi^S}$ starting from a stabilizer state $\ket{\phi^S}$ that is nearby $\tnorm{\psi - \psi'} \leq \eps$ with a unitary $U$ with blowup $B$. If $d> B^2\ell$ then $\calC_{ \ell}(\psi')$ has distance $d/B^2$ and $\calC_\ell(\psi') = U \calC_S$ for some stabilizer code $\calC_S$ with  $B\ell$-local stabilizers. Furthermore, $\calC_{\ell}(\psi') \sse_\delta \calC_{\ell}(\psi)$ and $\calC_\ell(\psi) \sse_{\sqrt{\eps n}} \calC_\ell(\psi')$.
\end{theorem}

\begin{corollary}\label{cor:infectiousness-ham-technical}
    Consider an $\ell$-local Hamiltonian $H = \sum_{i = 1}^m h_i$ with each $\norm{h_i}_\infty \leq 1$ and a gap $\Delta\leq 1$ between its smallest two eigenvalues. Let $\calC$ be the groundspace of $H$ and suppose $\calC$ has distance $\omega(1)$. Suppose one can prepare the state $\ket{\psi'} = U \ket{\phi^S}$ starting from a stabilizer state $\ket{\phi^S}$  that is nearby $\frac{1}{2}\tnorm{\psi - \psi'} < \frac{\Delta}{64 m}$ with a unitary $U$ with blowup $B\leq \sqrt{d/ \ell}$. Then there exists a stabilizer code  $\calC_S$  with $B\ell$-local stabilizer generators such that $\calC \approx_{\sqrt{\eps m/\Delta}} U \calC_S$ and $U\calC_S$ has distance $\geq d/B^2$.
\end{corollary}

\begin{proof}[Proof of \Cref{thm:robust-code-stab-containment}]
    We now outline the proof of this theorem as it is split up into different lemmas which are proved below.
    \begin{enumerate}
        \item ($U$ preserves robustness) The state $\ket{\phi} := U^\dagger \ket{\psi}$ has codespace $\calC_{\ell'}(\phi)$ that is also $(\eps, \delta, \ell')$-robust, for $\ell' = B \ell$ and distance $d' \geq d/B$ (\Cref{lem:robustcode:unitary-preserves-robustness}). 
        
        \item Let $S_{\ell'}(\phi^S)$ be the \emph{set} of states that are $\ell'$-locally equivalent to $\phi^S$. 
        Since $\calC_{\ell'}(\phi)$ is robust and $\frac{1}{2}\tnorm{\phi - \phi^S}\leq \eps$, the set $S_{\ell'}(\phi^S)$ is $\gamma$-approximately recoverable on regions of size less than $d$ for $\gamma = 2\sqrt{2} \delta$ (\Cref{claim:close-apxdist:set-locallyequiv}).

        \item (Stabilizer rounding step) Let $\Cstab[\ell']$ be the subspace of states that are stabilized by the $\ell'$-local stabilizers of $\phi^S$. Since $S_{\ell'}(\phi^S)$ is $\gamma$-approximately recoverable on regions with less than $d$ qubits, and $\gamma< \frac{1}{2\sqrt{2}}$, $\Cstab[\ell'](\phi^S)$ is actually \emph{exactly} recoverable on regions of size less than $d$ (\Cref{claim:robustcode:appxrecset-to-exactcode}).

        \item $\calC_{\ell'}(\phi^S) = \Cstab[\ell'](\phi^S)$ (\Cref{claim:robustcode:locallyequiv-eq-stabcode}).

        \item Since $d\geq B^2 \ell$,  $\calC_\ell(\psi') = U \calC_{\ell'}(\phi^S)$ (\Cref{claim:code-blowup-containment}). Furthermore, since $\calC_{\ell'}(\phi^S)$ has distance $d'\geq d/B$, the distance of $\calC_\ell(\psi')$ is at least $d'/B\geq d/B^2$ (\Cref{claim:distance-preserved}).

        \item ($\calC_\ell(\psi) \approx \calC_\ell(\psi')$) Since each element of $\calC_\ell(\psi')$ is $\ell$-locally-equivalent to $\psi'$, and $\psi'$ is $\eps$-close to $\psi$, each element of $\calC_\ell(\psi')$ is $\ell$-locally $\eps$-close to $\psi$. Since $\calC_\ell(\psi)$ is $(\eps, \delta, \ell)$-robust, each state in $\calC_\ell(\psi')$ is $\delta$-close to a state in $\calC_\ell(\psi)$. 
        Since $\calC_\ell(\psi') = U \Cstab[B\ell](\phi^S)$ it is the ground space of a  $B^2 \ell$-local Hamiltonian with gap $\Delta=1$
                Therefore $\calC_\ell(\psi')$ is $(\eps, \sqrt{\eps n}, B^2\ell)$-robust \Cref{prop:gap-to-robust},
                which using the same reasoning implies that $\calC_{B^2\ell}(\psi) \sse_{\sqrt{\eps n}} \calC_\ell(\psi')$. Since $d> B^2\ell$, $\calC_\ell(\psi) = \calC_{B^2\ell}(\psi)$ (\Cref{fact:codes-equal-if-dgel}).
            \end{enumerate}    
\end{proof}

\begin{lemma}\label{lem:robustcode:unitary-preserves-robustness}
    Suppose $\ket{\psi} = U \ket{\phi}$ for a unitary $U$ with blowup $B$. If $\calC_\ell(\psi)$ is $(\eps, \delta, \ell)$-robust, then $\calC_{\ell'}(\phi)$ is also $(\eps, \delta, \ell')$-robust for $\ell' = B\ell$, and $d_{\ell'}(\phi) \geq \frac{1}{B}d_\ell(\psi)$.
\end{lemma}
\begin{proof}
    That $d_{\ell'}(\phi) \geq \frac{1}{B} d_\ell(\psi)$ follows from \Cref{claim:distance-preserved}. Suppose for the sake of contradiction that $\calC_\ell(\psi)$ is $(\eps, \delta, \ell)$-robust, but $\calC_{\ell'}(\phi)$ is not. Then there exists some state $\ket{\varphi}$ such that $\frac{1}{2}\tnorm{\varphi - \phi'} \geq \delta$ for each $\phi' \in \calC_{\ell'}(\phi)$, but that for each $A\sse [n]$ with $|A|\leq \ell'$, $\frac{1}{2}\tnorm{\varphi_A - \phi_A} < \eps$. But now consider the state $U \ket{\varphi}$. Since trace norm is unitarily invariant, 
    \begin{align}
        \frac{1}{2} \tnorm{U \varphi U^\dagger - \psi'} = \frac{1}{2}\tnorm{\varphi - \psi} \geq \delta && \text{for each } \psi' \in \calC_\ell(\psi). 
    \end{align}
    Therefore, since $\calC_\ell(\psi)$ is $(\eps, \delta, \ell)$-robust, there exists some $S\sse [n]$ with size at most $\ell$ such that $\frac{1}{2}\tnorm{\psi_S - (U \varphi U^\dagger)_S} \geq \eps$.
\end{proof}

\begin{definition}[Approximate recoverability]\label{def:appxrec}
    For a subset $S$ of a Hilbert space on $n$ qubits, we say that a region $A \sse [n] = A B$ of qubits is \emph{$\gamma$-recoverable} if there exists a unitary $U$ and a state $\ket{\tau}$ such that for each $\ket{\psi}\in S$,
    \begin{align}
        \norm{(U_{A'BC} \ot I_{A}) \ket{\psi}_{AB}\ket{0}_{A'C} - \ket{\psi}_{A'B}\ket{\tau}_{AC}} \leq \gamma.
    \end{align}
    For some ancillary register $C$. Note that $\ket{\tau}$ is fixed, independent of $\ket{\psi}$.
\end{definition}
\begin{proposition}
    If a subspace $\calC$ is recoverable on region $A$ in that there exists a recovery operator such that $\textsf{Rec}(\tr_{A^\complement} \rho) = \rho $ for each $\rho \in \calC$. Then it is also $0$-recoverable as in \Cref{def:appxrec}
\end{proposition}
\begin{proof}
    Let $\textsf{Rec}$ be the recovery map for $\calC$ on region $A$. It can be described by its Stinespring representation
    \begin{align}
        \textsf{Rec}(\rho) = \tr_{C} U (\rho \otimes \proj{0}_{AC})U\dagger    
    \end{align}
    for some unitary $U$ and auxiliary system $C$. Now for each state $\ket{\psi}$ in $\calC$, since $\textsf{Rec}(\Tr_{A}(\psi)) = \psi$, it must be that $(U_{A'BC} \ot I_{A}) (\psi_{AB} \ot \proj{0}_{A'C}) (U_{A'BC} \ot I_{A})^\dagger = \psi_{A'B}\ot \tau_{AC}$ for some quantum state $\tau$ on registers $AC$. Without loss of generality we can assume $\tau$ is pure since otherwise we could consider its purification. So we have that for each $\ket{\psi} \in \calC$ there exists a $\ket{\tau}$ such that
    \begin{align}
        (U_{A'BC} \ot I_{A}) \ket{\psi}_{AB} \ket{0}_{A'C} = \ket{\psi}_{AB}\ket{\tau}_{AC}.\label{eq:perfect-correction-unitary}
    \end{align}
    We next show that for each $\ket{\psi} \in \calC$ the corresponding $\ket{\tau}$ is the same. Consider two orthogonal $\ket{\psi_1}, \ket{\psi_2}\in \calC$, with corresponding $\ket{\tau_1}$ and $\ket{\tau_2}$ as above. Now we consider their linear combination $\ket{\psi_3} = \alpha \ket{\psi_1} + \beta \ket{\psi_2}$. Since $\ket{\psi_3}$ is also in $\calC$, it also has a $\ket{\tau_3}$ so that $(U_{A'BC} \ot I_{A}) \ket{\psi_3}_{AB} \ket{0}_{A'C} = \ket{\psi_3}_{AB}\ket{\tau_3}_{AC}$. By linearity, we have 
    \begin{align}
        (\alpha \ket{\psi_1} + \beta \ket{\psi_2})\ket{\tau_3} = \ket{\psi_3}\ket{\tau_3} = \alpha \ket{\psi_1}\ket{\tau_1} + \beta \ket{\psi_2}\ket{\tau_2}.
    \end{align}
    Therefore it must be that $\ket{\tau_1} = \ket{\tau_2}$. Therefore, for any orthonormal basis of $\calC$ there exists a single $\ket{\tau}$ such that \Cref{eq:perfect-correction-unitary} holds for each basis state $\ket{\psi}$. By linearity, it also holds for the entire subspace $\calC$. 
\end{proof}

\begin{lemma}\label{lem:robustcode:close-to-robust-is-approxcode}
    Suppose $\calC_{\ell'}(\phi)$ is $(\eps, \delta, \ell')$-robust for $\delta < \frac{1}{8}$ and has distance $d>\ell'$. If there exists a stabilizer state $\ket{\phi^S}$ such that $\frac{1}{2}\tnorm{\phi - \phi^S} \leq \eps$, then the code $\calC_{\ell'}(\phi^S)$ is an exact code with distance at least $d$. Furthermore $\calC_{\ell'}(\phi^S)$ is a stabilizer code with stabilizers generated by the stabilizers of $\ket{\phi^S}$ with weight at most $\ell'$.
\end{lemma}
\begin{proof}
    We break up the proof of this lemma into the following sequence of claims. 

    \begin{claim}\label{claim:close-apxdist:set-locallyequiv}
        If $\calC_{\ell'}(\phi)$ is $(\eps, \delta, \ell')$ robust, has distance $d$, and $\frac{1}{2} \tnorm{\phi - \phi'} \leq \eps$, then the set of all states $\ell'$-locally equivalent to $\phi'$, $S_{\ell'}(\phi')$, is $2 \sqrt{2}\delta$-correctable on each region of size less than $d$.
    \end{claim}
    \begin{proof}
        Since each state in $S_{\ell'}(\phi')$ is $\ell'$-locally equivalent to $\phi'$, and  $\phi'$ is $\epsilon$ close to $\phi$, it follows that each element of $S_{\ell'}(\phi)$ is $\ell'$-locally $\eps$-close to $\phi$. Therefore, by the definition of robustness, each of the elements of $S_{\ell'}(\phi)$ are $\delta$-close to some element in $\calC_{\ell}(\phi)$.
            
        For each $\ket{\varphi'_i}\in S_{\ell'}(\phi')$ there exists a $\ket{\varphi_i}\in \calC_{\ell'}(\phi)$ such that $\frac{1}{2}\tnorm{\varphi_i - \varphi'_i}\leq \delta$. 

        Since $\calC_{\ell'}(\phi')$ has distance $d$, for each region $A$ of size $|A|< d$, there exists a $\ket{\tau}$ such that
        \begin{align}
            \norm{U_{A'BC} \ket{\varphi}_{AB}\ket{0}_{A'C} - \ket{\varphi}_{A'B}\ket{\tau}_{AC}} = 0 && \forall \ket{\varphi} \in \calC_{\ell'}(\phi')
        \end{align}
        as in \Cref{def:appxrec}.    
        
        Therefore by triangle inequality,
        \begin{align}
            \norm{U_{A'BC} \ket{\varphi'_i}_{AB}\ket{0}_{A'C} - \ket{\varphi'_i}_{A'B}\ket{\tau}_{AC}} \leq  \norm{U_{A'BC} \ket{\varphi_i}_{AB}\ket{0}_{A'C} - \ket{\varphi_i}_{A'B}\ket{\tau}_{AC}} \\
            + \norm{U_{A'BC} (\ket{\varphi'_i} - \ket{\varphi_i})_{AB} \ket{0}_{A'C} -   (\ket{\varphi'_i} - \ket{\varphi_i})_{AB}\ket{\tau}_{AC}}\\
            \leq 2\norm{\ket{\varphi'_i} - \ket{\varphi_i}}.
        \end{align}
        Furthermore, without loss of generality we can choose $\varphi_i$ such that $\braket{\varphi | \varphi_i} = \abs{\braket{\varphi| \varphi_i}}$ by multiplying a global phase to $\ket{\varphi_i}$. Therefore 
        \begin{align}
            \norm{\ket{\varphi_i'} - \ket{\varphi_i}} = \sqrt{2 - 2\Re\{\braket{\varphi_i | \varphi_i'}\} } &= \sqrt{2 - 2 \abs{\braket{\varphi_i | \varphi_i'}}}\\
            &\leq \sqrt{2 - 2 \abs{\braket{\varphi_i | \varphi_i'}}^2} = \sqrt{2} \cdot  \frac{1}{2}\tnorm{\varphi_i - \varphi_i'} \leq \sqrt{2} \delta
        \end{align}
        Where in the last line we used the fact that for pure states $\ket{\psi}, \ket{\psi'}$, their trace norm is related to their overlap as $\frac{1}{2} \tnorm{\psi - \psi'} = \sqrt{1- \abs{\braket{\psi | \psi'}}^2}$.

        Therefore, combining the equations above, we have that for each region $S_{\ell'}(\phi')$ is approximately recoverable on the region $A$.
        \begin{align}
            \norm{U_{A'BC} \ket{\varphi_i'}_{AB}\ket{0}_{A'C} - \ket{\varphi_i'}_{A'B}\ket{\tau}_{AC}} \leq 2\sqrt{2} \delta && \text{for each} \ \ket{\varphi_i} \in S_{\ell'}(\phi')
        \end{align}

                                    \end{proof}

    \begin{claim}\label{claim:robustcode:local-set-contains-stabcode}
        For stabilizer state $\ket{\phi}$, let $\Cstab[\ell'](\phi)$ be the subspace of states that are stabilized by the $\ell'$-local stabilizers of $\ket{\phi}$, and let $S_{\ell'}(\phi)$ be the set of states that are $\ell'$ locally equivalent to $\phi$. If $S_{\ell'}(\phi)$ is $\gamma$-recoverable for all regions of size $< d$ for some $d> \ell'$ and $\gamma <\frac{1}{2\sqrt{2}}$, then
        \begin{align}
            \Cstab[\ell'](\phi) \cap \textsf{STAB} \sse S_{\ell'}(\phi). \label{eq:claim:local-set-contains-stabcode}
        \end{align} 
    \end{claim}
    \begin{proof}
        Suppose for the sake of contradiction that there exists a state $\ket{\varphi} \in \pbra{ \Cstab[\ell'] \setminus S_{\ell'}(\phi) }\cap \textsf{STAB}$. Since $\ket{\varphi}\in \Cstab[\ell']$ all of its stabilizers commute with the $\ell'$-local stabilizers of $\ket{\phi}$ but since it's not in $S_{\ell'}(\phi)$, it's not $\ell'$-locally equivalent to $\ket{\phi}$. Therefore, it must have some $\ell'$-local stabilizer $\sigma$ that is not a stabilizer of $\ket{\phi}$ but commutes with the stabilizers of $\ket{\phi}$. Therefore, $\sigma$ maps $\ket{\phi}$ to $\ket{\phi'} := \sigma \ket{\phi}$, both $\ket{\phi}$ and $\ket{\phi'}$ are in $S_{\ell'}(\phi) \cap \textsf{STAB}$ and $\ket{\phi} \neq \ket{\phi'}$. Let $A \sse [n]$ be the subset of qubits that $\eta$ is supported on. Since $|A| \leq \ell' < d$, $A$ is a $\gamma$-approximately recoverable region.

        Since $S_{\ell'}$ has $\gamma$-approximate distance greater than $\ell'$, and $|A|\leq \ell' < d$, $S_{\ell'}$ is $\gamma$-approximately correctable on $A$. Let $U_{A'BC}$ be the correcting unitary as in \Cref{def:appxrec}, then there exists a state $\ket{\tau}$ such that 
        \begin{align}
            U_{A'BC}  \ket{\phi}_{AB}\ket{0}_{A'C} \approx_\gamma \ket{\phi}_{A'B}\ket{\tau}_{AC} \quad 
            \text{and} \quad  U_{A'BC} \ket{\phi'}_{AB}\ket{0}_{A'C} \approx_\gamma \ket{\phi'}_{A'B}\ket{\tau}_{AC}. \label{eq:robustcode:twostatesrecoverable}
        \end{align}

        Where we use $\ket{a}\approx_\gamma \ket{b}$ to denote that $\norm{\ket{a} - \ket{b}}\leq\gamma$.    
        Since $\ket{\phi}= \sigma \ket{\phi'}$ and $\sigma$ is supported only on $A$, we have that
        \begin{align}
            \sigma_A U_{A'BC}  \ket{\phi'}_{AB}\ket{0}_{A'C} =  U_{A'BC}   \ket{\phi}_{AB}\ket{0}_{A'C} \label{eq:robustcode:errorsmallregion}
        \end{align}
        Combining \Cref{eq:robustcode:twostatesrecoverable,eq:robustcode:errorsmallregion}
        \begin{align}
            \ket{\phi}_{A'B}\ket{\tau}_{AC} \approx_\gamma  U_{A'BC} \ket{\phi}_{AB}\ket{0}_{A'C} = \sigma_{A} U_{A'BC}   \ket{\phi'}_{AB}\ket{0}_{A'C} \approx_\gamma \ket{\phi'}_{A'B}\ot \sigma_A \ket{\tau}_{AC}. \label{eq:robustcode:states-close}
        \end{align}
        So by the triangle inequality, letting $\ket{\tau'}_{AC} :=\sigma_A \ket{\tau}_{AC}$ we have
        \begin{align}
            \norm{\ket{\phi}\ket{\tau} - \ket{\phi'} \ket{\tau'}} \leq 2\gamma < \frac{1}{\sqrt{2}}. \label{eq:robustcode:correctedstateclose}
        \end{align}
        We now will prove a lower bound on $\norm{\ket{\phi}\ket{\tau} - \ket{\phi'} \ket{\tau'}}$. Using that $\norm{\ket{a}-\ket{b}} \geq \frac{1}{2}\tnorm{\proj{a} - \proj{b}}$, we have
        \begin{align}\label{eq:robustcode:closefromoverlap}
                \norm{\ket{\phi}\ket{\tau} - \ket{\phi'} \ket{\tau'}} \geq \frac{1}{2} \tnorm{\phi\ot \tau - \phi' \ot \tau'} = \sqrt{1 - \abs{\braket{\phi | \phi'} \braket{\tau | \tau'}}^2} \geq \sqrt{1 - \abs{\braket{\phi | \phi'}}^2}
        \end{align}

        Since $\ket{\phi}, \ket{\phi'}$ are both stabilizer states, $\abs{\braket{\phi|\phi'}}$ is a power of $\frac{1}{2}$ (\Cref{fact:stab-overlap-powerof2}). Since they are not equal, we then have that $|\braket{\phi | \phi'}|^2 \leq \frac{1}{2}$. Therefore
                                        \begin{align}
            \norm{\ket{\phi}\ket{\tau} - \ket{\phi'} \ket{\tau'}} \geq \frac{1}{\sqrt{2}}. \label{eq:robustcode:correctedstatefar}
        \end{align}
        This contradicts \Cref{eq:robustcode:correctedstatefar}.

    \end{proof}
    \begin{claim}\label{claim:robustcode:appxrecset-to-exactcode}
        For stabilizer state $\ket{\phi}$, let $\Cstab[\ell'](\phi)$ be the subspace of states that are stabilized by the $\ell'$-local stabilizers of $\ket{\phi}$, and let $S_{\ell'}(\phi)$ be the set of states that are $\ell'$ locally equivalent to $\phi$. If $S_{\ell'}(\phi)$ is $\gamma$-recoverable for all regions of size $< d$ for some $d> \ell'$ and $\gamma <\frac{1}{2\sqrt{2}}$, then  $\Cstab[\ell'](\phi)$ has exact distance $d$.
                            \end{claim}
    \begin{proof}
        Suppose for the sake of contradiction that $\Cstab[\ell'](\phi)$ has distance less than $d$. Then there exists a logical operator $\sigma \in \paulis[n]$ supported on at most $d$ qubits that commutes with the stabilizer group $\calS_{\ell'}$ of $\Cstab[\ell']$ but is not itself in $\calS_{\ell'}$.  Furthermore there exists distinct stabilizer states $\ket{\psi} \neq \ket{\psi'} \in \Cstab[\ell'](\phi) \cap \textsf{STAB}$ such that  $\ket{\psi'} = \sigma\ket{\psi}$.

                Since \Cref{claim:robustcode:local-set-contains-stabcode} implies that $\Cstab[\ell'](\phi)\cap \textsf{STAB} \sse S_{\ell'}(\phi)$, we also have that $\ket{\psi}, \ket{\psi'} \in S_{\ell'}(\phi)$.

    Let $A\sse [n]$ be the subset of qubits in the support of $\sigma$.  Since $|A| \leq \ell' < d$, $A$ is a $\gamma$-approximately recoverable region for $S_{\ell'}(\phi)$, so $\sigma$ is an error from which we can approximately recover.  Following the same argument as in the proof of \Cref{claim:robustcode:local-set-contains-stabcode} from \Cref{eq:robustcode:twostatesrecoverable,eq:robustcode:errorsmallregion,eq:robustcode:states-close,eq:robustcode:correctedstateclose,eq:robustcode:closefromoverlap,eq:robustcode:correctedstatefar} since $\gamma < \frac{1}{2\sqrt{2}}$, and $\ket{\psi}, \ket{\psi'}$ are distinct stabilizer states, we arrive at a contradiction.  
    \end{proof}

    By \Cref{claim:close-apxdist:set-locallyequiv} we have that $S_{\ell'}(\phi^S)$ is $\gamma:= 2\sqrt{2} \delta$-correctable on regions of size less than $d$. Since by the lemma statement $\delta < 1/8$, we have that $\gamma < \frac{1}{2 \sqrt{2}}$, and furthermore $d > \ell'$ \Cref{claim:robustcode:appxrecset-to-exactcode} implies that  $\Cstab[\ell'](\phi^S)$ has exact distance $d$.

    \begin{claim}\label{claim:robustcode:locallyequiv-eq-stabcode}
        $\calC_{\ell'}(\phi^S)= \Cstab[\ell'](\phi^S)$. 
    \end{claim}
    \begin{proof}
    Since each element of $\calC_{\ell'}(\phi^S)$ is $\ell'$-locally equivalent to $\ket{\phi^S}$, it also is stabilized by $\ket{\phi^S}$'s $\ell'$-local stabilizers. Therefore $\calC_{\ell'}(\phi^S)\sse \Cstab[\ell'](\phi^S)$. 

    We now show the other direction: $\Cstab[\ell'](\phi^S) \sse \calC_{\ell'}(\phi^S)$. \Cref{claim:robustcode:local-set-contains-stabcode}, \Cref{eq:claim:local-set-contains-stabcode} imply that $\Cstab[\ell'](\phi^S) \cap \textsf{STAB} \sse S_{\ell'}(\phi^S)$, therefore
    \begin{align*}
        \Cstab[\ell'](\phi^S) = \textsf{span}\pbra{\Cstab[\ell'](\phi^S) \cap \textsf{STAB} } \sse \textsf{span}\pbra{ S_{\ell'}(\phi^S)} = \calC_{\ell'}(\phi^S).
    \end{align*}
    The first equality follows from the fact that $\Cstab[\ell'](\phi^S)$ is a stabilizer code, so it is spanned by stabilizer states. The second equality follows from the definition of $\calC_{\ell'}(\phi^S)$ as the span of states that are $\ell'$-locally equivalent to $\ket{\phi^S}$.

    \end{proof}
    Finally \Cref{claim:robustcode:locallyequiv-eq-stabcode} shows that  $\calC_{\ell'}(\phi^S) = \Cstab[\ell'](\phi^S)$ completing the proof.
\end{proof}

\subsection{Circuit lower bounds for explicit codes} 
\Cref{cor:infectiousness-ham-technical} tells us that codespaces that are not approximately perturbed Hamiltonians cannot be approximately prepared by $\altCQ[1]$ circuits. Now the remaining step to prove lower bounds for explicit quantum codes is to prove that a code is not approximately a perturbed stabilizer state. We show two ways of doing this. The first holds for codes wth a codestate that has lots of pairwise correlations between qubits, and makes use of the Disentangling Lemma of \cite{BPT}. The second lower bound approach  holds for codespaces that have dimension that is not a power of 2. This uses the fact that stabilizer codes have a dimension that is always a power of two. This second result is also shown independently \cite{wei2025long} with a different approach.

\subsubsection{Lower bounds via the Disentangling Lemma}
\label{ssec:exact-code-lb}
In this section we use the Infectiousness property of $\altCQ[1]$ states in codes to prove that explicit states can not be prepared by $\altCQ[1]$ circuits. We will make use of the following variant of the Disentangling lemma from \cite{BPT}.
\begin{lemma}[Disentangling Lemma v2 (implicit in proof of DL in \cite{BPT})]\label{lem:disentangling-two}
    Suppose $\calC$ is the codespace corresponding to the $1$-eigenspace of pairwise commuting projectors $\{\Pi_i\}$. For each subset of qubits $M\sse [n]$, if together $M$ and its boundary $ M \cup \partial^+(M)$ is a correctable region, then there exists a unitary operator $U_{\partial_{\partial(M)^+}}$ acting only on $\partial^+(M)$ such that
    \begin{align}
        U_{\partial^+(M)}^\dagger\Pi U_{\partial^+(M)} = \proj{\eta_{MB}}\ot \Pi_{B'C}
    \end{align}
    For some decomposition of the Hilbert space on $\partial^+M$ into $\calH_{\partial^+(M)} = \calH_B \ot \calH_{B'}$,  $C = (M \cup \partial^+(M))^c$, and some fixed state $\ket{\eta_{MB}}$. 
\end{lemma}

\begin{corollary}\label{cor:dl:prodstate}
    Suppose $\calC$ is a codespace as above, and $M$ a region of qubits such that $N(M) = M \cup \delta^+(M)$ is a correctable region. Then   for each $\rho \in \calC$, 
    \begin{align}
        \rho_{MC} = \rho_M \ot \rho_{C}.
    \end{align}
    Where $C = N(M)^c$.
\end{corollary}

Combining \Cref{thm:le-code-contained-in-commproj-code} and \Cref{cor:dl:prodstate} of \Cref{lem:disentangling-two}, we get a lower bound for preparing local codes with high distance and lots of qubits that are pairwise correlated. 

\begin{theorem}\label{thm:code-lb-exact}
    Let $\ket{\psi}$ be an $n$-qubit quantum state and $\ell\geq 1$ an integer. Suppose $\calC_\ell(\psi)$ is $(\eps, \delta)$-robust for $\delta < 1/8$. 
    If there exists $t$ disjoint regions $S_1, S_2, \dots, S_t \sse [n]$ of the qubits such that $\frac{1}{2} \tnorm{\psi_{S_i S_j} - \psi_{S_i}\ot \psi_{S_j}} > 2\delta$ for each $i\neq j \in [t]$. 
    Then preparing any state $\eps$-close to $\calC_\ell(\psi)$ from stabilizer requires a circuit that has blowup at least $ B \geq \Omega\pbra{\frac{\min\{d_\ell(\psi), t\} \cdot t}{\ell^2 n}}^{1/5}$ 
\end{theorem}
\begin{corollary}\label{cor:ham-dl-exp-lb}
    Consider an $\ell$-local Hamiltonian $H = h_i$ with $\norm{h_i}_\infty \leq 1$ for each $i$ and a spectral gap $\Delta$ between its smallest two eigenvalues. Suppose its groundspace, $\calC$ has distance $\omega(1)$. If there exists a single state $\ket{\psi}\in \calC$, and $t$ disjoint regions $S_1, S_2, \dots, S_t \sse [n]$ of the qubits such that $\frac{1}{2} \tnorm{\psi_{S_i S_j} - \psi_{S_i}\ot \psi_{S_j}} > \gamma$ for each $i\neq j \in [t]$, then preparing \emph{any} state that is $\frac{\gamma^2 \Delta}{4m}$-close in trace distance to $\calC$ from a stabilizer state, requires a circuit that has blowup at least $B \geq \Omega\pbra{\frac{\min\{d_\ell(\psi), t\} \cdot t}{\ell^2 n}}^{1/5}$.
\end{corollary}
\begin{proof}[Proof of \Cref{thm:code-lb-exact}]
    Suppose $\ket{\phi'}= U\ket{\phi^S}$ and $\frac{1}{2}\tnorm{\phi - \phi'}\leq \eps$ for some $\phi \in \calC_\ell(\psi)$, some stabilizer state $\ket{\phi^S}$ and a circuit $U$ with blowup $B$. 
    Note that since $\ket{\phi}\in \calC_\ell(\psi)$,  $\calC_\ell(\phi) = \calC_\ell(\psi)$ and so $d:= d_\ell(\phi) = d_\ell(\psi)$. Suppose for the sake of contradiction that $B< \pbra{d/\ell}^{1/2}$ and $B < \pbra{(\min\{d, t\}-1)t/(2 \ell^2 n)}^{1/5}$. 
    By \Cref{thm:robust-code-stab-containment} and \Cref{prop:gap-to-robust}
        since $B< \pbra{d/\ell}^{1/2}$, 
    \begin{align}
        \calC_\ell(\psi) = \calC_\ell(\phi)\sse_{\delta} \calC' \label{eq:explcodelb:ssestab}
    \end{align}
    where $\calC'$ has distance $\geq d/B$ and is the common $+1$-eigenspace of at most $n$ commuting local projectors $\{\Pi_i\}$ that are $B^2\ell$-local. Let $G$ be the graph with $S_1, \dots, S_t$ as vertices, and an edge between $S_a\neq S_b$ if there exists a $\Pi_i$ that acts nontrivially on both $S_a$ and $S_b$. Since there are at most $n$ different projectors, each acting nontrivially on at most $B^2\ell$ qubits, the total number of edges is $|E| \leq n \cdot (B^2 \ell)^2$. Therefore the average degree is $\widetilde{\textsf{deg}} = 2 |E| / |T| \leq 2B^4 \ell^2 n / t$. By a simple counting argument there exists some vertex $S_a$ with degree at most $\widetilde{\textsf{deg}}$. Now the neighborhood of $S_a$ has size at most $|N(S_a)|\leq  \widetilde{\textsf{deg}} + 1 \leq 2B^4 \ell^2 n/t + 1$. 
    
    Since we assumed for the sake of contradiction that $B < \pbra{(\min\{d, t\}-1)t/(2 \ell^2 n)}^{1/5}$,  we have that $|N(S_a)| < \min\{d/B, t\}$. Therefore, since $\calC'$ has distance $\geq d/B$, the region $N(S_a)$ is correctable and there exists an $S_b \in N(S_a)^c$. So \Cref{cor:dl:prodstate} implies that
    \begin{align}
        \frac{1}{2}\tnorm{\psi'_{S_aS_b} - \psi'_{S_a}\ot \psi'_{S_b}} = 0 && \text{for each} \ \ket{\psi'}\in \calC'
    \end{align}
    By \Cref{eq:explcodelb:ssestab} there exists a $\ket{\psi'} \in \calC'$ such that $\frac{1}{2}\tnorm{\psi - \psi'} \leq \delta$. Therefore, using the triangle inequality and that the trace norm is non-increasing under partial trace,
    \begin{align}
        \frac{1}{2}\tnorm{\psi_{S_a S_b} - \psi_{S_a} \ot \psi_{S_b}} &\leq \frac{1}{2} \tnorm{\psi'_{S_a S_b} - \psi'_{S_a}\ot \psi'_{S_b}} + 2 \cdot \frac{1}{2} \tnorm{\psi - \psi'}\\
        &\leq 2 \delta.
    \end{align}
    This contradicts our theorem assumption that $\frac{1}{2}\tnorm{\psi_{S_a S_b} - \psi_{S_a} \ot \psi_{S_b}} > 2 \delta$. Therefore either $B > \pbra{d/ \ell}^{1/2}$ or $B >\pbra{(\min\{d, t\}-1)t/(2 \ell^2 n)}^{1/5} = \Omega(\min\{d, t\}t/\ell^2 n)^{1/5}$. Note that we also have $\pbra{d / \ell}^{1/2} \geq \Omega(\min\{d, t\} t/\ell^2 n)^{1/5}$ since we assume that $B \geq 1$ in the Theorem statement. So in either case we have that $B \geq \Omega(\min\{d, t\} t/\ell^2 n)^{1/5}$, completing the proof.

    \end{proof}

\paragraph{Lower bound for approximately preparing a history state}
An explicit state that we can prove a lower bound for using \Cref{cor:ham-dl-exp-lb} is the following Feynman-Kitaev history state for preparing the CAT state.
\begin{align}
    \ket{\Hcat}:= \frac{1}{\sqrt{n}} \sum_{t=0}^n \ket{\textsf{unary(t)}}_{\textsf{time}} \ot\ket{\textsf{CAT}_t} \ket{0^{n-t}}_{\textsf{state}}
\end{align}
This state encodes the history of preparing the $\ket{\textsf{CAT}} = \frac{1}{\sqrt{2}}(\ket{0^n} + \ket{1^n})$ state by starting with the initial state $\ket{0^n}$, applying the $H$ gate to the first qubit at time $t=1$, and for each subsequent time steps $t$ apply a CNOT gate from qubit 1 to qubit $t$. Thus at time $t$, the state of the computation is $\ket{\textsf{CAT}_t}\ket{0^{n-t}}$.
\thmcathistorystaterobustlb*
\begin{proof}

The $\ket{\textsf{H(CAT)}}$ state is the unique ground state of the local Hamiltonian
\[H = H_{in} + H_{prop} + H_{clock}.\]
$H_{prop}$ enforces that the ground state is a valid history state of the form $\sum_{t} \ket{\textsf{unary}(t)} \ket{\psi_t}$ such that at each time $t$, $\ket{\psi_t} = U_t \ket{\psi_{t-1}}$ for $U_t$ being the $t$th gate in our circuit constructing the CAT state. 
$H_{in}$ enforces that the state at time $t=0$ is $\ket{0^n}$,
and $H_{clock}$ enforces that the clock register only contains valid unary encodings. We refer the reader to \cite{nirkhe2018approximate} for the explicit definition of the Hamiltonian terms. For our purposes, we will use the well-known facts that $H$ is $5$-local, is the sum of $\Theta(n)$ terms and that it has a spectral gap of $\Delta = \Omega(1/n^2)$ \cite{aharonov2008adiabatic}. Furthermore, since $\ket{\Hcat}$ is the unique groundstate of $H$, the groundspace has distance $>2n$. We now show that the CAT history state has many qubits that are highly pairwise correlated.
\begin{claim}
    Let $S$ be the subset of the first $n/2$ qubits in the $\textsf{state}$ register of $\ket{\textsf{H(CAT)}}$. Each $i\neq j \in S$ satisfies $\frac{1}{2}\tnorm{\Hcat_{ij} - \Hcat_i \ot \Hcat_j} \geq 1/16$
\end{claim}
\begin{proof}
    This proof closely follows part of the proof of Theorem 1 in \cite{nirkhe2018approximate}. For each $i <j \in S$ we consider the observables $A_i = \proj{0}_i$ and $B_j = \proj{1}_j$ acting on the $\textsf{state}$ register of $\ket{\Hcat}$.

    \begin{align}
        \tr(\Hcat_{ij} A_i B_j)
                &= \frac{1}{n} \sum_{t = 0}^n  \bra{\textsf{CAT}_t 0^{n-t}} A_i B_j \ket{\textsf{CAT}_{t} 0^{n-t}} = 0
    \end{align}
    Each of the terms in this sum will be $0$ since each state $\ket{\textsf{CAT}_t 0^{n-t}} = \ket{\textsf{CAT}_t} \ket{0^{n-t}}$ is only supported on basis states $\ket{0^n}$ and $\ket{1^t0^{n-t}}$, neither of which have a 0 occur before a 1. On the other hand, 
    \begin{align}
        \tr(\Hcat_{ij} A_i) = \frac{1}{n} \sum_{t = 0}^n  \bra{\textsf{CAT}_t 0^{n-t}} A_i \ket{\textsf{CAT}_{t} 0^{n-t}} \geq 1/2
    \end{align}
    since for each $t\in [n]$, $\bra{0^{n-t}} A_i \ket{\textsf{CAT}_{t} 0^{n-t}} = 1$ if $t< i$ and is $1/2$ otherwise. Furthermore, 
    \begin{align}
        \tr(\Hcat_{ij} B_j) = \frac{1}{n} \sum_{t = 0}^n  \bra{\textsf{CAT}_t 0^{n-t}} B_j \ket{\textsf{CAT}_{t} 0^{n-t}} \geq 1/4
    \end{align}
    since $j \leq  n/2$ and $\bra{0^{n-t}} B_j \ket{\textsf{CAT}_{t} 0^{n-t}} = \frac{1}{2}$ when $t\geq j$ and $0$ otherwise. Therefore,
    \begin{align}
        \tr(\Hcat_{i}\ot \Hcat_j A_i B_j) = \tr(\Hcat_i A_i) \cdot \tr(\Hcat_j B_j) \geq 1/8.
    \end{align}
    By the duality of Schatten norms,
    \begin{align}
        \frac{1}{2} \tnorm{\Hcat_{ij} - \Hcat_i \ot \Hcat_j} &\geq \frac{1}{2}\abs{\tr(\Hcat_{ij} A_i) + \tr(\Hcat_{ij} B_j) - \tr(\Hcat_{i}\ot \Hcat_j A_i B_j)} 
        &\geq \frac{1}{16}.
    \end{align}
\end{proof}

    Now we can apply \Cref{cor:ham-dl-exp-lb} to get a lower bound for approximately preparing $\ket{\Hcat}$. Since $H$ is $\ell=5$-local, has gap  $\Delta \geq \Omega(1/n^2)$, the groundspace has distance $d>2n$, and we satisfy the conditions of \Cref{cor:ham-dl-exp-lb} with $t = n/2$ and $\gamma = 1/16$, we get that approximately preparing a state $\frac{1}{16^2 n^2 \cdot  4m} = \Omega(1/n^3)$- close in trace distance to $\ket{\Hcat}$ starting from a stabilizer sate requires a circuit with blowup at least $\Omega(n^{2/5})$. Since the blowup of a circuit $U$ is at most $2^\depth(U)$, it follows that the required depth is at least $\Omega(\log n)$.
\end{proof}

\subsubsection{Lower bounds for codes with dimension that is not a power of \texorpdfstring{$2$}{2}}

In this section, we show the following theorem

\thmdimrobustlb*

We will use the following fact that follows from basic linear algebra.
\begin{fact}\label{fact:subspace-dim}
    Consider two subspaces $\calC_1, \calC_2 \sse \calH$ of some bounded dimensional Hilbert space $\calH$. If $\textsf{dim}(\calC_1) > \textsf{dim}(\calC_2)$, then there exists a vector $\ket{\psi} \in \calC_1$ that is orthogonal to all vectors in $\calC_2$.
\end{fact}
\begin{proof}
    Since $\calC_2$ is a closed subspace, we can split up our Hilbert space into $\calH = \calC_2 \oplus \calC_2^\perp$ where $\calC_2^\perp = \{\ket{y} \in \calH : \braket{y | x} = 0 \ \forall \ket{x} \in \calC_2\}$. So, we have that
    \begin{align}
        \textsf{dim}(\calC_1) = \textsf{dim}(\calC_1 \cap \calC_2)  + \textsf{dim}(\calC_1 \cap \calC_2^\perp).
    \end{align}
    Therefore,
    \begin{align}
        \textsf{dim}(\calC_1\cap \calC_2^\perp) = \textsf{dim}(\calC_1) - \textsf{dim}(\calC_1 \cap \calC_2) \geq \textsf{dim}(C_1) - \textsf{dim}(C_2) > 0.
    \end{align}
    Thus, there exists a state in $\calC_1$ that is orthogonal to $\calC_2$. 
\end{proof}

\begin{proof}[Proof of \Cref{thm:dim-power2-lb}]
    Suppose for the sake of contradiction that $\calC$ has dimension that is not a power of 2 and that there exists a state $\ket{\psi'}$ that can be prepared starting from a stabilizer state with a circuit with blowup $B < \sqrt{d/\ell}$ such that $\frac{1}{2} \tnorm{\psi - \psi'} \leq \eps \leq \frac{\Delta}{64 m}$ for some $\ket{\psi} \in \calC$. Then \Cref{cor:infectiousness-ham-technical} implies that there exists a stabilizer code $\calC_S$ such that $\calC$ is $\sqrt{\eps m/\Delta} \leq 1/8$-approximately equivalent to $\calC$. So for each state in $\calC$ it has trace distance at most $1/8$ from $U \calC_S$  and vice versa. Therefore, there is no state in $\calC$ that is orthogonal to all the states in $U\calC_S$ and vice versa. Therefore, \Cref{fact:subspace-dim} implies that $\textsf{dim}(\calC) = \textsf{dim}(U \calC_S)$. But since $\calC_S$ is a stabilizer code, it has dimension that is a power of 2, which is a contradiction. Therefore $B\geq \sqrt{d/\ell}$ and so the depth is at least $\log B \geq \frac{1}{2} \log (d/\ell)$.
\end{proof}

 \section{Connections and basic properties of the magic hierarchy}

In this section we prove some of the basic properties we mentioned in the introduction about circuits in the magic hierarchy. We begin by proving that the level of a circuit in the magic hierarchy is related to other known complexity measures in \Cref{ssec:connections}, we then prove the incomparability of $\altCQ[1]$ and $\altQC[1]$ in \Cref{ssec:CQneqQC}

\subsection{Equivalence with Fanout, intermediate measurement, and T-depth} \label{ssec:connections}

In this subsection we elaborate on the connections between the number of alternations and other interesting complexity measures mentioned in \Cref{ssec:intro:connections}. These results follow straightforwardly from known results --- some of which are folklore. We restate \Cref{prop:fanouteqMHlevel,prop:IMleqMHlevel,prop:IMpareqMHlevel} and sketch their proofs below. 

When comparing circuit models and claiming that one can simulate the other, we implicitly are allowing the use of a polynomial number of ancillary qubits $\ket{0}$ and partial trace.

For a circuit in the magic hierarchy, we use $\cliffdepth$ to refer to the number of Clifford circuits, and $\qnczrounds$ as the number of $\QNCZ$ circuits. It is straightforward to see that the $\cliffdepth$ and $\qnczrounds$ required for a circuit is bounded by the level of the magic hierarchy, $\MHlevel$, as
    \begin{align}
        \frac{1}{2}\MHlevel \  \leq \ \cliffdepth, \ \qnczrounds \ \leq \ \frac{1}{2} \MHlevel + 1 \label{eq:level-bound-Cliff-QNC}.
    \end{align}

    \paragraph{Adaptive intermediate measurements:} 
    Adaptive intermediate measurements are another surprisingly powerful nonlocal operation that is used to supplement the power of constant depth circuits. 
    We consider constant-depth quantum circuits ($\QNCZ$)that are allowed mid-circuit measurements where classical processing is done on the measurement outcomes and fed back into the next stage of the circuit. See \Cref{ssec:intro:connections} for our more formal definition of this model.
                                
    We consider the case that the classical circuits are restricted to be polynomials-sized \emph{parity circuits}, meaning that they allow only gates computing the parity of a subset of the bits $\textsf{parity}(x) = \sum_i x_i \mod 2$. Furthermore, as is typical for classical circuits, the output of the parity gate has unbounded fan-out and can be copied many times. 
    In the case where the classical computation allows for only parity gates, the number of required intermediate measurement rounds, $\IMparityrounds$,  and magic hierarchy level are equivalent.
                \IMpareqMHlevel*
    \begin{proof}[Proof]
        A key insight from measurement-based quantum computing \cite{browne2011computational} is that Clifford circuits can be parallelized to constant-depth using teleportation, which involves making several Bell-basis measurements and Pauli corrections. This corresponds to a $\QNCZ$ circuit augmented with a single layer of intermediate measurements. For completeness, we show how to paralleize Clifford circuits with intermediate measurements. In particular, we highlight that the classical processing needed only involves computing parities of many inputs. Therefore $\IMparityrounds \leq \cliffdepth$.
         
        Using the standard gate teleportation technique, we will compress any Clifford circuit into constant depth up to Pauli corrections. Below we review the teleportation gadget.
        \begin{equation}
            \begin{quantikz}[row sep={.7cm,between origins}]
                \lstick{$\ket{\psi}$}  \qw &  \qw& \gate[2]{B^\dagger} & \meter{} & \cw            \rstick{$z$} \\
                \lstick{$\ket{0}$} \qw & \gate[2]{B} &         \qw         & \meter{}&   \cw  \rstick{$x$} \\
                \lstick{$\ket{0}$} \qw & \qw         &   \qw  & \qw \rstick{$ Z^z X^x\ket{\psi}$}
            \end{quantikz}
            \quad \text{where} \quad
            \begin{quantikz}[row sep={.7cm,between origins}]
                \qw & \gate[2]{B} \qw & \qw\\
                \qw & \qw & \qw
            \end{quantikz}
            =
            \begin{quantikz}[row sep={.7cm,between origins}]
                \qw & \gate{H} \qw & \ctrl{1}& \qw  \\
                \qw & \qw & \targ{} & \qw
            \end{quantikz}
        \end{equation}
        We can do this more generally for an $n$-qubit state by teleporting each of the qubits in parallel with $B^{\ot n}$. In this case, we will have $2n$ measurement outcomes $\vec{z},\vec{x} \in \zo^n$ so that our Pauli corrections are $Z(\vec{z})X(\vec{x}) $. If we apply some $n$-qubit unitary $U$ to the bottom wire after teleportation, it is as if we applied it directly to the state $\ket{\psi}$, but up to some Pauli corrections.
        \begin{equation}
            \begin{quantikz}
                \lstick{$\ket{\psi}$} & \qwbundle{n} & \gate{U} & \qw
            \end{quantikz}
            {\quad  \rightarrow \quad }
            \begin{quantikz}[row sep={.7cm,between origins}]
                \lstick{$\ket{\psi}$} & \qwbundle{n} &  \qw& \gate[2,disable auto height]{(B^\dagger)^{\ot n}} & \meter{} & \cw            \rstick{$\vec{z}$} \\
                \lstick{$\ket{0}$} & \qwbundle{n} & \gate[2,,disable auto height]{B^{\ot n}} &         \qw         & \meter{}&   \cw  \rstick{$\vec{x}$} \\
                \lstick{$\ket{0}$} & \qwbundle{n} & \qw         &         \gate{U}   & \qw \rstick{$ U Z(\vec{z})X(\vec{x})  \ket{\psi}$}
            \end{quantikz}
        \end{equation}
        In the special case where $U$ is Clifford, it has the nice property that it maps Pauli operators to Pauli operators, so $U Z(\vec{z})X(\vec{x})  =  Z(\vec{z'}) X(\vec{x'}) U $ for some $x', z' \in \zo^n$. Furthermore, each Clifford operator is a \emph{linear} map from Paulis to Paulis, so in particular the map from $(x,z)$ to $(x', z')$ is a linear map over $\FF_2$. We can continue doing this to compress $D$ layers of a circuit as follows.

        \begin{equation}\label{eq:gate-teleportation-full-circuit}
            \begin{quantikz}[row sep={.7cm,between origins}]
                \lstick{$\ket{\psi}$} & \qwbundle{n} & \qw & \gate{U_1}  & \gate[wires=2,disable auto height]{(B^\dagger)^{\otimes n}} \qw  & \meter{}\qw & \cw            \rstick{$\vec{x_1}$}\\
                \lstick{$\ket{0^n}$}  & \qwbundle{n} &  \gate[wires=2][1.2cm]{B^{\otimes n}} &\qw &\qw & \meter{} \qw & \cw            \rstick{$\vec{x_1}$}\\
                \lstick{$\ket{0^n}$} & \qwbundle{n} &  & \gate{U_2} &  \gate[wires=2,disable auto height]{(B^\dagger)^{\otimes n}} \qw  & \meter{}\qw & \cw            \rstick{$\vec{x_2}$}\\
\lstick{\vdots} & & \vdots &  &    & \meter{}\qw & \cw            \rstick{$\vec{x_2}$}\\
                & & \vdots & & &\\
                \lstick{\vdots} & & \gate[wires=2][1.2cm]{B^{\otimes n}} \qw & \qw &  \vdots  & \rstick{\vdots}\\
\lstick{$\ket{0^n}$} & \qwbundle{n} &  & \gate{U_D} \qw & \qw & \qwbundle{n} \rstick{$\ket{\psi'}$}
            \end{quantikz}
        \end{equation}
        Where $\ket{\psi'} = U_D \cdot  Z(\vec{z_{D-1}}) X(\vec{x_{D-1}})  \cdot U_{D-1} \dots Z(\vec{z_2})  X(\vec{x_2})\cdot  U_2 \cdot  Z(\vec{z_1}) X(\vec{x_1}) \cdot U_1 \ket{\psi}$. Now if all $U_i$ are Clifford, we can commute the Pauli's through so that $\ket{\psi'} = Z(\vec{z'}) X(\vec{x'})  \cdot U_D \dots U_2 U_1 \ket{\psi}$. Furthermore, the map from $(x,z)$ to $(x', z')$ is linear over $\FF_2$, we denote this map as $M$. Let $a$ be the number of measurements so that $M : \FF_2^{a} \to \FF_2^{2n}$.

        We now show that it is possible to compute the corrections $My$ for any measurement outcome $y \in \FF_2^a$ with a classical circuit with only parity gates with unbounded fan-in and fan-out. Note that for each $y\in \FF_2^a$, $(My)_i = \textsf{Parity}_{M_i}(y)$ where $\textsf{Parity}_{M_i}$ is the function that takes a parity of the subset of bits that are $1$ in the $i$th row of $M$. Therefore we can compute $My$ in parallel with one layer of Parity gates. Then conditioned $My = (x', z')$, we make the Pauli X and Z corrections. So we've shown that $\IMparityrounds \leq \cliffdepth$.

        On the other hand, we show that a Clifford circuit can also simulate one round of intermediate measurements when the classical processing only uses parity gates. First, for each qubit that is measured, we copy the bit (in the standard basis) into an ancillary register set to $\ket{0}$ by applying a $\textsf{CNOT}$ gate. This ancillary qubit is left alone for the rest of the computation. By the principal of delayed measurement if we measure this ancilla register at the end of the computation it is equivalent to measuring the qubit in the standard basis before CNOTing it into the ancillary qubit. Then, since Clifford circuits can compute the Parity function $\textsf{parity}(x) = \sum_{i} x_i \mod 2$ and the Fanout gate, they can simulate any classical circuit with parity gates. Therefore,
        \begin{align}
            \IMparityrounds = \cliffdepth
        \end{align}
        Using \Cref{eq:level-bound-Cliff-QNC} we get the desired bound.
    \end{proof}
        More generally, if the computational power of the classical processing is unbounded then the number of rounds of measurements, \IMrounds provides a lower bound to the level of the hierarchy.
    \IMleqMHlevel*
    \begin{proof}
        This follows directly from \Cref{prop:IMpareqMHlevel} since allowing more classical computational power will not increase the required number of rounds of intermediate measurements.
    \end{proof}

    We highlight that in general, circuits with intermediate measurements implement quantum channels that are not necessarily unitary. On the other hand, circuits in the $\MH$ are unitary, which might make them a more familiar avenue for analyzing the complexity of circuits with intermediate measurements. We remind the reader that $\MH$ circuits are able to implement non-unitary channels with the use of ancillary qubits and partial trace.

\paragraph{Fanout-depth}
    The required number of alternations between $\Cliff$ and $\QNCZ$ circuits is related by a constant factor to the required fanout-depth in $\QNCZF$ circuits--- which consist of arbitrary two qubit gates, and the Fanout gate  $\ket{b, x_1, \dots, x_m} \to \ket{b, x_1 \oplus b, \dots, x_m \oplus b}$. For implementing a unitary $U$ let $\fanoutdepth$ be the required number of layers of parallel fanout gates in a $\QNCZF$ circuit, and $\cliffdepth$ the number of Clifford rounds required in a $\MH$ circuit. Then, we have
                \fanouteqMHlevel*
    \begin{proof}
        Since the Fanout gate is a Clifford operator, clearly the fanout depth is at least the number of required Clifford rounds. 
        
        We now show that any Clifford circuit can be implemented by $\QNCZF$ with fanout depth 4 to get the second inequality. We will use the same gate teleportation circuit (\Cref{eq:gate-teleportation-full-circuit}) in the proof of \Cref{prop:IMpareqMHlevel}, and just need to show that the classical computation for determining the Pauli corrections can actually be done with just 4 layers of Fanout gates. 

        Recall that for meaurement outcome $y \in \FF_2^a$, the Pauli corrections $(x',z')$ are equal to $M y$ for some linear map $M : \FF_2^a \to \FF_2^{2n}$.  First, we fan out each of the $a$ input qubits $2n$ times. For each $i \in [2n]$ and $y\in \FF_2^a$, $(My)_i = \textsf{Parity}_{M_i}(y)$ where $\textsf{Parity}_{M_i}$ is the function that takes a parity of the subset of qubits that are $1$ in the $i$th row of $M$. Therefore we can compute $My$ in parallel with one layer of Parity gates, which is equivalent to the fanout gate conjugated by single-qubit Hadamard gates. Conditioned on these registers we do the Pauli corrections with a layer of controlled $X$ gates and layer of cotrolled $Z$ gates. So far we used only 2 layers of fanout gates. We then uncompute, using $2$ more layers. Therefore,
        \begin{align}
            \cliffdepth \leq \fanoutdepth \leq 4 \cdot \cliffdepth.
        \end{align}
        Using \Cref{eq:level-bound-Cliff-QNC} we get the desired bound.        
    \end{proof}
    \noindent \Cref{prop:fanouteqMHlevel} straighforwardly implies \Cref{thm:intro:MHeqQACZF} which we restate and prove below.
    \MHeqQACZF*
        \begin{proof}
        Since fanout-depth and the level of the magic hierarchy are related by a constant factor (\Cref{prop:fanouteqMHlevel}), we have that 
        \[\QNCZF = \bigcup_{k\in \NN} \MH_k = \MH\]
        Furthermore \cite{hoyer2005quantum,takahashi2016collapse} showed that $\QNCZF = \QACZF$. Thus, by increasing the number of allowed alternations we interpolate between the well-studied $\QNCZ$ circuit model and the more powerful $\QACZF$ circuits.
    \end{proof}

\paragraph{Magic and T-depth}
    \emph{Magic} in quantum computation refers to the presence of non-Clifford gates in a circuit. A widely used measure of magic is the T-count—--the number of $T$ gates required in a Clifford+$T$ circuit. While Clifford circuits alone are classically simulable \cite{aaronson2004improved}, adding any non-Clifford gate enables universal quantum computation. Importantly, circuits with low $T$-count can still be efficiently simulated, albeit with an exponential overhead in the number of $T$ gates \cite{bravyi2016improved}.
    
    The number of alternations between $\QNCZ$ and $\Cliff$ can be viewed as another measure of magic that is gate-set agnostic. Since $\QNCZ$ circuits allow arbitrary two-qubit operations, the required number of $\QNCZ$ layers serves as a lower bound on the \emph{$T$-depth}—--the number of circuit layers that contain at least one $T$ gate in a Clifford+$T$ circuit \cite{amy2013meet}. 
    \proptdepthmhlevel*
    $T$-depth may be particularly relevant in fault-tolerant architectures, where Clifford gates can typically be implemented transversally, while non-Clifford gates like $T$ are more costly to realize.
    
    Whether the number of $\QNCZ$ layers and T-depth are equivalent remains an open question. The Solovay-Kitaev theorem allows for the approximation of arbitrary universal gate sets with $\log(1/\epsilon)$ overhead, and it is unclear whether this can be parallelized with the use of a ``free'' Clifford circuit.    

\paragraph{Connections and implications for state preparation lower bounds}    
Since $\MH$ and $\QNCZF$ circuits are unitary, if either of these circuits prepares a state, it is doing so with a specific unitary, so \Cref{prop:fanouteqMHlevel} tells us that the level of the magic hierarchy and the fanout depth required to prepare a particular quantum state are also equivalent up to a constant factor. In particular, this is true even when we require that the states are prepared \emph{cleanly}, meaning the ancilla qubits begin and end in the $\ket{0}$ state. So a lower bound against the $\MHlevel$ for preparing a state cleanly, implies a lower bound on the $\fanoutdepth$ for preparing that state cleanly. 

We note however, that a lower bound on the $\MHlevel$ of preparing a state $\ket{\psi}$ cleanly would not necessarily imply a lower bound on the needed $\IMparityrounds$ for preparing $\ket{\psi}$ since they are only equal when $\MHlevel$ is allowed partial trace and therefore \emph{non-clean} computation.

\subsection{Incomparability of \texorpdfstring{$\altCQ[1]$}{A1CQ} and \texorpdfstring{$\altQC[1]$}{A1QC}}
\label{ssec:CQneqQC}
As mentioned in the introduction, the order of the Clifford and $\QNCZ$ circuit matters, as we show below that in the first level, $\altCQ[1]$ and $\altQC[1]$ are incomparable.\footnote{We thank Henry Yuen and Ben Foxman for discussions leading to this observation.}

\thmCQneqQC*
\begin{proof}
    A single-qubit observable $O_i$ on the output of an $\altCQ[1]$ circuit $U$ corresponds to an evolved observable $U O_i U^\dagger$ on the input. Since $U$ is just a Clifford circuit followed by a $\QNCZ$ circuit this evolved observable can be written as a linear combination of $O(1)$ $n$-qubit Pauli operators. This is because after conjugating the observable $O_i$ on a single qubit by a $\QNCZ$ circuit it will only be supported only on the $O(1)$ qubits in the backwards lightcone, So its Pauli decomposition will have at most $4^{O(1)} = O(1)$ terms. Furthermore, Cliffords map Paulis to Paulis, so conjugating by the Clifford circuit will not increase the number of terms in this sum. On the other hand this is not in general true for circuits with the opposite ordering $\altQC[1]$, for example consider the Clifford circuit is the fanout operator which will map a single qubit $X$ on the control qubit to $X^{\otimes n}$, and if the $\QNCZ$ circuit is just a tensor product of the single qubit $T = \proj{0} + e^{i\pi/4} \proj{1}$ gate on each qubit, this maps the observable to $(TXT)^{\otimes n} \propto (X + Y)^{\otimes n}$ which has support on $2^n$ distinct Pauli strings. We use the same argument to analyze how a single qubit operator on the input state propagates through to the output to get the reverse direction.
\end{proof}

 We note that the above theorem is only shows incomparability in the context of implementing a unitary, it is an open question whether the same is true for state preparation. \section{Connections to classical circuit lower bounds}\label{sec:classical-conn}
In this section we show the implications of magic hierarchy lower bounds for \emph{classical circuit complexity}. The implications of state preparation lower bounds are summarized in \Cref{fig:barriers}.

Proving explicit lower bounds against $\TCZ$ circuits — constant-depth classical circuits with threshold gates — is a longstanding open problem in complexity theory. Despite extensive efforts, no nontrivial lower bounds are known even for depth-3 $\TCZ$ circuits. This challenge is tightly connected to the \emph{natural proofs barrier}, that rule out many natural proof techniques for proving lower bounds \cite{natural-proofs}. A more detailed discussion of the natural proofs barrier and more fine-grained implications between $\MH$ lower bounds and $\TCZ$ lower bounds is provided in \Cref{ssec:intro:barriers}.

Because $\QACZF$ circuits can simulate $\TCZ$ circuits, any nontrivial lower bound against $\MH = \QACZF$ for implenting a boolean function would imply a lower bound against $\TCZ$, which would be a breakthrough in complexity theory. In this section, we show that proving a lower bound against a fixed level of the magic hierarchy also implies a corresponding lower bound for a fixed depth $\TCZ$ circuit. 

\boolfunctczlb*

In the context of state preparation, the story is less obvious. It isn't immediately clear that a lower bound for state preparation would imply any classical circuit lower bound. Recently, Rosenthal's work  \cite{rosenthal2024efficient} shows that for each quantum state $\ket{\psi}$, there exists a Boolean function $f_\psi$ such that a $\QACZF$ circuit with query access to $f_\psi$ can prepare $\ket{\psi}$ to within exponentially small error\footnote{We refer the interested reader to Rosenthal's thesis \cite{rosenthal-thesis}, for a clear exposition of this indirect reduction.}. This implies that quantum circuit lower bounds for preparing some explicit state do, in fact, imply classical circuit lower bounds for an explicit Boolean function. We use this to deduce how lower bounds against certain levels of the magic hierarchy imply $\TCZ$ depth lower bounds.  

\statepreptczlb*

In the rest of this section, we prove these theorems by getting into the fine-grained details of how to simulate a $\TCZ$ circuit with $\MH$ (\Cref{ssec:MH-impl-classical}), and how to implement the algorithm in Rosenthal's one-query state synthesis algorithm (\Cref{ssec:MH-statesynthesis-impl}). For the latter, the general algorithm we use is implicit in \cite{rosenthal-thesis}, but without many implementation details necessary to determine the magic hierarchy depth. We thus construct and step through a detailed implementation to determine the required magic hierarchy depth.

\newcommand{\EX}[2][]{\textsf{EX}_{{#1}}^{#2}}
\newcommand{\THR}[2][]{\textsf{TH}_{{#1}}^{#2}}
\subsection{Magic hierarchy implementations of classical circuits}\label{ssec:MH-impl-classical}
We show how to implement various classical functions in the magic hierarchy to ultimately determine how many levels are needed to simulate a depth-$d$ $\TCZ$ circuit. We must first clarify the distinction between clean and nonclean computation. In this section we prove \Cref{thm:boolfuncbarrier} via \Cref{lem:threshold-level}.

\paragraph{Clean versus nonclean computation of boolean functions} We say that a unitary $U$ \emph{cleanly} computes $f$ if for each input $x\in \zo^n$ to $f$ and $b\in \zo$, the unitary maps $\ket{x, 0}\ket{0^a}$ to the state $\ket{x}\ket{f(x)}\ket{0^a}$ with $a$ ancillas that start and end in the all zeros state. \emph{Noncleanly} computing $f$ only requires that the output register measures to $\ket{f(x)}$, so the entire state can look like $\ket{\psi_x}\ket{f(x)}$ for any state $\ket{\psi_x}$ on $n+a$ qubits. Given a nonclean implementation of a function $f$, we can always convert it to a clean implementation by copying the output qubit onto an ancilla and then uncomputing. 

\begin{lemma}\label{lem:exact-complexity}
    For any $k \in \{0, 1, \dots, n\}$ the exact function $\EX[n]{k}(x_1, \dots, x_n) = \mathbf{1}\{\sum_i x_i = k\}$ can be computed noncleanly in $\altCQ[4]$ and cleanly in $\altCQ[6]$.
\end{lemma}
\begin{proof}
    Takahashi and Tani \cite{takahashi2016collapse} implement the $\EX[n]{k}$ function by first implementing a circuit that reduces this problem to the problem of computing $\textsf{OR}$ on only $\log(n)$ qubits. That is they map the $n$-qubit input state $\ket{x}$ to a state $\ket{\phi}$ on $\log(n)$ qubits such that $\EX[n]{k}(x) = 1$ if and only if the OR gate applied to $\ket{\phi}$ yields 1. The circuit doing this reduction was introduced in \cite{hoyer2005quantum} where it is easy to see it can be non-cleanly implemented in $\altCQ[2]$. Takahashi and Tani then show that $\textsf{OR}_m$ can be cleanly implemented by an $\exp(m)$-size a $\QACZF$ circuit, which then is a polynomial in $n$ when $m = \poly(n)$. We refer the reader to \cite{takahashi2016collapse} for more details, but the idea is that the circuit simply computes each of its $2^m$ Fourier characters (parity functions) in parallel, and collecting them into one phase by preparing a CAT state and applying CZ gates in parallel, and uncomputing the cat state to detect the phase. Overall, this circuit is in $\altCQ[2]$. Putting this together with the reduction, we see that we can compute $\textsf{OR}_n^t$ in polynomial size by first reducing to $\textsf{OR}_{\log n}$, then applying $\textsf{OR}_{\log n}$. Each of these steps is in $\altCQ[2]$, so we can compute $\EX[n]{k}$ \emph{noncleanly} in $\altCQ[4]$. If we wish to compute $\EX[n]{k}$ cleanly, we must uncompute the first circuit reducing the problem from $n$ to $\log(n)$ inputs, this circuit is in $\altCQ[2]$. So in total, cleanly computing the exact function can be done in $\altCQ[6]$.
\end{proof}

\begin{lemma}\label{lem:threshold-in-MH}
    For any $t \in \{0, 1, \dots, n\}$, the threshold function $\THR[n]{t}(x_1, \dots, x_n) = \mathbf{1}\{\sum_i x_i \geq t\}$ can be computed noncleanly in $\altCQ[4]$ and cleanly in $\altCQ[8]$. 
\end{lemma}
\begin{proof}
    As used in both \cite{hoyer2005quantum,takahashi2016collapse}, the threshold gate can be implemented by fanning out the input $n-t$ times and for each $k \in \{t, t+1, \dots, n\}$ computing $\EX[n]{k}$, then taking the parity of each of these $n-t$ outcomes. Implementing each of the $\EX[n]{k}$ gates noncleanly is in $\altCQ[4]$ so this computation noncleanly computes $\THR[n]{t}$ in a $\textsf{Fanout}$-$\altCQ[4]$-$\textsf{Parity}$ circuit, which is still a $\altCQ[4]$ circuit. To implement this cleanly, we just need to uncompute the layer of exact gates and the fanout gate, so this looks like: $\textsf{Fanout}$-$\altCQ[4]$-$\textsf{Parity}$-$\altCQ[4]$-$\textsf{Fanout}$, which is $\altCQ[8]$.
\end{proof}

\begin{lemma}\label{lem:threshold-level}
    For any function $f$ that can be implemented by a depth-$d$ $\TCZ$ circuit, $f$ can be implemented noncleanly by $\altCQ[4d]$ and cleanly by $\altCQ[8d]$.
\end{lemma}
\begin{proof}
    To simulate a $\TCZ$ circuit, for each threshold gate in the $\TCZ$ circuit, we can implement it with a quantum implementation of the threshold gate, and a fanout gate, since $\TCZ$ circuits allow for unbounded fanout. Thus each layer of threshold gates in the classical circuit will be translated into a layer of fanout gates followed by threshold gates. Since fanout is Clifford and the threshold gate can be noncleanly implemented in $\altCQ[4]$ (by \Cref{lem:threshold-in-MH}), this corresponds to a layer of $\altCQ[4]$. All together, this nonclean implementation of $f$ is in $\altCQ[4d]$. To implement the clean version of $f$, we just need to uncompute all gates but the last one touching the output, giving us a circuit in $\altCQ[8d]$.
\end{proof}

\subsection{State sythesis algorithm in the magic hierarchy} \label{ssec:MH-statesynthesis-impl}
In this section we prove the following theorem which is a restatement of our \Cref{thm:stateprepbarrier}.

\begin{theorem}
    Rosenthal's single-query unitary for preparing any state $\ket{\psi}$ to within $\exp(-\Omega(n))$ error in trace distance can be implemented by a $\MH$ circuit that first applies a layer of single qubit gates, makes a query to $f_{\psi}$, and then applies an $\altCQ[29]$ circuit.
\end{theorem}
\begin{proof}
    We will not re-prove the completeness of Rosenthals result, but we will step through the algorithm to show where the complexity is coming from, which will require us to get more into the details of parts of the algorithm than in \cite{rosenthal2024efficient,rosenthal-thesis}. We are using the version of the one-query state synthesis in Sections 2.2.4 and 2.3.1 from Rosenthal's thesis \cite{rosenthal-thesis} that makes use of \emph{hash states}. Rosenthal attributes the idea to use hash states to private communication with Fermi Ma. 
A hash state is an $n$-qubut state $\ket{\phi} = |S|^{-1/2}\sum_{x \in S} \sigma_x \ket{x}$ for some subset $S\sse [n]$ that has size $|S| = 2^k$ a power of two, and phases $\sigma_x \in \{-1,1\}$. Furthermore, there exists a linear transformation $A : \FF_2^n \to \FF_2^k$ that is one-to-one on $S$. Rosenthal shows that for any state $\ket{\psi}$, there is a linear combination of $T = \poly(n)$ subset states $\ket{\Phi} = \sum_{j= 0}^{T-1} \sqrt{\beta^j} \ket{\phi_j}$ for some fixed $\beta \in [0,1]$ such that $\ket{\Phi}$ is $O(1/\exp(n))$-close to $\ket{\psi}$. We prepare a related state as follows:
    \begin{enumerate}
        \item \textbf{[Single-qubit gates]} Prepare the $t = \log T$-qubit state $\ket{\sigma} = \sqrt{\frac{1-\beta}{1 - \beta^T}} \sum_{j = 0}^{T-1} \sqrt{\beta^j}\ket{j}$ on register $\textsf{A}$. This can be done with a layer $L$ of single qubit gates since $\ket{\sigma}$ is a product state 
        \[
            \ket{\sigma} = \sqrt{\frac{1-\beta}{1 - \beta^T}} \bigotimes_{j = 0}^{t-1} \pbra{\ket{0} + \beta^{2^{j-1}} \ket{1}}_\textsf{A}
        \]
        \item $\boldsymbol{[\mathsf{Query}}{\boldsymbol{+}}\boldsymbol{\altCQ[11]]}$ \label{step:constrolled-hashstate} Conditioned on register $\textsf{A}$ containing value $j$: prepare the state $\ket{\phi_j}$ on register~$\textsf{B}$.
            \[\sqrt{\frac{1-\beta}{1 - \beta^T}} \sum_{j = 0}^{T} \sqrt{\beta^j} \ket{j}_\textsf{A} \ket{\phi_j}_\textsf{B}\]
        To do this, Let $A_j : \FF_2^n \to \FF_2^{k_j}$ be the linear map that is one to one on the subset $S_j$ for $\ket{\phi_j} \propto \sum_{x \in S_j} \sigma_x \ket{x}$. We first prepare the uniform superposition over all $y \in \zo^{k_j}$ on the first $k_j$ qubits by applying a layer of Hadamard (this can be done in parallel with applying the gate $L$ in the previous step). We then make a query on $\ket{y}_\textsf{B}$ to get the unique $x\in \zo^n$ such that $A_j x = y$. Simultaneously in parallel we query to get the phase $\sigma_x$. At this point the state is proportional to
        \[\sum_{j = 0}^{T} \sqrt{\beta^j} \ket{j}_\textsf{A} \ot \sum_{x\in S_j}\sigma_{x}\ket{x}_\textsf{B}\ket{A_jx}_\textsf{C}\]
        We then uncompute $y = A_j x$ by evaluating $A_j$ on $x$. We show in \Cref{lem:cA} below that this can be done in $\altCQ[11]$. This prepares the desired state.
        \item $\boldsymbol{[\QNCZ]}$ Apply $L^\dagger$ to register $\textsf{A}$. Rosenthal shows that if we then measure register $\textsf{A}$ and the outcome is all zeros, then the post-measurement state on $\textsf{B}$ is exponentially-close to $\ket{\psi}$. The problem now is that it's only with probability $O(1/n)$ that we get $0$ as our measurement outcome. 
    \end{enumerate}
    To boost the probablity of success, Rosenthal does the above many times in parallel, on registers $\textsf{A}_1, \textsf{B}_1, \dots, \textsf{A}_s, \textsf{B}_s$ for some $s = \poly(n)$. We introduce another $n$-qubit register $\textsf{D}$ initialized to the zero state. Then, we evaluate the unitary that does the following: if any of the $\textsf{A}_i$ registers contain the all zeros state, for the first such $i$, swap the state $\textsf{B}_i$ with the state in the $\textsf{D}$ register. If none of the $\textsf{A}_i$ states are all zeros then we do nothing. We show below in \Cref{lem:multi-swap-complexity} that this unitary can be implemented in $\altCQ[18]$. Overall this algorithm uses single-qubit gates followed by a query and then $\altCQ[11]$-$\QNCZ$-$\altCQ[18]$ which is $\altCQ[29]$.
    \end{proof}

    \begin{lemma}\label{lem:cA}
        For any $t = O(\log n)$, and linear maps $\cbra{A_j : \FF_2^n \to \FF_2^n :  \  j \in [2^t]}$, the following unitary can be implemented by a $\altCQ[11]$ circuit. The state $\ket{j}_\textsf{A}\ket{x}_\textsf{B}\ket{y}_\textsf{C}$ for $\textsf{A}$ a $t$-qubit register, and $\textsf{B}, \textsf{C}$ as $n$-qubit registers, is mapped to $\ket{j}_\textsf{A} \ket{x}_\textsf{A} \ket{A_jx \oplus y}_{\textsf{C}}$.
    \end{lemma}
    \begin{proof}
        The existence of this algorithm is implicit in \cite{rosenthal-thesis}, but details of implementation and its fine-grained complexity are missing. For completeness, we provide a detailed construction. 
        \begin{enumerate}
            \item \label{step:cA:unary} $\boldsymbol{[\altCQ[4]]}$ We first compute a unary encoding of $\ket{j}_{\textsf{A}}$, meaning preparing the state $\bigotimes_{i=1}^{2^t} \ket{\delta_{i,j}}$, where $\delta_{i,j} = \mathbf{1}\{i = j\}$. To do this we first fanout $j$ many times. In parallel, for each possible value of $k \in [2^t]$, we check if $k=j$. To do this, on the $k$th copy of $\ket{j}$, we flip each of the bits $i\in [t]$ of $j$ such tht $k_i=0$. So now $j = k$ iff the resulting state is $\ket{1^t}$. Thus we apply an AND gate to get a variable that indicates if $k=j$. This process used a layer of fanout, then a layer of Pauli-X gates, then an AND gate. Since computing AND noncleanly is in $\altCQ[4]$ \Cref{lem:exact-complexity} this process noncleanly computes the unary encoding of $\ket{j}$ in $\Cliff$-$\altCQ[4]$ which is $\altCQ[4]$.
        \[\ket{j}_\textsf{A} \ket{x}_\textsf{B} \ket{y}_\textsf{C} \bigotimes_{i\in [2^t]} \ket{\delta_{i,j}} \ket{\text{junk}}\] 
            \item \label{step:cA:fanout-x-ind} {$\boldsymbol{[\Cliff]}$} Fanout input $\ket{x}_\textsf{A}$ and $\ket{\delta_{i,j}}$ many times.
            \item \label{step:cA:Ai-parallel} {$\boldsymbol{[\Cliff]}$} For each $i$ in parallel, compute $\ket{A_i x}$. This can be done with a layer of parity gates since $(A_i x)_k$ is just the parity of a subset of bits of $x$ indicated by the $k$th row of $A_i$.
            \[
               \ket{j}_\textsf{A} \ket{x}_\textsf{B} \ket{y}_\textsf{C}  \ket{\text{junk}} \bigotimes_{i \in [2^t]} \ket{x} \ket{A_i x} \ket{\delta_{i,j}}^{\otimes n}
            \]
            \item \label{step:cA:ind-and-map} {$\boldsymbol{[\QNCZ]}$} For each $i$ in parallel, use $n$ parallel Toffoli gates to compute each $\ket{\delta_{i,j} \wedge A_i x}$ into a register $\textsf{A}'_i$.
            \[
\ket{j}_\textsf{A} \ket{x}_\textsf{B} \ket{y}_\textsf{C}  \ket{\text{junk}}  \ket{x}^{\otimes 2^t} \ot \ket{A_j x}_{\textsf{A}'_i} \bigotimes_{i \neq j} \ket{0}_{\textsf{A}'_i}\bigotimes_{i \in [2^t]} \ket{\delta_{i,j}}^{\otimes n}
            \]
            \item {$\boldsymbol{[\Cliff]}$} Compute $n$ parity gates in parallel between the qubits of the $\textsf{A}'_i$'s into the output register $\ket{y}_\textsf{C}$.
            \[
                \ket{j}_\textsf{A} \ket{x}_\textsf{B} \ket{A_j x \oplus y}_\textsf{C}  \ket{\text{junk}}  \ket{x}^{\otimes 2^t} \ot \ket{A_j x}_{\textsf{A}'_i} \bigotimes_{i \neq j} \ket{0}_{\textsf{A}'_i}\bigotimes_{i \in [2^t]} \ket{\delta_{i,j}}^{\otimes n}
            \]
            \item {$\boldsymbol{[\altCQ[5]]}$} Now uncompute \Cref{step:cA:unary,step:cA:fanout-x-ind,step:cA:Ai-parallel,step:cA:ind-and-map}. Together they are in $\altCQ[5]$, so their inverse is in $\altQC[5]$.
        \end{enumerate}
        All together this is $\altCQ[5]$-$\QNCZ$-$\altQC[5]$ which is $\altCQ[11]$.
        
    \end{proof}

    \begin{lemma}\label{lem:multi-swap-complexity}
        Let $U$ be a unitary on $n$-qubit input registers $\textsf{A}_1, \textsf{B}_1, \dots, \textsf{A}_s, \textsf{B}_s, \textsf{D}$ with the following  behavior in the computational basis: if any of the $\textsf{A}_i$ contain the all zeros state, then for the first such $i$, swap registers $\textsf{B}_i$ and $\textsf{D}$. Then $U$ can be implemented by a polynomial-sized $\altCQ[18]$ circuit. 
    \end{lemma}

    \begin{proof}
        This algorithm is adapted from the proof of Lemma A.2.1 in \cite{rosenthal-thesis}, though we provide a more detailed construction so we can give a more fine-grained analysis. 
    Consider the input state $\ket{a_1}_{\textsf{A}_1} \dots \ket{a_s}_{\textsf{A}_s} \ket{x_1}_{\textsf{B}_1} \dots \ket{x_s}_{\textsf{B}_s}$
    \begin{enumerate}
        \item \label{step:indicate-swap}{$\boldsymbol{[\altCQ[8]]}$} For each $j \in [s]$, we prepare an indicator variable $\delta_j$ that is equal to 1 if $\textsf{B}_j$ is the register we are supposed to swap  with $\textsf{D}$. The value of $\delta_j$ is $\EX{0}\pbra{\EX{0}(a_1), \EX{0}(a_2), \dots, \EX{0}(a_{j-1}), 1- \EX{0}(a_j)}$. Note that at most $1$ of the indicators will be set to 1. So we can (noncleanly) compute this with fanout gates and two layers of exact gates, which each are in $\altCQ[4]$, so in total this can be done in $\mathbf{\altCQ[8]}$.
        \item {$\boldsymbol{[\Cliff]}$} Fanout the indicator variables many times. For each $j \in [s]$ in parallel compute $\delta_j \wedge x_j$ into a new ancilla register we label $\textsf{C}_j$. This can be done using Toffoli gates in parallel, which is in $\mathbf{\QNCZ}$:
        \[
            \pbra{
                \bigotimes_{j=1}^s \ket{x_j}_{\textsf{A}_j}\ket{\delta_j}^{\otimes \poly(n)} \ket{\delta_j \wedge x_j}_{\textsf{C}_j}
            }
            \otimes \ket{0}_\textsf{D}
        \]
        \item {$\boldsymbol{[\Cliff]}$}  In parallel, compute the parity of the qubits between all the $\textsf{C}_j$ registers into register $\textsf{D}$. Since at most one of the $\textsf{C}_j$ will be nonzero, this parity implements an OR for us. Suppose $k\in [s]$ is the index that should be swapped: $\delta_k = 1$. This operation will copy $\ket{x_k}$ into $\textsf{B}$. We then fanout $\ket{x_k}_{\textsf{B}}$ many times, resulting in the state.
         \[
            \pbra{
                \bigotimes_{j=1}^s  \ket{x_j}_{\textsf{A}_j} \ket{\delta_j}^{\ot \poly(n)} \ket{\delta_j \wedge x_j}_{\textsf{C}_j}
            }
            \otimes \ket{x_k}^{\otimes \poly(n)}_\textsf{D}.
        \]
        \item {$\boldsymbol{[\QNCZ]}$} Uncompute the $\textsf{C}_j$ registers in parallel by computing $x_k \wedge \delta_j$ into $\textsf{C}_i$ by applying \textbf{Toffoli} gates in parallel controlled on a copy of $\ket{\delta_j}$ and $\ket{x_k}_\textsf{D}$, and targeted on $\textsf{C}_i$. 
        \[
            \pbra{
                \bigotimes_{j=1}^s  \ket{x_j}_{\textsf{A}_j} \ket{\delta_j}^{\ot \poly(n)} 
            }
            \otimes \ket{x_k}^{\otimes \poly(n)}_\textsf{D}.
        \]
        \item {$\boldsymbol{[\QNCZ]}$} We next need to replace $\ket{x_k}_{\textsf{A}_k}$ with $\ket{0}_{\textsf{A}_k}$. We still have many copies of $x_k$, so in parallel for each $j$: apply a Toffoli from a copy of $\ket{x_k}$ and a copy of $\ket{\delta_j}$ to compute $x_k \wedge \delta_j$ and XOR it into $\textsf{A}_j$. This will have the desired effect:
        \[
            \ket{0}_{\textsf{A}_{k}}
            \pbra{
                    \bigotimes_{j\neq k}  \ket{x_j}_{\textsf{A}_j} \ket{\delta_j}^{\ot \poly(n)} 
            }
            \otimes \ket{x_k}^{\otimes \poly(n)}_\textsf{D}.
        \]
        \item {$\boldsymbol{[\altCQ[8]]}$} Now we just need to uncompute the remaining intemediate registers. Uncompute the circuit preparing the indicator registers by uncomputing \Cref{step:indicate-swap} which is $\altCQ[8]$, and we fanout on $\textsf{D}$ to reduce back down to one copy of $\ket{x_k}$. We note that if there was no such $k$ with $\delta_k = 1$, then the above computation would have the state $\ket{0}_\textsf{D}$ in the $\textsf{D}$ register, the same as the initial state as desired.
    \end{enumerate}
    In total this is in $\altCQ[8]$-$\QNCZ$-$\altCQ[8]$ which is $\altCQ[18]$.
\end{proof}

\subsection{Future directions on connections to classical complexity}
It is plausible that the relationship between lower bounds on $\MH$ level and $\TCZ$ depth in \Cref{thm:boolfuncbarrier,thm:stateprepbarrier} are not tight. Tightening this bound to illuminate what level of the magic hierarchy we start proving, say depth-4 $\TCZ$ lower bounds is an interesting further direction. 

Unlike for state preparation, in the context of implementing a unitary, there are no known reductions to classical circuits, and thus proving lower bounds for implementing explicit unitaries do not immediately imply any classical circuit lower bounds.

Furthermore, we note that the reduction in \cite{rosenthal2024efficient} requires the use of ancillary qubits. Thus there are also no known barriers against proving lower bounds by circuits that are not allowed ancillary qubits.

\section{Cat state gluing lower bounds}
\begin{question}
    Given a CAT state on $2n$ qubits $\frac{1}{\sqrt{2}}(\ket{0^{2n}} + \ket{1^{2n}})$ what is the depth required to convert it into a tensor product of two $n$-qubit CAT states?
\end{question}
Below we show that this requires $\Omega(\log n)$ depth, which is the same as if one just ``un-did'' half of the $\CATn[2n]$ state with $n$ $\CNOT$ gates in parallel, sending $\CATn[2n] \to \CATn[n] \ot \proj{0^n}$, and then just preparing the $\CATn[n]$ state on the second register. We can consider this question with \emph{biased} CAT states.
\begin{definition}[Biased CAT state]
    $\ket{\bCATn[n]{\alpha}} := \sqrt{\alpha} \ket{0^n} + \sqrt{1 - \alpha} \ket{1^n}$.
\end{definition}
Below we show that mapping between $ \ket{\bCATn[2n]{\alpha}}$ and $\ket{\bCATn[n]{\beta}} \ot \ket{\bCATn[n]{\beta}}$ requires depth $\Omega\pbra{\log (n)}$ if $H(\alpha) \neq 2H(\beta)$.
\begin{theorem}\label{thm:cat-state-gluing-lb}
    Suppose $U\in \unitaries[2n]$ is a $2n$ qubit quantum circuit for $n\geq 2$ such that 
    \begin{align}
        U \proj{\bCATn[2n]{\alpha}} U^\dagger = \proj{\bCATn[n]{\beta}} \ot \proj{\bCATn[n]{\beta}}
    \end{align}
    For some $\alpha, \beta \in [0,1]$. If $H(\alpha) \neq 2 H(\beta)$, then $U$ requires depth $\Omega(\log n)$. Furthermore, this statement can be made robust as follows. Consider $V\in \unitaries[2n]$ such that
    \begin{align}
        V \proj{\bCATn[2n]{\alpha}} V^\dagger = \psi && \tnorm{\psi - \proj{\bCATn[n]{\beta}} \ot \proj{\bCATn[n]{\beta}}}\leq \eps.
    \end{align}
    Then $\depth(V) \geq \Omega\pbra{\log\min\cbra{n, \frac{1}{\eps}\cdot \abs{2 H(\beta) - H(\alpha)}^{\ln 4}}}$.
\end{theorem}
\begin{proof}
    Let $S = \{1, n+1\}$, a subset containing the labels of the first qubit of each of the two $\ket{\bCATn[n]{\beta}}$ states. Suppose for the sake of contradiction that $d < \frac{1}{2}\log(n/2)$. Since $|\dlc{S}| \leq 2^{2d} |S| = 2^{2d + 1}$ we have that $|\dlc{S}| < n$. Therefore, there exists a pair of qubits $i,j \in [2n]$, one from each half of the qubits $i\in \{1, \dots, n\}, j \in \{n+1, \dots, 2n\}$ such that  $i, j \notin \dlc{S}$. We let $T = \{i,j\}$, and have that $\blc(S) \cap \blc(T) = \emptyset$.
    \begin{align*}
        \MI[\bCAT{\beta}\ot \bCAT{\beta}]{S}{T} &= \MI[\bCAT{\beta}]{1}{i} + \MI[\bCAT{\beta}]{1}{j- n}\\
        &= 2H(\beta)
    \end{align*}
    On the other hand, we have that $|\blc(S) \cup \blc(T)| \leq 2^{2d} \cdot 4 < 2n$. Therefore, again we have that 
    \begin{align*}
        \MI[{\bCATn[2n]{\alpha}}]{\blc(S)}{\blc(T)}  = H(\alpha).
    \end{align*}
    By \Cref{lem:mi-lightcones} we have that $\MI[\bCAT{\beta}\ot \bCAT{\beta}]{S}{T} \leq \MI[{\bCATn[2n]{\alpha}}]{\blc(S)}{\blc(T)}$. So we have that $2 H(\beta) \leq H(\alpha)$. 

    We can make the same argument to show that $2H(\beta) \geq H(\alpha)$. As above, if $d < \frac{1}{2}\log(n/2)$ we can find a $a_1, b_1 \in [n]$ and $a_2, b_2 \in \{n+1, \dots, 2n\}$ with nonintersecting lightcones, $\flc(a_1, a_2) \cap \flc(b_1, b_2) = \emptyset$, so that $\MI[\bCAT{\alpha}]{a_1,a_2}{b_1,b_2}\leq \MI[\bCAT{\beta}\ot \bCAT{\beta}]{\flc(a_1a_2)}{\flc(b_1b_2)}$. Since $\{a_1, a_2\} \sse \flc(a_1a_2)$ we have that $\flc(a_1a_2)$ contains some qubits in both registers, likewise for $\flc(b_1b_2)$. So for the same reasoning as above, we have that $\MI[\bCAT{\beta}\ot \bCAT{\beta}]{\flc(a_1a_2)}{\flc(b_1b_2)} = 2H(\beta)$. But we also have that $\MI[\bCAT{\alpha}]{a_1a_2}{b_1b_2} = H(\alpha)$. Therefore $H(\alpha) \leq 2 H(\beta)$. Combined with our $H(\alpha) \geq 2 H(\beta)$ shown above, it follows that $H(\alpha) = H(\beta)$.

    Now suppose that instead of preparing $\bCATn[n]{\beta}\ot \bCATn[n]{\beta}$ we prepare a state $\psi$ that is $\eps$-close 
    \begin{align*}
        \psi = V \proj{\bCATn[2n]{\alpha}} V^\dagger && \tnorm{\psi - \proj{\bCATn[n]{\beta}}\ot \proj{\bCATn[n]{\beta}}} \leq \eps.
    \end{align*}
    Suppose that $\depth(V) < \frac{1}{2}\log(n/2)$. As shown above there exists $S, T\sse [2n]$ with $|S|, |T| = 2$ such that $\blc(S) \cap \blc(T) = \emptyset$, $\MI[\bCAT{\beta}\ot \bCAT{\beta}]{S}{T} = 2H(\beta)$ and $\MI[\bCAT{\alpha}]{\blc(S)}{\blc(T)} = H(\alpha)$. But now since $V$ maps $\bCAT{\alpha}$ to $\psi$ instead of $\bCAT{\beta}\ot \bCAT{\beta}$, \Cref{lem:mi-lightcones} now gives us that
    \begin{align}
        H(\alpha) = \MI[\bCAT{\alpha}]{\blc(S)}{\blc(T)} &\leq \MI[\psi]{S}{T}\\
        &\leq \MI[\bCAT{\beta}\ot \bCAT{\beta}]{S}{T} + \abs{\MI[\bCAT{\beta}\ot \bCAT{\beta}]{S}{T} - \MI[\psi]{S}{T}}\\
        &\leq 2H(\beta) + (2 \eps (|S| + |T|)  + H(\eps))\\
        &= 2H(\beta) + (8\eps + H(\eps)).
    \end{align}
    In the second inequality we used the triangle inequality, and in the third we used the Fannes-Audenaert inequality. 
    Since $H(\eps) \leq (4 \eps(1-\eps))^{1/\ln(4)}\leq e \cdot \eps^{1/\ln (4)}$, and $\eps \leq \eps^{1/\ln(4)}$ we have that $8 \eps + H(\eps) \leq (8+e) \eps^{1/\ln(4)}$. Therefore, if $H(\alpha) > 2H(\beta)$ then
    \begin{align}
         (8 + e) \eps^{1/\ln(4)} &\geq H(\alpha) - 2H(\beta)\\
        \eps &\geq \pbra{\frac{H(\alpha) - 2H(\beta)}{8+e}}^{\ln(4)} > 0.037 \cdot \pbra{H(\alpha) - 2 H(\beta)}^{\ln(4)}.
    \end{align}
    This bound on $\eps$ followed only from the assumption that $d < \frac{1}{2}\log(n/2)$ and $H(\alpha) > 2 H(\beta)$. Therefore, if $H(\alpha) > 2 H(\beta)$ and $\eps \leq 0.037 \cdot \pbra{H(\alpha) - 2 H(\beta)}^{\ln(4)}$ then $d\geq \Omega(\log n)$, as claimed.
\end{proof}
Furthermore, we show that if $H(\alpha) > 2 H(\beta)$ the lower bound can be further improved.
\begin{theorem}\label{thm:cat-state-gluing-eps-indep}
    Furthermore, if $H(\alpha) > 2 H(\beta)$ and our approximation is $\eps \leq 0.037 \cdot \pbra{H(\alpha) - 2H(\beta)}^{\ln 4}$ then any unitary $V\in \unitaries[2n]$ such that
    \begin{align}
        V \proj{\bCATn[2n]{\alpha}} V^\dagger = \psi && \tnorm{\psi - \proj{\bCATn[n]{\beta}} \ot \proj{\bCATn[n]{\beta}}}\leq \eps
    \end{align}
    has $\depth(V) \geq \Omega\pbra{\log n}$.
\end{theorem}
\begin{proof}
    As shown in the proof of \Cref{thm:cat-state-gluing-lb} above there are $S', T' \sse [2n]$ with $|S'|, |T'| = 2$  such that $\flc(S') \cap \flc(T') = \emptyset$,  $\MI[\bCAT{\alpha}]{S'}{T'} = H(\alpha)$, and $\MI[\bCAT{\beta}\ot \bCAT{\beta}]{\flc(S)}{\flc(T)} = 2H(\beta)$. Therefore, by \Cref{lem:mi-lightcones} we have
    \begin{align}
        H(\alpha) = \MI[\bCAT{\alpha}]{S'}{T'} &\geq \MI[\psi]{\flc(S')}{\flc(T')}\\
        &\geq \MI[\bCAT{\beta}^{\ot 2}]{\flc(S')}{\flc(T')} - \abs{\MI[\bCAT{\beta}^{\ot 2}]{\flc(S')}{\flc(T')} - \MI[\psi]{\flc(S')}{\flc(T')}}\\
        &\geq 2H(\beta) - (2 \eps |\flc('S)| \cdot |\flc(T')|  + H(\eps))\\
        &= 2H(\beta) - (8 \eps \cdot 2^d  + H(\eps)).
    \end{align}
    Again using $H(\eps) \leq e \cdot \eps^{1/\ln(4)}$ and $\eps \leq \eps^{1/\ln(4)}$, we have $\abs{2H(\beta) - H(\alpha)} \leq(8 \cdot 2^d + e) \eps^{1/\ln(4)}$. Therefore, 
    \begin{align}
        (8 \cdot 2^d + e) \eps^{1/\ln(4)}&\geq  \abs{2H(\beta) - H(\alpha)} \\
        2^d  &\geq  \frac{1}{8}\pbra{\abs{2H(\beta) - H(\alpha)}\eps^{-1/\ln(4)} - e}\\
        d &\geq \Omega\pbra{\log\pbra{\frac{1 }{ \eps} \cdot \abs{2H(\beta) - H(\alpha)}^{\ln 4}}}
    \end{align}
\end{proof}

This begs the question, what happens when the entropies of the biased CAT states are balanced, is it then possible to ``glue'' the biased CAT states?
\begin{openproblem}
    Given an $\alpha$-biased CAT state on $2n$ qubits $ \sqrt{\alpha}\ket{0^{2n}} + \sqrt{1-\alpha}\ket{1^{2n}}$ what is the depth required to convert it into a tensor product of two $n$-qubit $\beta$-biased CAT states? In particular when $H(\alpha) = 2 H(\beta)$.
\end{openproblem}
Furthermore, given the ability to glue biased CAT states, is this helpful for any interesting task?
\begin{openproblem}
    Given the black box ability to glue biased CAT states with $H(\alpha) = 2 H(\beta)$, would it help you do something else? Such as prepare a CAT state, a Nekomata state, or compute Parity? 
\end{openproblem}

\bibliography{refs.bib}
\bibliographystyle{alpha}

\end{document}